\documentclass[11pt, letter]{article}%{scrartcl}

\usepackage{amsmath,amsfonts,amssymb,amsthm}
\usepackage{xspace}
\usepackage{paralist}
\usepackage{graphicx}
\usepackage{array}
\usepackage{fullpage}

\usepackage[ruled,vlined,linesnumbered]{algorithm2e}
\SetFuncSty{textsc}
\SetDataSty{textit}
\SetCommentSty{textrm}
\SetKwFunction{Extract}{ExtractFirstDigit}
\SetKwData{Forward}{forward}
\SetKw{True}{true}
\newcommand{\etal}{\textit{e{}t~a{}l.}\xspace}

\DeclareMathOperator{\cost}{cost}
\DeclareMathOperator{\opt}{opt}

\DeclareMathOperator{\Exp}{E}
\DeclareMathOperator*{\argmin}{arg\,min}
\DeclareMathOperator*{\argmax}{arg\,max}

\newtheorem{theorem}{Theorem}[section]

\newtheorem{definition}[theorem]{Definition}
\newtheorem{lemma}[theorem]{Lemma}
\newtheorem{observation}[theorem]{Observation}
\newtheorem{corollary}[theorem]{Corollary}

\newtheorem{claim}[theorem]{Claim}
\newtheorem{case}[]{Case}
\newtheorem{assumption}[]{Assumption}
\newtheorem{property}[]{Property}

\newcommand{\seclab}[1]{\label{sec:#1}}
\newcommand{\secref}[1]{Section~\ref{sec:#1}}
\newcommand{\subsecref}[1]{Subsection~\ref{sec:#1}}

\newcommand{\thmlab}[1]{{\label{theo:#1}}}
\newcommand{\thmref}[1]{Theorem~\ref{theo:#1}}

\newcommand{\lemlab}[1]{\label{lemma:#1}}
\newcommand{\lemref}[1]{Lemma~\ref{lemma:#1}}

\newcommand{\figlab}[1]{\label{fig:#1}}
\newcommand{\figref}[1]{Figure~\ref{fig:#1}}

\newcommand{\tablab}[1]{\label{tab:#1}}
\newcommand{\tabref}[1]{Table~\ref{tab:#1}}

\newcommand{\deflab}[1]{\label{def:#1}}
\newcommand{\defref}[1]{Definition~\ref{def:#1}}

\newcommand{\alglab}[1]{\label{alg:#1}}
\newcommand{\algref}[1]{Algorithm~\ref{alg:#1}}

\newcommand{\obslab}[1]{\label{obs:#1}}
\newcommand{\obsref}[1]{Observation~\ref{obs:#1}}

\newcommand{\corlab}[1]{\label{cor:#1}}
\newcommand{\corref}[1]{Corollary~\ref{cor:#1}}

\newcommand{\caselab}[1]{\label{case:#1}}
\newcommand{\caseref}[1]{Case~\ref{case:#1}}

\newcommand{\claimlab}[1]{\label{claim:#1}}
\newcommand{\claimref}[1]{Claim~\ref{claim:#1}}

\newcommand{\asslab}[1]{\label{ass:#1}}
\newcommand{\assref}[1]{Assumption~\ref{ass:#1}}

\newcommand{\proplab}[1]{\label{prop:#1}}
\newcommand{\propref}[1]{Property~\ref{prop:#1}}

\newcommand{\NP}{\textbf{NP}} % nach Belieben anpassen
\newcommand{\Frechet}{Fr\'echet\xspace}

\newcommand{\atgen}{\symbol{'100}}
\providecommand{\eps}{{\varepsilon}}%
\newcommand{\Astop}{\overline{a}}

\providecommand{\ceil}[1]{\left\lceil {#1} \right\rceil}
\providecommand{\floor}[1]{\left\lfloor {#1} \right\rfloor}
\providecommand{\pth}[2][\!]{#1\left({#2}\right)}
\providecommand{\brc}[1]{\left\{ {#1} \right\}}
\providecommand{\cbrc}[1]{\left\langle{#1}\right\rangle}

\providecommand{\cardin}[1]{\left\lvert#1\right\rvert}
\newcommand{\pbrc}[1]{\left[ {#1} \right]}
\renewcommand{\Re}{{\rm I\!\hspace{-0.025em} R}}
\newcommand{\Na}{{\rm I\!\hspace{-0.025em} N}}

\newcommand{\nrClusters}{\ensuremath{k}}
\newcommand{\lenClusters}{\ensuremath{\ell}}
\newcommand{\constA}{\ensuremath{\alpha}}
\newcommand{\constB}{\ensuremath{\beta}}

\newcommand{\initialCost}{\ensuremath{D}}
\newcommand{\costOpt}{\ensuremath{\delta^{*}}}

\newcommand{\trajectory}[2]{\ensuremath{{#1}_{#2}}}
\newcommand{\trajectorySimp}[2]{\ensuremath{\widehat{{#1}_{#2}}}}

\newcommand{\inputSym}{\ensuremath{\tau}}
\newcommand{\centerSym}{\ensuremath{c}}

\newcommand{\inputTraj}[1]{\trajectory{\inputSym}{#1}}

\newcommand{\inputTrajSimp}[1]{\trajectorySimp{\inputSym}{#1}}

\newcommand{\centerTraj}[1]{\trajectory{\centerSym}{#1}}

\newcommand{\VtxSet}{\ensuremath{\mathcal{V}}}

\newcommand{\distFr}[2]{\ensuremath{d_F\pth{#1,#2}}}
\newcommand{\distFrConst}{\ensuremath{d_F}}

\newcommand{\setSubC}[4]{\ensuremath{\brc{ #1(#2) ~|~ #2 \in [#3,#4]} }}

\newcommand{\range}[2]{\ensuremath{[#1]_{#2}}}
\newcommand{\minSubC}[4]{\ensuremath{\min( #1[#3,#4])}}
\newcommand{\maxSubC}[4]{\ensuremath{\max( #1[#3,#4])}}

\begin{document}

%\linenumbers

\title{Clustering time series under the \Frechet distance
\thanks{The conference version of this paper will be published at 26th ACM-SIAM Symposium on Discrete Algorithms (SODA) 2016.}
}

\author{%
   Anne Driemel%
   \thanks{Department of Mathematics and Computer Science, TU Eindhoven, The
   Netherlands; 
      \texttt{a.driemel}\hspace{0cm}\texttt{\atgen{}tue.nl}. 
      Work on this paper was partially funded by NWO STW project ``Context
      awareness in predictive analytics'' and NWO Veni project ``Clustering time
      series and trajectories (10019853)''. } %
   \and%
   Amer Krivo\v{s}ija%
   \thanks{Department of Computer Science, TU Dortmund, Germany; 
      \texttt{amer.krivosija}\hspace{0cm}\texttt{\atgen{}tu-dortmund.de}. Work on this paper has been partly supported by DFG within the Collaborative Research Center SFB 876 ``Providing Information by Resource-Constrained Analysis'', project A2.} %
   \and %
   Christian Sohler%
   \thanks{Department of Computer Science, TU Dortmund, Germany; 
      \texttt{christian.sohler}\hspace{0cm}\texttt{\atgen{}tu-dortmund.de}. Work on this paper has been partly supported by DFG within the Collaborative Research Center SFB 876 ``Providing Information by Resource-Constrained Analysis'', project A2.} %
}

\date{\today}

\maketitle

\begin{abstract}
The \Frechet{} distance is a popular distance measure for curves. We study the problem of clustering time series under the \Frechet{} distance.  In particular, we give $(1+\eps)$-approximation algorithms for variations of the following problem with parameters $\nrClusters$ and $\lenClusters$.  Given $n$ univariate time series $P$, each of complexity at most $m$, we find $\nrClusters$ time series, not necessarily from $P$, which we call \emph{cluster centers} and which each have complexity at most $\lenClusters$, such that (a) the maximum distance of an element of $P$ to its nearest cluster center or (b) the sum of these distances is minimized.  Our algorithms have running time near-linear in the input size for constant $\eps, \nrClusters$ and $\lenClusters$.  To the best of our knowledge, our algorithms are the first clustering algorithms for the \Frechet{} distance which achieve an approximation factor of $(1+\eps)$ or better. 

Keywords: time series, longitudinal data, functional data, clustering, \Frechet{} distance, dynamic time warping, approximation algorithms.
\end{abstract}
\vfill
\thispagestyle{empty}
\pagebreak
\setcounter{page}{1}

\section{Introduction}

Time series are sequences of discrete measurements of a continuous signal. Examples of data that are often represented as time series include stock market values, electrocardiograms (ECGs), temperature, the earth's population, and the hourly requests of a webpage. In many applications, we would like to automatically analyze time series data from different sources: for example in industry 4.0 applications, where the performance of a machine is monitored by a set of sensors. An important tool to analyze (time series) data is clustering. The goal of clustering is to partition the input into groups of similar time series. Its purpose is to discover hidden structure in the input data and/or summarize the data by taking a representative of each cluster.  Clustering is fundamental for performing tasks as diverse as data aggregation, similarity retrieval, anomaly detection, and monitoring of time series data.  Clustering of time series is an active research topic
\cite{
ar-dcaa-2013,
Boecking2014129,
cmp-fcis-07, 
Hsu2014358,
jp-fdc-13,
Li2013243,
liao05survey,
tc-sbc-12,
Petitjean2011678,
rm-fcbwm-06, 
sf-mpfa-14,
hd-tsc-14,
wsh-cbc-06,
zt-MODIS-14,
qg-nadtw-12}
however, most solutions lack a rigorous algorithmic analysis. Therefore, in this paper we study the problem of clustering time series from a theoretical point of view. 
%One line of this research focuses on finding \emph{local
%patterns}~\cite{bbgll-dcpcs-08, lee2007trajectory, li2011movemine,
%Trasarti2011Mining,keogh2011shapelets}. While this problem is very different,
%since its difficulty lies in finding the right partitioning, we think that
%advancements in clustering the trajectories as a whole may eventually also lead
%to better pattern mining algorithms.

Formally, a \emph{time series} is a recording of a signal that changes over time. It consists of a series of paired values $(w_i,t_i)$, where $w_i$ is the $i$th measurement of the signal, and $t_i$ is the time at which the $i$th measurement
was taken. A common approach is to treat time series data as point data in high-dimensional space.  That is, a time series
$S_j=(w_1,t_1),\dots,(w_m,t_m)$ of the input set is treated as the point
$S^{*}_j=(w_1,\dots,w_m)$ in $m$-dimensional Euclidean space. Using this simple interpretation of the data, any clustering algorithm for points can be applied. Despite it being common practice (for a survey, see~\cite{liao05survey}), it has many limitations. 
One major drawback is the requirement that all time series must have the same length and the sampling must be regular and synchronized. In particular, the latter requirement is often hard to achieve. 

In this paper, we follow a different approach that has also received much attention in the literature. In order to formulate the objective function of our clustering problem, we consider a distance measure that allows for irregular sampling and local shift and that only depends on the ``shape'' of the analyzed time series: the \Frechet{} distance.
The \Frechet{} distance is defined for continuous functions $f:[0,1]\rightarrow
\Re$.
Two functions $f,g$ have \Frechet{} distance at most $\delta$, if one can move simultaneously 
but at different and possibly changing positive speed on the domains of $f$ and $g$ from $0$ to $1$ such 
that at all times $|f(x)-g(y)|\le \delta$, where $x$ and $y$ are the current positions on the respective domains. 
The \Frechet{} distance is the infimum over the values $\delta$ that allow such a movement. For a formal
definition, see \secref{problem}. The \Frechet{} distance between two monotone
functions $f,g:[0,1] \rightarrow \Re$ equals the maximum of $|f(0)-g(0)|$ and $|f(1)-g(1)|$. 
This also implies, that the \Frechet{} distance between two functions is completely 
determined by the sequence of local extrema ordered from $0$ to $1$. 
In order to consider the \Frechet{} distance of time series, we view a time series as a specification of
the sequence of local extrema of a function. 
%This is equivalent to assuming that the function is obtained
%by linear interpolation between the measurements. 

Once we have defined the distance measure we can also formulate the clustering problem we want to study.
We will look at two different variants: $k$-center clustering and $k$-median clustering. Both methods are
based on the idea of representing a cluster by a cluster center, which can be thought of as being a
representative for the cluster. In $k$-center clustering, the objective is to find a set of $k$ time series (the cluster centers)
such that the maximum \Frechet{} distance to the nearest cluster center is minimized over all input time series.
In $k$-median clustering, the goal is to find a set of $k$ centers that minimizes the sum of distances to 
the nearest centers. However, this is not yet the problem formulation we consider.

We would like to address another problem that often occurs with time series: noise. In many applications where
we study time series data, the measurements are noisy. For example, physical measurements typically have measurement errors
and when we want to determine trends in stock market data, the effects of short term trading is usually not of interest.
Furthermore, a cluster center that minimizes the \Frechet{} distance to many curves may have a complexity (length of the time series)
that equals the sum of the complexities of the time series in the cluster. This will typically lead to a vast overfitting
of the data. We address this problem by limiting the complexity of the cluster center. Thus,
our problem will be to find $k$ cluster centers each of complexity at most $\ell$ that minimize the $k$-center and
$k$-median objective under the \Frechet{} distance.
It seems that  none of the existing approaches summarizes the
input along the time-dimension in such a way.  

\paragraph{Our results} 

To the best of our knowledge, clustering under the \Frechet{} distance has not been studied before. 
We develop the first $(1+\eps)$-approximation algorithm for the $k$-center and the $k$-median problem 
under \Frechet{} distance when the complexity of the centers is at most $\ell$. For constant $\eps, 
k$ and $\ell$, the running time of our algorithm is $\widetilde O(nm)$, where $n$ is the number of input 
time series and $m$ their maximum complexity. (see \thmref{k:l:center:main} and \thmref{k:l:median:main}.)
We also prove that clustering univariate curves is \NP-hard and show that the doubling dimension of the 
\Frechet{} (pseudo)metric space is infinite for univariate curves.

\paragraph{Challenge and ideas}
High-dimensional data pose a common challenge in many clustering applications.
The challenge in clustering under the \Frechet distance is twofold:
\begin{compactenum}[(A)]
\item High dimensionality of the joint parametric
space of the set of time series:  For the seemingly simple task of computing the (\Frechet{})
\emph{median} of a fixed group of time series, state-of-the-art algorithms run
exponentially in the number of time series~\cite{ahn2015middle, hr-fdre-14}, since the standard
approach is to search the joint parametric space for a monotone 
path~\cite{ag-cfdbt-95, akw-mpcfd-10,buchin2012four, bbw-eapcm-09, feldman2012gps, GudVahr12,
hr-fdre-14, mssz-fdsl-2011}. 
\item High dimensionality of the metric space: the doubling dimension of the \Frechet metric space is
infinite, as we will show, even if we restrict the length of the time series.
Existing ($1+\eps)$-approximation algorithms \cite{abs-cm-10, kumar2010lineartime} with comparable running time are only known to work for special cases
as the Euclidean $k$-median problem or (more generally) for metric spaces with bounded doubling dimension 
%\etal~
\cite{abs-cm-10}.
\end{compactenum}
These two challenges make the clustering task particularly difficult, even for short univariate time
series. Our approach exploits the low dimensionality of the ambient space of the
time series and the fact that we are looking for low-complexity cluster 
centers which best describe the input. 
We introduce the concept of \emph{signatures} (\defref{signature}) which capture critical points of
the input time series. We can show that each signature vertex of an input curve
needs to be matched to a different vertex of its nearest cluster center~(\lemref{nec:suff}).  
Furthermore, we use a technique akin to \emph{shortcutting}, which has been used
before in the context of partial curve matching~\cite{bbw-eapcm-09, dh-jydfd-13}.
We show that any vertex of an optimal solution that is \emph{not} matched to a
signature vertex, can be omitted from the solution without changing its
cost~(\thmref{remove:one}).

These ingredients enable us to generate a constant-size set of candidate solutions.
The ideas going into the individual parts can be summarized as follows.
For the $k$-center problem we generate a candidate set based on the 
entire input and a decision parameter. If the candidate set turns out too large,
then we conclude that there exists no solution for this parameter, since more
vertices would be needed to ``describe'' the input. 
For the $k$-median problem we use the approach of random sampling used
previously by Kumar~\etal~\cite{kumar2010lineartime} and Ackermann~\etal~\cite{abs-cm-10}.
We show that one can generate a constant-size candidate set that contains a
$(1+\eps)$-approximation to the $1$-median based on a sample of constant size.
To achieve this, we observe that a vertex of the optimal solution that is
matched to a signature vertex and which is unlikely to be induced by our
sample, can be omitted without increasing the cost by a factor of more than
$(1+\eps)$~(\lemref{omit:low:prob}).

\paragraph{Related Work}
A distance measure that is closely related to the \Frechet{} distance is \emph{Dynamic time warping (DTW)}.
DTW has been popularized in the field of data mining~\cite{dtswk-08,mueller07dtw} and is known for its unchallenged
universality.
It is a discrete version of the \Frechet
distance which measures the total sum of the distances taken at certain fixed
points along the traversal. The process of traversing the curves with varying
speeds is sometimes referred to as ``time-warping'' in this context.
It has been successfully used in the automated classification of time series
describing phenomena as diverse as surgical processes \cite{Forestier2012255},
whale singing \cite{bmp-ackw-07}, chromosomes \cite{Legrand2008215},
fingerprints \cite{888711}, electrocardiogram (ECG) frames \cite{1013101} and
many others.  
While DTW was born in the 80's from the use of dynamic programming with the
purpose of aligning distorted speech signals, the \Frechet{} distance was
conceived by Maurice \Frechet{} at the beginning of the 20th century in the
context of the study of general metric spaces. 
The best known algorithms for computing either distance measure between two
input time series have a worst-case running time which is roughly quadratic in
the number of time stamps~\cite{buchin2012four, mueller07dtw}. Faster
algorithms exist for both problems under certain input
assumptions~\cite{dhw-afd-12,keogh2005exact}. Recently, Bringmann
showed~\cite{Bringmann14} that the \Frechet{} distance between two polygonal
curves lying in the plane cannot be computed in strongly subquadratic time, unless
the Strong Exponential Time Hypothesis (SETH) fails. This was extended by Bringmann and Mulzer to the 1-dimensional case \cite{bm-adfd-15}.

Both distance measures consider
only the ordering of the measurements and ignore the explicit time stamps. This
makes them robust against local deformations.  Both distance measures deal well
with irregular sampling and have been used in combination with curve
simplification techniques~\cite{ahmw-nltcs-05,dhw-afd-12,keogh1999scaling}.
However, while the \Frechet distance is inherently independent of the sampling
of the curves, DTW does not work well when one of the two curves is sampled
much less frequently (see \figref{subsampling}). Since we are interested in finding cluster centers of low
complexity, we therefore focus on the \Frechet distance.

\begin{figure}\centering
\includegraphics{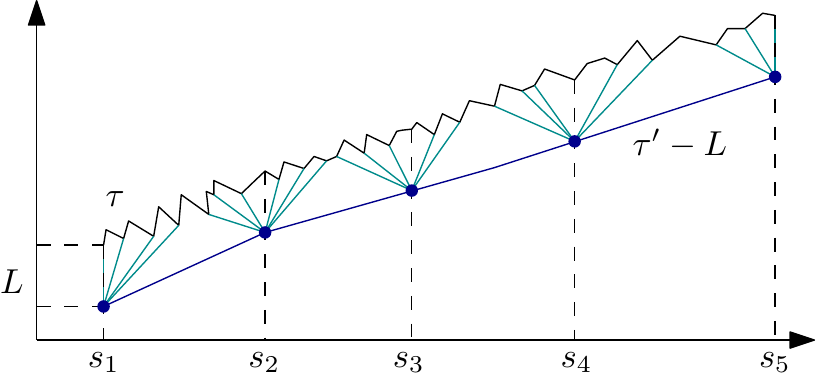}
\caption{Illustration why a continuous distance measure performs better than a
discrete distance measure if sampling rates differ substantially. Shown is the
curve $\tau$ and the same curve subsampled at times $s_1,\dots,s_5$ (for
better visibility, the subsampled curve $\tau'$ is translated by $L$). A discrete 
assignment between the vertices of the curves is shown in light blue.
$\tau$ records a linear process with some error $\eps$. A suitable distance measure should estimate
the distance between $\tau$ and $\tau'$ in the interval $[0,2\eps]$. This is true
for the (normalized) DTW, if the rate of subsampling is high enough. Similarly
it is true for the (continuous) \Frechet distance at any rate of subsampling. 
However, as the rate of subsampling decreases, DTW will estimate the distance
increasingly larger.
}
\figlab{subsampling}
\end{figure}

The problem of clustering points in general metric spaces has been extensively studied 
in many different settings. The problem is known to be computationally hard in most of its
variants~\cite{aloise2009nphard,jain2002greedy,megiddo1984geo,Vaziranibook}.
A number of polynomial-time constant-factor approximation algorithms are known 
%whose running is polynomial in the input size and the number of clusters
~\cite{AryaGKMMP04,CharikarG05, cgts-kmedian-02, CharikarL12, c-kmc-09,feder1988optimal,Gonzalez1985,HochbaumShmoys1985,jain2002greedy,ls-akm-13}. 
For the $k$-center problem the best such algorithm achieves a $2$-approximation
\cite{Gonzalez1985, HochbaumShmoys1985} and this is also the best lower bound
for the approximation ratio of the polynomial-time algorithms unless $\textbf{P}=\NP$. This picture
looks much different if one makes certain assumptions on the graph that defines
the metric, as done in the work by Eisenstadt~\etal~\cite{emk-akcpg-14}, for
example. 

For the $k$-median problem in a general metric space a
$(1+\sqrt{3}+\eps)$-approximation can be achieved \cite{ls-akm-13}. The
best lower bound for the approximation ratio of polynomial time algorithm is
$1+2/e\approx 1.736$ unless $\NP\subseteq \text{DTIME}\left[ n^{O(\log\log n)}\right]$ \cite{jain2002greedy}. If the data points are from the
Euclidean space $\Re^d$, a number of different algorithms are known that compute
%a $(1+\eps)$-approximation to the $k$-median problem in time polynomial in the number of input points $n$ and cluster centers $k$ \cite{FeldmanLangberg11, HarPeledKushal07, harpeledmazumdar2004}, as well of algorithms that have exponential time dependency on $k$ \cite{c-kmc-09,kumar2010lineartime}.
a $(1+\eps)$-approximation to the $k$-median problem 
%in time polynomial in the number of input points but exponential in $k$ 
\cite{c-kmc-09,FeldmanLangberg11, HarPeledKushal07, harpeledmazumdar2004, KR99, kumar2010lineartime}. 
Ackermann \etal~\cite{abs-cm-10}
show that under certain conditions a randomized $(1+\eps)$-approximation algorithm by Kumar
\etal~\cite{kumar2010lineartime} can be extended to general distance measures. 
In particular, they show that one can compute a $(1+\eps)$-approximation to 
the $k$-median problem if the distance metric is such that the $1$-median 
problem can be $(1+\eps)$-approximated using information only from a random
sample of constant size. They show that this is the case for
metrics with bounded doubling dimension. In our case, the doubling dimension is,
however, unbounded.

%Indyk gave a randomized algorithm for the $k$-median problem in general metric
%spaces~\cite{i-sta-99}. 
%However, his problem formulation assumes a
%constant-time distance oracle and requires that cluster centers are chosen from
%the input set.  In the same paper, he describes a fast procedure for
%probabilistic comparison of two candidate solutions for the 1-median problem. 
%This technique was extended and simplified by Thorup~\cite{thorup04}.

A different approach would be to use the technique of metric
embeddings, which proved successful for clustering in Euclidean 
high-dimensional spaces. A metric embedding is a mapping between two metric
spaces which preserves distances up to a certain distortion.
Unlike the similar Hausdorff distance, for which non-trivial embeddings are
known, it is not known if the \Frechet distance can be embedded into
an $\ell_{p}$ space using finite dimension~\cite{IndMat04}.
Any finite metric space can be embedded into
$\ell_\infty$, which is therefore considered as ``the
mother of all metrics''. In turn, any bounded set in $\Re^d$ with $\ell_{\infty}$ 
can be embedded as curves with the \Frechet distance (see also
\secref{nphard}).   Recent work by Bartal~\etal~\cite{bartal2014impossible}
suggests that, due to the infinite doubling dimension, a metric embedding of the
\Frechet distance into an $\ell_p$-space needs to have at least super-constant
distortion.  On the positive side, Indyk showed how to embed the \Frechet
distance into an inner product space for the purpose of designing approximate
nearest-neighbor data structures \cite{i-approxnn-02}. 
However, this work focuses on a different version of the \Frechet distance,
namely the \emph{discrete} \Frechet distance, which is aimed at sequences
instead of continuous curves. For any $t>1$, the resulting data structure is
$O(\pth{\log m \log \log n}^{t-1})$-approximate, uses space exponential in
$tm^{1/t}$, and achieves $(m+\log n)^{O(1)}$ query time, where $m$ is the maximum length of
a sequence and $n$ is the number of sequences stored in the data structure. 

The problem of clustering under the DTW distance has also been studied in the data mining community. In particular, substantial effort has been put into extending Lloyd's algorithm~\cite{lloyd82} to the case of DTW, an example is the work of Petitjean \etal~\cite{Petitjean2011678}. In order to use Lloyds algorithm, one first has to be able to compute the mean of a set of time series. For DTW, this turns out to be a non-trivial problem and current solutions are not satisfactory.  One of the problems is that the complexity of the mean is prone to become quite high, namely linear in the order of the total complexity of the input. This can lead to overfitting in the same way as discussed for the \Frechet distance. For a more extensive discussion we refer to~\cite{ar-dcaa-2013} and references therein.

Also in statistics, the problem of clustering longitudinal data, or functional data, is an active research topic and solutions based on wavelet decomposition and principal component analysis have been suggested \cite{cmp-fcis-07, jp-fdc-13, rm-fcbwm-06, wsh-cbc-06}.

\section{Preliminaries}

A \emph{time series} is a series $(w_1,t_1),\ldots ,(w_m,t_m)$ of measurements
$w_i$ of a signal taken at times $t_i$. We assume $0=t_1<t_2<\ldots <t_m=1$ and $m$ is finite. A
time series may be viewed as a continuous function $\tau: [0,1] \rightarrow \Re$ by
linearly interpolating $w_1,\dots,w_m$ in order of $t_i$, $i=1,\ldots m$. We
obtain a polygonal curve with \emph{vertices} $w_1=\tau(t_1),\dots,
w_m=\tau(t_m)$ and segments between $w_i$ and $w_{i+1}$ called \emph{edges}
$\overline{w_i w_{i+1}}=\{xw_i+(1-x)w_{i+1}|x\in [0,1]\}$. We will simply refer
to $\tau$ as a \emph{curve}. 
When defining such a curve, we may write ``curve $\tau:[0,1]\rightarrow\Re$ with $m$
vertices'', or ``curve $\tau=w_1,\ldots, w_m$''.
\footnote{Notice that this notation does not specify the points of time at which the measurements
are taken. The reason is that the \Frechet distance only depends on the ordering of
the measurements (and their values), but not on the exact points of time when the measurements
are taken.}
%
%We will simply write \emph{curve $\tau:[0,1]\rightarrow\Re$ with $m$ vertices}
%or \emph{curve $\tau=w_1,\ldots, w_m$}, instead of \emph{curve
%$\tau:[0,1]\rightarrow\Re$, with vertices $w_1=\tau(t_1),\ldots,w_m=\tau(t_m)$,
%where $0=t_1<\ldots<t_m=1$, and with $\tau(t_i)\in\Re$}.
We say that such a curve $\tau$ has \emph{complexity} $m$. 
We denote its set of vertices with  $\VtxSet(\tau)$. For
any $t_i<t_j \in [0,1]$, we denote the subcurve of $\tau$ starting at
$\tau(t_i)$ and ending at $\tau(t_j)$ with $\tau[t_i,t_j]$. 
We define 
$\minSubC{\tau}{t}{t^{-}}{t^{+}}=\min\setSubC{\tau}{t}{t^{-}}{t^{+}}$, and
$\maxSubC{\tau}{t}{t^{-}}{t^{+}}=\max\setSubC{\tau}{t}{t^{-}}{t^{+}}$, 
to denote the minimum and maximum along a (sub)curve.

Furthermore, we will use the following non-standard notation for an interval.
For any $a,b\in \Re$ we define $\cbrc{a,b}=\pbrc{\min(a,b), \max(a,b)}$. 
We denote $\range{h}{\delta}=[h-\delta, h+\delta]$. For any
set $P$, we denote its cardinality with $|P|$.

Let $\mathcal H$ denote the set of continuous and increasing functions
$f:[0,1]\rightarrow[0,1]$ with the property that $f(0)=0$ and $f(1)=1$. 
For two given functions
%signals
$\tau:[0,1]\rightarrow \Re$ and $\pi:[0,1]\rightarrow \Re$,
their \emph{\Frechet distance} is defined as 
\begin{equation} \label{def:frechet} 
\distFr{\tau}{\pi}=\inf_{f\in \mathcal H}\; \max_{t \in [0,1]} \| \tau(f(t))-\pi(t)
\|,
\end{equation}
The \Frechet{} distance between two time series is defined as the \Frechet{} distance
of their corresponding 
continuous functions.
%signals.
\footnote{
In the above definition, the \Frechet{} distance is only defined for functions with domain $[0,1]$.
This is mostly for simplicity of analysis.
We can easily extend it to other domains using an arbitrary homeomorphism that identifies this domain with $[0,1]$. 
The \Frechet{} distance is invariant under reparametrizations. 
All our algorithms use the parametrizations only implicitely. Any input time
series is given as an ordered list of measurements, without explicit
time-stamps. Therefore our definitions are without loss of generality.
}
Note that any $f \in \mathcal H$ induces a bijection between the two curves.
We refer to the function $f$ that realizes the \Frechet distance as a
\emph{matching}. It may be that such a matching exists in the limit only. That
is, for any $\eps>0$, there exists a $f\in \mathcal H$ that matches each point
on $\tau$ to a point on $\pi$ within distance $\distFr{\tau}{\pi}+\eps$.

The \Frechet{} distance is a pseudo metric~\cite{ag-cfdbt-95}, i.e. it satisfies all properties of
a metric space except that there may be two different functions that have distance
$0$. If one considers the equivalence classes that are induced by functions of pairwise
distance $0$ we can obtain a metric space $(\Delta,\distFrConst)$ defined by the \Frechet
distance and the set $\Delta$ of all (equivalence classes of) univariate time series.
We denote with $\Delta_m$ the set of all univariate time series of complexity at
most $m$.

We notice that only the ordering of the $w_i$ is relevant and that under \Frechet distance two curves can be thought of being identical, if they have the same sequence of local minima and maxima. Therefore, we can assume that a curve is induced by its sequence of local minima and maxima and we will use the term curve in the paper to describe the equivalence class of curves with pairwise \Frechet distance 0.

\begin{definition}\deflab{concatenation}
Let two curves $\tau_1:[a_1,b_1]\rightarrow\mathbb{R}$, $0\leq a_1\leq b_1\leq 1$, and $\tau_2:[a_2,b_2]\rightarrow\mathbb{R}$, $0\leq a_2\leq b_2\leq 1$ be given, such that $\tau_1(b_1)=\tau_2(a_2)$. The concatenation of $\tau_1$ and $\tau_2$ is a curve $\tau$ defined as $\tau=\tau_1\oplus\tau_2:[0,1]\rightarrow\mathbb{R}$, such that 
\[ \tau(t)=\left(\tau_1\oplus\tau_2\right)(t)=\begin{cases}\tau_1\left(a_1+\left(b_1-a_1+b_2-a_2\right)\cdot t\right) & \text{if  } t\leq \frac{b_1-a_1}{b_1-a_1+b_2-a_2}
\\ \tau_2\left(b_2-\left(b_1-a_1+b_2-a_2\right)\cdot \left(1-t\right)\right) & \text{if  } t> \frac{b_1-a_1}{b_1-a_1+b_2-a_2} 
\end{cases}\]

\end{definition}

We are going to use the following simple observations throughout the paper. 
\begin{observation}\obslab{fd:concat}
Let two curves $\tau:[0,1]\rightarrow \Re$ and $\pi:[0,1]\rightarrow \Re$
be the concatenations of two subcurves each, $\tau=\tau_1\oplus \tau_2$ and
$\pi=\pi_1\oplus \pi_2$, then it holds  that \[\distFr{\tau}{\pi}\leq \max\lbrace
\distFr{\tau_1}{\pi_1}, \distFr{\tau_2}{\pi_2}\rbrace\]
\end{observation}

\begin{observation}\obslab{segments}
Given two edges $\overline{a_1 a_2}$ and $\overline{b_1 b_2}$ with $a_1,a_2,b_1,b_2\in \Re$, it holds that 
\[\distFr{\overline{a_1 a_2}}{\overline{b_1 b_2}}=\max\pth{|a_1-b_1|,|a_2-b_2|}.\]
\end{observation}

We make the following \emph{general position} assumption on the input. For every
input curve $\inputTraj{}$ we assume that no two vertices $\inputTraj{}$ have
the same coordinates and any two differences between coordinates of two
vertices of $\inputTraj{}$ are different. This assumption can easily be achieved
by symbolic perturbation. Furthermore we assume that $\inputTraj{}$ has no edges 
of length zero and its vertices are an alternating sequence of minima and
maxima, i.e. no vertex lies in the linear interpolation of its two neighboring vertices.

\subsection{Problem statement}
\seclab{problem}

Given a set of $n$ time series $P=\lbrace\inputTraj{1},\ldots , \inputTraj{n}\rbrace \subseteq \Delta_m$ and parameters $\nrClusters,\lenClusters \in \Na$, we define a \emph{$(\nrClusters,\lenClusters)$-clustering}  as a set of $\nrClusters$ time series $C=\lbrace\centerTraj{1},\ldots , \centerTraj{\nrClusters}\rbrace$ taken from $\Delta_{\lenClusters}$  which minimize one of the following cost functions:
\begin{equation} \label{def:kcenter} 
\cost_{\infty}(P,C)=\max_{i=1,\ldots n}\; \min_{j=1,\ldots \nrClusters}
\distFr{\inputTraj{i}}{\centerTraj{j}}. \end{equation} 

\begin{equation} \label{def:kmedian} 
\cost_1(P,C)=\sum_{i=1,\ldots n}\; \min_{j=1,\ldots \nrClusters}
\distFr{\inputTraj{i}}{\centerTraj{j}}. \end{equation} 

We refer to the clustering problem as 
\emph{$(\nrClusters,\lenClusters)$-center} (Equation~\ref{def:kcenter}) 
and \emph{$(\nrClusters,\lenClusters)$-median} (Equation~\ref{def:kmedian}),
respectively. 
We denote the cost of the optimal solution as 
\[ \opt^{(i)}_{\nrClusters,\lenClusters}(P) = \min_{C \subset \Delta_{\lenClusters}} \cost_{i}(P,C),\]
where the restrictions on $C$ are as described above and $i \in \{\infty,1\}$.
Note that this corresponds to the classical definition of the $k$-median problem (resp. $k$-center problem) 
%if we choose the (non-) metric space $(\Delta_{\ell} \cup P, D)$, where the distance function $D$ is 
if $D$ is a distance measure on $\Delta_{\ell} \cup P$, 
defined as $D(p,q)=\infty$ for $p,q \in P$
and $D(p,q)=\distFr{p}{q}$ otherwise. Note that the new distance measure $D$
does not satisfy the triangle inequality and is therefore not a metric.

\section{On signatures of time series}
\seclab{on:signatures}

Before introducing our signatures, we first review similar notions traditionally
used for the purpose of curve compression. 
A simplification of a curve is a curve which is lower in complexity (it has
fewer vertices) than the original curve and which is similar to the original
curve. This is captured by the following standard definitions.  

\begin{definition}\deflab{min:error:simp}
We call a curve $\pi$ a minimum-error $\ell$-simplification of $\inputTraj{}$ if
for any curve $\pi'$ of at most $\ell$ vertices, it holds that
$\distFr{\pi'}{\tau} \geq \distFr{\pi}{\tau}$.
\end{definition}

\begin{definition}\deflab{min:size:simp}
We call a curve $\pi$ a minimum-size $\varepsilon$-simplification of $\inputTraj{}$ if $\distFr{\pi}{\inputTraj{}}\leq \varepsilon$ and
for any curve $\pi'$ such that  $\distFr{\pi'}{\inputTraj{}}\leq \varepsilon$, it holds that the complexity of $\pi'$ is at least as much as the complexity of $\pi$.
\end{definition}

The simplification problem has been studied under different names for
multidimensional curves and under various error measures, in domains, such as
cartography~\cite{dp-73,ramer1972iterative}, computational
geometry~\cite{godau1991natural}, and pattern recognition~\cite{pratt2002search}. 
  Often, the simplified curve is restricted to vertices of the input curve and
  the endpoints are kept. However, in our clustering setting, we need to use  
  the more general problem definitions stated above.
  
  Historically, the first minimal-size curve simplification algorithm was a
  heuristic algorithm independently suggested in the 1970's by Ramer and Douglas
  and Peucker~\cite{dp-73,ramer1972iterative} and it remains popular in the area
  of geographic information science until today.  It uses the Hausdorff error
  measure and has running time $O(n^2)$ (where $n$ denotes the complexity of the
  input curve), but does not offer a bound to the size of the simplified curve.
  Recently, worst-case  and average-case lower bounds on the number of vertices
  obtained by this algorithm were proven by
  Daskalakis~\etal~\cite{daskalakis2010good}.  Imai and Iri~\cite{ii-pac-1988}
  solved both the minimum-error and minimum-size simplification problem under
  the Hausdorff distance by modeling it as a shortest path problem in directed
  acyclic graphs.

Curve simplification using the \Frechet distance was first proposed by
Godau~\cite{godau1991natural}.  The current state-of-the-art approximation algorithm for
simplification under the \Frechet distance was suggested by
Agarwal~\etal~\cite{ahmw-nltcs-05}. This algorithm computes a $2$-approximate
minimal-size simplification in time $O(n\log n)$.  The framework of Imai and Iri
is also used for the streaming algorithm of Abam~\etal~\cite{abh-sals-10} under
the \Frechet distance.  Driemel and Har-Peled~\cite{dh-jydfd-13} introduced the
concept of a \emph{vertex permutation} with the aim of preprocessing a curve for
curve simplification. The idea is that any prefix of the permutation represents
a bicriteria approximation to the minimal-error curve simplification. In
\secref{computing:signatures} we will use this concept to develop improved
algorithms in our setting where the curves are time series.  

For time series, a concept similar to simplification  called 
\emph{segmentation} has been extensively studied in the area of data
mining~\cite{bingham2006segmentation,himberg01,terzi2006efficient}.
The standard approach for computing exact segmentations is to use dynamic
programming which yields a running time of $O(n^2)$.

\begin{figure}[h]
\centering
\includegraphics{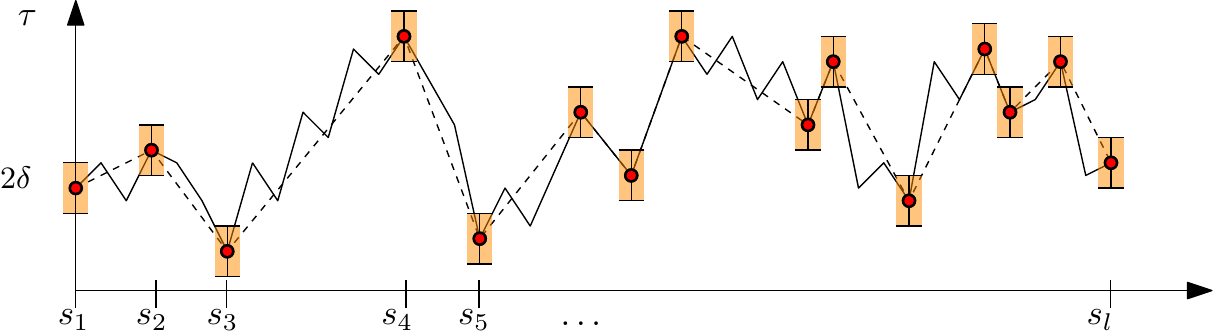}
\caption{Example of a $\delta$-signature of a time series}
\figlab{signature}
\end{figure}

We now proceed to introduce the concept of signatures (\defref{signature}).  Our definition
aligns with the work on computing important minima and
maxima in the context of time series compression~\cite{pratt2002search}.
Intuitively, the signatures provide us with the ``shape'' of a time series at
multiple scales.  Signatures have a unique hierarchical structure
(see~\lemref{canonical:signature}) which we can exploit in order to achieve
efficient clustering algorithms. Furthermore, the signatures of a curve
approximate the respective simplifications under the \Frechet distance (see \lemref{apx:min:error:simp}).
\figref{signature} shows an example of a signature. We show several crucial
properties of signatures in \subsecref{properties:signatures}.  Signatures
always exist and are easy to compute (see \secref{computing:signatures}).  In
particular, we show how to compute the $\delta$-signature of a curve in one pass
in linear time~(\thmref{computing:delta:signature}), and how to preprocess a
curve in near-linear time for fast queries of the signature of a certain
size~(\thmref{compute:signatures}). 

\begin{definition}[$\delta$-signature]\deflab{signature}
We define the $\delta$-signature of any 
%(polygonal) 
curve $\tau: [0,1]\rightarrow
\Re$ as follows.
The signature is a curve $\sigma: [0,1] \rightarrow \Re$ defined by a series of values $0=t_1 <\dots<t_\lenClusters=1$ as the linear interpolation of $\tau(t_i)$ in the order of the index $i$, and with the following properties.\\ 
For $1 \leq i \leq \lenClusters -1$ the following conditions hold: 
\begin{compactenum}[(i)]
\item (non-degeneracy) if $i \in [2,\lenClusters-1]$ then $\tau(t_i) \notin \langle \tau(t_{i-1}), \tau(t_{i+1})\rangle $,
\item (direction-preserving)\\if $\tau(t_i)<\tau(t_{i+1})$ for $t < t' \in [t_i,t_{i+1}]$:
$\tau(t)-\tau(t') \leq 2\delta$, and \\if $\tau(t_i)>\tau(t_{i+1})$ for $t < t'
\in [t_i,t_{i+1}]$: $\tau(t')-\tau(t) \leq 2\delta$,
\item (minimum edge length)\\if $i \in [2,\lenClusters-2]$ then $|\tau(t_{i+1})-\tau(t_i)| > 2\delta$, and \\if
$i\in \{1,\lenClusters-1\}$ then $|\tau(t_{i+1})-\tau(t_{i})| > \delta$,
\item (range) for $t \in [t_i,t_{i+1}]$: \\if 
$i \in [2,\lenClusters-2]$ then $\tau(t) \in \langle\tau(t_i), \tau(t_{i+1})\rangle$, and \\if 
$i=1$ and $\lenClusters>2$ then $\tau(t) \in \langle \tau(t_i),\tau(t_{i+1})\rangle \cup \langle
\tau(t_i)-\delta,\tau(t_i)+\delta\rangle$, and \\if 
$i=\lenClusters -1$  and $\lenClusters>2$ then $\tau(t) \in \langle \tau(t_{i-1}),\tau(t_i)\rangle \cup \langle
\tau(t_i)-\delta,\tau(t_i)+\delta\rangle$, and \\if
$i=1$ and $\lenClusters=2$ then $\tau(t) \in \langle \tau(t_{1}),\tau(t_2)\rangle \cup \langle
\tau(t_1)-\delta,\tau(t_1)+\delta\rangle \cup \langle
\tau(t_2)-\delta,\tau(t_2)+\delta\rangle.$
\end{compactenum}
\end{definition}

%this is new
It follows from the properties (i) and (iv) of \defref{signature} that the parameters $t_i$ for $i\in [1,\ell]$ specify vertices of $\tau$. 
Furthermore, it follows that the vertex $\tau(t_i)$ is either a minimum or maximum on $\tau[t_{i-1}, t_{i+1}]$ for $i\in [2,\ell-1]$.

For a signature $\sigma$ we will simply write \emph{signature $\sigma:[0,1]\rightarrow\Re$ with $\lenClusters$ vertices} or \emph{signature $\sigma=v_1,\ldots,v_\lenClusters$}, instead of \emph{signature $\sigma:[0,1]\rightarrow\Re$, with vertices $v_1=\sigma(s_1),\ldots,v_\lenClusters=\sigma(s_\lenClusters)$, where $0=s_1<\ldots<s_\lenClusters=1$}. We assume that the parametrization of $\sigma$ is chosen such that $\sigma(s_j)=\tau(s_j)$, for any $j\in\lbrace 1,\ldots,\lenClusters\rbrace$.

We remark that while the above definition is somewhat cumbersome, the stated
properties turn out to be the exact properties needed to prove
\thmref{remove:one}, which in turn enables the basic mechanics of our clustering
algorithms.

\subsection{Useful properties of signatures}
\seclab{properties:signatures}

\begin{lemma}\lemlab{fd:signature}
It holds for any $\delta$-signature $\sigma$ of $\tau$ that $\distFr{\sigma}{\tau}\leq \delta$.
\end{lemma}

\begin{proof}
%\Amer{the partially changed text}
Let $t_1<\dots<t_l \in [0,1]$ be the series of parameter values of vertices on $\tau$ that describe $\sigma$. We construct a greedy matching between each signature edge $e_i=\overline{\tau(t_i)\tau(t_{i+1})}$ and the corresponding subcurve $\widehat{\tau}=\tau[t_i,t_{i+1}]$ of $\tau$. Assume first for simplicity that it holds $\tau(t_i)<\tau(t_{i+1})$ (i.e. the traversal of the signature is directed upwards at the time) and none of its endpoints are endpoints of $\tau$. We process the vertices of the subcurve $\widehat{\tau}$ while keeping a current position $v$ on the edge $e$. The idea is to walk as far as possible on $\widehat{\tau}$ while walking as little as possible on $e_i$. We initialize $v=\tau(t_i)$ and match the first vertex of $\widehat{\tau}$ to $v$. When processing a vertex $w$, we update $v$ to $\max(v,w-\delta)$ and match $w$ to the current position $v$ on $e_i$. By the direction-preserving condition in \defref{signature} and by \obsref{segments} every subcurve of $\widehat{\tau}$ is matched to a subsegment of $e_i$ within \Frechet distance $\delta$. If for the edge $e_i$ it holds that $\tau(t_i)>\tau(t_{i+1})$ (traversal directed downwards) the construction can be done symmetrically by walking backwards on $\widehat{\tau}$ and $e_i$. If the first vertex of $\widehat{\tau}$ is an endpoint of $\tau$, we start the above construction with the first vertex that lies outside the range $[\tau(0)-\delta,\tau(0)+\delta]$. The skipped vertices can be matched to $\tau(0)$. As for the remaining case if the last vertex of $\widehat{\tau}$ is an endpoint of $\tau$, we can again walk backwards on $\widehat{\tau}$ and $e_i$ and the case is analogous to the above.
\footnote{
Technically, the constructed matching is not strictly increasing. However, for any $\eps>0$ it can be perturbed slightly to obtain a proper bijection. The result is then obtained in the limit.
}
\end{proof}

\begin{lemma}\lemlab{nec:suff}
Let $\sigma=v_1,\dots,v_\lenClusters$ be a $\delta$-signature of $\tau=w_1,\dots,w_m$. Let $r_i=[v_i-\delta,v_i+\delta]$,  for $1\leq i \leq \lenClusters$, be ranges centered at the vertices of $\sigma$ ordered along $\sigma$. It holds for any curve $\pi$ if $\distFr{\tau}{\pi}\leq\delta$, then $\pi$ has a vertex in each range $r_i$, and such that these vertices appear on $\pi$ in the order of $i$.
\end{lemma}

\begin{proof}
%new version of the proof
For any $i=\lbrace 3,\ldots \lenClusters-2\rbrace$ the vertices $v_{i-1}, v_i$ and $v_{i+1}$ satisfy that $|v_i-v_{i-1}|> 2\delta$ and $|v_{i+1}-v_{i}|> 2\delta$. This implies that $r_{i-1}\cap r_i = \emptyset$, $r_{i}\cap r_{i+1} = \emptyset$. 
Let $\pi(p_i)$ be the point matched to $v_i$ under a matching that witnesses $\distFr{\tau}{\pi}$ and $p_i\in [0,1]$, for all $1\leq i \leq \ell$. It holds that $p_1<p_2<\ldots <p_\ell$.
%The curve $\pi$ contains the points $\pi(p_{i-1})$, $\pi(p_{i})$ and $\pi(p_{i+1})$ (not necessarily the vertices) matched by \Frechet matching to $v_{i-1}$, $v_i$ and $v_{i+1}$ respectively, and $p_{i-1}<p_i<p_{i+1}$.
Therefore the curve $\pi$ visits the ranges $r_{i-1}$, $r_i$ and $r_{i+1}$ in the order of the index $i$. Since $v_i \notin \langle v_{i-1}, v_{i+1}\rangle $ the curve $\pi$ must change direction (from increasing to decreasing or vice versa) between visiting $r_{i-1}$ and $r_{i+1}$. Furthermore, $\pi$ cannot go beyond $r_i$ between visiting $r_{i-1}$ and $r_{i+1}$, i.e. there is no point $x\in\pi\left[p_{i-1}, p_{i+1} \right]$ such that it holds that $x\notin r_i$ and there is an ordering $v_{i-1}<v_i<v_i+\delta<x$ or $v_{i-1}>v_i>v_i-\delta>x$. This follows from $v_i$ being a local extremum on $\tau$. Therefore, the change of the direction of $\pi$ takes place in a vertex in $r_i$. 

%For any $i=\lbrace 3,\ldots \lenClusters-2\rbrace$ the vertices $v_{i-1}, v_i, v_{i+1}$ satisfy that $|v_i-v_{i-1}|\geq 2\delta$ and $|v_{i+1}-v_{i}|\geq 2\delta$. The curve $\pi$ must visit the ranges $r_{i-1},r_i$ and $r_{i+1}$ in the order of their index. Since $r_{i-1}\cap r_i = \emptyset$, $r_{i}\cap r_{i+1} = \emptyset$ and $v_i \notin \langle v_{i-1}, v_{i+1}\rangle $ the curve $\pi$ must change direction between visiting $r_{i-1}$ and $r_{i+1}$. Furthermore, $\pi$ cannot go beyond $r_i$ between visiting $r_{i-1}$ and $r_{i+1}$ (i.e. if  $u_{i-1}=\pi(p_{i-1})\in r_{i-1}$ and $u_{i+1}=\pi(p_{i+1})\in r_{i+1}$ are two points on $\pi$, there is no point $x\in\pi\left[p_{i-1}, p_{i+1} \right]$  such that $x\notin r_i$ and there is an ordering $v_{i-1}<v_i<x$ or $v_{i-1}>v_i>x$), because $v_i$ is a local extremum on $\tau$ by the conditions (i) and (iv) of Definition~\ref{def:signature} and it holds that $\distFr{\tau}{\pi}\leq\delta$. Therefore, $\pi$ needs to have a vertex in $r_i$. 

For $i=2$ we use a similar argument. Note that $\pi(0)$ has to be matched to $v_1$ by the definition of the \Frechet distance. As before, $\pi$ has to visit the ranges $r_2$ and $r_3$ in this order and it holds that $r_2 \cap r_3 = \emptyset$. Either the first vertex of $\pi$ already lies in $r_2$, or again $\pi$ has to change direction and therefore needs to have a vertex in $r_2$. The case $i=\lenClusters-1$ is symmetric.  
The fact that the points $\tau(0)$ and $\tau(1)$ have to be matched to $\pi(0)$ and $\pi(1)$ closes the proof.
%The fact that the points $\pi(0)$ and $\pi(1)$ have to be matched to $p_1$ and $p_s$ closes the proof.
\end{proof}

%Alte Version
%The following is a direct implication of \lemref{nec:suff} and the minimum-edge-length condition in \defref{signature}.
%\begin{corollary}\corlab{nec:suff}
%Let $\sigma$ be a $\delta$-signature of $\tau$. Let 
%$s$ denote the number of vertices of $\sigma$ and let $\delta$ denote their 
%distance $\distFr{\sigma}{\tau}$. Then,
%\begin{compactenum}[(i)]
%\item any curve $\pi$ with $\distFr{\pi}{\tau}\leq \delta$ needs to have at least
%$s-2$ vertices,
%\item for any curve $\pi$ with less than $s-2$ vertices it holds that
%$\distFr{\pi}{\tau} > \delta$. 
%\end{compactenum}
%\end{corollary}

The following is a direct implication of \lemref{nec:suff} and the minimum-edge-length condition in \defref{signature}, since $\sigma$ is a $\delta$-signature and there has to be at least one vertex in each of the ranges centered in vertices which are not endpoints of $\tau$. 
\begin{corollary}\corlab{nec:suff}
Let $\sigma$ be a signature of $\tau$ with $\lenClusters$ vertices and $\distFr{\sigma}{\tau}\leq \delta$. Then any curve $\pi$ with $\distFr{\pi}{\tau}\leq \delta$ needs to have at least $\lenClusters-2$ vertices.
\end{corollary}
%\begin{corollary}\corlab{nec:suff}
%Let $\sigma$ be a $\delta$-signature of $\tau$. Let $s$ denote the number of vertices of $\sigma$ and let $\delta$ denote their distance $\distFr{\sigma}{\tau}$. Then any curve $\pi$ with $\distFr{\pi}{\tau}\leq \delta$ needs to have at least $s-2$ vertices.
%\end{corollary}

\begin{theorem}\thmlab{remove:one}
Let $\sigma=v_1,\dots,v_\lenClusters$ be a $\delta$-signature of $\tau=w_1,\dots,w_m$.
Let $r_j=[v_j-\delta,v_j+\delta]$ be ranges centered at the vertices of $\sigma$
ordered along $\sigma$, where $r_1=[v_1-4\delta,v_1+4\delta]$ and
$r_\lenClusters=[v_\lenClusters-4\delta,v_\lenClusters+4\delta]$. Let $\pi$ be a curve with
$\distFr{\tau}{\pi} \leq \delta$ and let $\pi'$ be a curve obtained by removing 
some vertex $u_i=\pi(p_i)$ from $\pi$ with $u_i \notin \bigcup_{1\leq j \leq \lenClusters} r_j$. 
It holds that $\distFr{\tau}{\pi'} \leq \delta$. 
\end{theorem}

%In order to prove this theorem, we have the following lemma, which looks almost the same.
In order to prove this theorem, we have the following lemma, which is a slight variation of the Theorem and it simplifies the case when the \Frechet distance is obtained in the limit.

\begin{lemma}\lemlab{remove:one}
Let $\sigma=v_1,\dots,v_\lenClusters$ be a $\delta$-signature of $\tau=w_1,\dots,w_m$.
Let $r_j=[v_j-\delta,v_j+\delta]$ be ranges centered at the vertices of $\sigma$
ordered along $\sigma$, where $r_1=[v_1-4\delta,v_1+4\delta]$ and
$r_\lenClusters=[v_\lenClusters-4\delta,v_\lenClusters+4\delta]$. Let $\pi$ be a curve with
$\distFr{\tau}{\pi} <\delta$ and let $\pi'$ be a curve obtained by removing 
some vertex $u_i=\pi(p_i)$ from $\pi$ with $u_i \notin \bigcup_{1\leq j \leq \lenClusters} r_j$. 
For any $\eps>0$, it holds that $\distFr{\tau}{\pi'} \leq \delta+\eps$. 
%\Anne{I changed the ranges around
%the endpoints (from $3\delta$ to $4\delta$). I think it is possible to prove it
%for $3\delta$ by adding another case (if we care).}
\end{lemma}

We obtain the \thmref{remove:one} from \lemref{remove:one} as follows. 
\begin{proof}[Proof of \thmref{remove:one}]
By the theorem statement, we are given $\tau,\pi$ and $\delta$, such that 
$\distFr{\tau}{\pi} \leq \delta$. By the definition of the \Frechet distance
 it holds for any $\eps > 0$ that $\distFr{\tau}{\pi} < \delta+\eps$. 
Let $\delta'=\delta+\eps$ for some $\eps>0$ small enough such that: 
\begin{compactenum}[(i)]
\item the $\delta$-signature of $\tau$ is equal to the $\delta'$-signature of $\tau$ (see also \lemref{canonical:signature} for the existence of such a signature), and
\item any vertex $u_i$ of $\pi$ satisfying the conditions in \thmref{remove:one} also 
satisfies the conditions of \lemref{remove:one} for $\delta'$.
\end{compactenum}
Now we can apply \lemref{remove:one} using $\delta'$, implying that $\distFr{\tau}{\pi'} \leq \delta+2\eps$.
Since this is implied for any $\eps > 0$ small enough, we have
$\lim_{\eps \rightarrow 0} \distFr{\tau}{\pi'} 
\leq \lim_{\eps\rightarrow 0} \left( \delta+2\eps \right) = \delta$. 
\end{proof}

\begin{proof}[Proof of \lemref{remove:one}]
Let $f$ denote the witness matching from $\pi$ to $\tau$, that maps each point on $\pi$ to a point on $\tau$ within distance $\delta$.\footnote{The existence of such a matching $f$ is ensured since $\distFr{\tau}{\pi} <\delta$.} Intuitively, we removed $u_i$ and its incident edges from $\pi$ by replacing the incident edges with a new ``edge'' connecting the two subcurves which were disconnected by the edge removal. The obtained curve is called $\pi'$. We want to construct a matching $f'$ from $\pi'$ to $\tau$ based on $f$ to show that their \Frechet distance is at most $\delta$. 

%Since we define the curves as functions $[0,1]\rightarrow\mathbb{R}$, 
Because of the continuity of the curves, 
we have to describe the ``edge'' connecting disconnected parts. Let $\pi(p_{i-1})$ and $\pi(p_{i+1})$ be the endpoints of the disconnected components. Let $\pi[p^{-},p^{+}]$ denote the subcurve by which $\pi$ and $\pi'$ differ.  In particular, $p^{-}$ and $p^{+}$ are such that $\pi'$ can be written as a concatenation of a prefix and a suffix curve of $\pi$: 
$\pi'= \pi[0,p^{-}] \oplus \pi[p^{+},1]$ 
%$\pi'= \pi[0,p^{-}] \circ \pi[p^{+},1]$ 
and $p_i$ is contained in the open interval $(p^{-},p^{+})$. Note that $\pi(p^{-})=\pi(p^{+})$.  Furthermore, it is clear that $\pi[p^{-},p^{+}]$ consists of two edges with $u_i$ being the minimum or maximum connecting them. (Otherwise, if $u_i$ was neither a minimum nor a maximum on $\pi$, then $\pi[p^{-},p^{+}]$ is empty. In this case the claim holds trivially.) 

The new ``edge'' $\pi'[p_{i-1}, p_{i+1}]$ consists of three parts: the edge $\pi[p_{i-1}, p^-]$, the point $\pi[p^-]$ and the edge $\pi[p^+, p_{i+1}]$. This is illustrated by \figref{vertexremovalexample}.
%Furthermore, either $p^{-}=p_{i-1}$ or $p^{+}=p_{i+1}$, i.e., one of them represents 
%an adjacent vertex of $u_i$. The other point was located on an incident edge of $u_i$.  

\begin{figure}[h]\centering
\includegraphics[page=1]{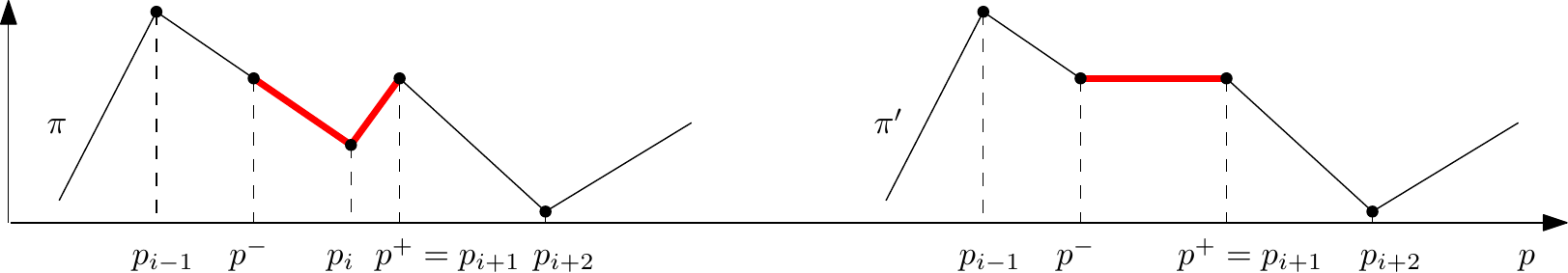}\\
\caption{The removal of the vertex $\pi(p_i)$ from $\pi$. The curve $\pi[p^-,p^+]$ is marked red}
\figlab{vertexremovalexample}
\end{figure}
%\Amer{Alternative paragraph above} Let $u_{i-1}$ and $u_{i+1}$ be the adjacent
%vertices of $u_i$ on $\pi$. If the vertex $u_i$ is neither minimum nor maximum
%on $\pi$ between $u_{i-1}$ and $u_{i+1}$, then the removal of $u_i$ does not
%change the curve $\pi$, and we do not have to do anything. Otherwise let
%$\pi\left[ p^{-}, p^{+}\right]$ be the subcurve by which $\pi$ and $\pi'$
%differ. This means that either $\pi(p^{-})$ or $\pi(p^{+})$ is an adjacent
%vertex of $u_i$ and the other is lying on an incident edge, that is
%conceptually cut by adding the point with $\pi(p^{-})=\pi(p^{+})$. This also
%means that $\pi'$ can be written as a concatenation of a prefix and a suffix
%curve of $\pi$: $\pi'= \pi[0,p^{-}] \circ \pi[p^{+},1]$.

In the construction of $f'$ we need to show that the subcurve $\tau[f(p^{-}),f(p^{+})]$, which was matched by $f^{-1}$ to the missing part, can be matched to some subcurve of $\pi'$, while respecting the monotonicity of the matching. The proof is a case analysis based on the structure of the two curves.  In order to focus on the essential arguments, we first make some global assumptions stated below.  The first two assumptions can be made without loss of generality. We also introduce some basic notation which is used throughout the rest of the proof.
\begin{assumption}
We assume that $\pi(p_i)$ is a local minimum on $\pi$ (otherwise we 
first mirror the curves $\tau$ and $\pi$ across the horizontal time axis to obtain
this property without changing the \Frechet distance).
\end{assumption}

Let $z_{\min}= \argmin_{t \in [f(p^{-}),f(p^{+})]} \tau(t)$. 
Let $\tau[s_j,s_{j+1}]$ be the subcurve of $\tau$ bounded by two
consecutive signature vertices, such that $z_{\min} \in [s_j,s_{j+1}]$.

\begin{assumption}
We assume that $\tau(s_j) < \tau(s_{j+1})$ (otherwise we first reparametrize
the curves $\tau$ and $\pi$ with $\phi(t)=1-t$, i.e., reverse the direction of
the time axis, to obtain this property without changing the \Frechet distance;
note that this does not change the property of $\pi(p_i)$ being a local
minimum).
\asslab{sign:edge:asc}
\end{assumption}

\begin{assumption}
We assume that neither $s_j=0$, nor $s_{j+1}=1$ (These are boundary cases which
will be handled at the end of the proof).
\asslab{inner:sig:edge}
\end{assumption}

\begin{property}[Signature]
By \defref{signature} we can assume that 
\begin{compactenum}[(i)] 
\item $\tau(s_{j+1}) - \tau(s_{j}) > 2\delta$,\label{item:sig:length}
\item $\tau(s_{j}) = \minSubC{\tau}{t}{s_{j-1}}{s_{j+1}}$,\label{item:sig:sj:min}
\item $\tau(s_{j+1}) = \maxSubC{\tau}{t}{s_j}{s_{j+2}}$,\label{item:sig:sj1:max}
\item $\tau(t) \geq \tau(t')-2\delta$ for $s_{j}\leq t' < t \leq
s_{j+1}$.\label{item:sig:descent}
\item $\tau(s_{j+1}) - \tau(s_{j+2}) > 2\delta$,\label{item:sig:length:2}
\end{compactenum}
By the general position assumption the minimum $\tau(s_{j})$ and the maximum
$\tau(s_{j+1})$ are unique on their respective subcurves.
\proplab{signature}
\end{property}

\begin{property}[\Frechet]
Any two points matched by $f$ have distance at most $\delta$ from each other.
In particular, for any two $0\leq p' < p \leq 1$, it holds that
\begin{compactenum}[(i)] 
\item $\tau(f(p)) -\delta \leq \pi(p) 
\leq \tau(f(p)) + \delta$,\label{item:frechet:dist}
\item $\minSubC{\tau}{t}{f(p')}{f(p)} - \delta
\leq \minSubC{\pi}{p}{p'}{p} 
\leq \minSubC{\tau}{t}{f(p')}{f(p)} + \delta$,\label{item:frechet:min}
\item $\maxSubC{\tau}{t}{f(p')}{f(p)} - \delta
\leq \maxSubC{\pi}{p}{p'}{p} 
\leq \maxSubC{\tau}{t}{f(p')}{f(p)} + \delta$.\label{item:frechet:max}
\end{compactenum} 
\proplab{frechet}
\end{property}

Our proof is structured as case analysis. We consider first the case $\tau(z_{\min}) \geq \pi(p^{-}) - \delta$.
%, which means that the curves $\pi'[p^{-}, p^{+}]$ and $\tau[f(p^{-}), f(p^{+})]$ are close to each other. 
This is illustrated by \figref{shortcutting:case1}. 

\begin{case}[Trivial case]
$\tau(z_{\min}) \geq \pi(p^{-}) - \delta$
\caselab{trivial}
\end{case}
\begin{figure}[h]\centering
\includegraphics[page=1]{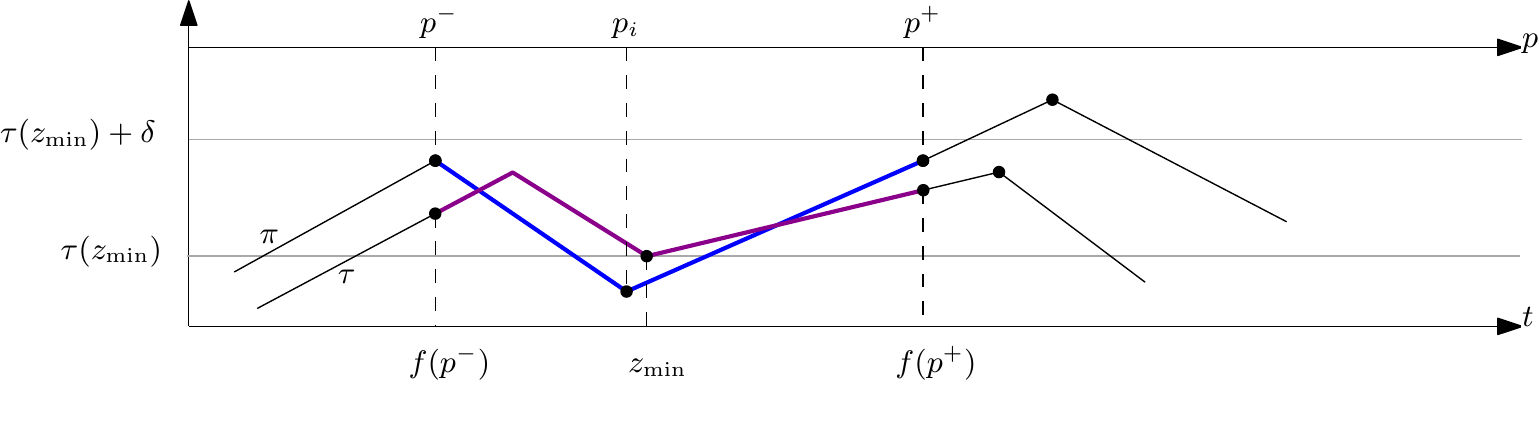}\\
\caption{Example of \caseref{trivial}. The broken part of the matching $f$ is indicated by fat lines.}
\figlab{shortcutting:case1}
\end{figure}
\begin{proof}[Proof of \caseref{trivial}]
As a warm-up exercise we quickly check that the above case is indeed trivial. In
this case, we would simply match $\pi(p^{-})$ to the subcurve
$\tau[f(p^{-}),f(p^{+})]$ and the remaining subcurves $\pi[0,p^{-}]$ and
$\pi[p^{+},1]$ can be matched as done by $f$.\footnote{Note that the constructed
matching is not a bijection. However, for any $\eps>0$, it can be perturbed to obtain a 
proper bijection.} Indeed, 
\begin{equation*}
\setSubC{\tau}{t}{f(p^{-})}{f(p^{+})} \subseteq \range{\pi(p^{-})}{\delta},
\end{equation*}
since by the case distinction 
\begin{equation*}
\minSubC{\tau}{t}{f(p^{-})}{f(p^{+})} = \tau(z_{\min})   \geq \pi(p^{-}) - \delta
\end{equation*}
and by \propref{frechet}, 
\begin{equation*}
\maxSubC{\tau}{t}{f(p^{-})}{f(p^{+})}  \leq  \maxSubC{\pi}{p}{p^{-}}{p^{+}} +\delta = \pi(p^{-}) +\delta.
\end{equation*}
\end{proof}

\begin{assumption}
We assume in the rest of the proof that $\tau(z_{\min}) < \pi(p^{-}) - \delta$ (non-trivial case).
\asslab{nontrivialcase}
\end{assumption}

%From now on we assume the non-trivial case.
Intuitively, we want to extend the subcurves of the
trivial case in order to fix the broken matching. The difficulty lies in finding
suitable subcurves which cover the broken part $\tau[f(p^{-}),f(p^{+})]$ and
whose \Frechet distance is at most $\delta$. Furthermore, the endpoints need to
line up suitably such that we can re-use $f$ for the suffix and prefix curves.

%We have the following useful claims which will be used throughout the proof.
The next two claims settle the question, to which extent signature vertices can be
included in the subcurve $\tau[f(p^{-}),f(p^{+})]$ for which we need to fix the broken matching. 

\begin{claim}
If $s_{j+1} \in [f(p^{-}),f(p^{+})]$ then
$\setSubC{\tau}{t}{s_{j+1}}{f(p^{+})} \subseteq [\tau(s_{j+1})-2\delta,
\tau(s_{j+1})]$.
\claimlab{descent:after:sj1}
\end{claim}
\begin{proof}
%New Proof.
We have to prove that
\[\min(\tau[s_{j+1},f(p^+)])\geq \tau(s_{j+1})-2\delta \text{ and } \max(\tau[s_{j+1},f(p^+)])\leq \tau(s_{j+1}).\]

The subcurve $\pi[p^-, p^+]$ consists of two edges $\overline{\pi(p^-)\pi(p_i)}$ and $\overline{\pi(p_i)\pi(p^+)}$  and $\pi(p_i)$ is the minimum of the subcurve. For the lower bound we distinguish two cases: $p_i\leq f^{-1}(s_{j+1})\leq p^+$ and $p^-\leq f^{-1}(s_{j+1}) < p_i$. 

If $p_i\leq f^{-1}(s_{j+1})\leq p^+$, then by \propref{frechet}
\[\min(\tau[s_{j+1},f(p^+)])\geq \min(\pi[f^{-1}(s_{j+1}), p^+])-\delta = \pi(f^{-1}(s_{j+1})) -\delta \geq \tau(s_{j+1})-2\delta.\]

If $p^-\leq f^{-1}(s_{j+1}) < p_i$, then since $z_{\min}<s_{j+1}$ it holds that
\[\tau(z_{\min}) = \min(\tau[f(p^-),s_{j+1}]) \geq \min(\pi[p^-,f^{-1}(s_{j+1})]) - \delta = \pi(f^{-1}(s_{j+1})) - \delta \geq \tau(s_{j+1})-2\delta.\]
It follows that 
\[\min(\tau[s_{j+1},f(p^+)])\geq \min(\tau[f(p^-), f(p^+)]) =\tau(z_{\min}) \geq \tau(s_{j+1})-2\delta\] 
as claimed.

%We prove first that $s_{j+2}\notin [f(p^{-}), f(p^{+})]$. If we assume the opposite, then two signature vertices $\tau(s_{j+1})$ and $\tau(s_{j+2})$ are in $\tau[f(p^-),f(p^+)]$. By \propref{signature}(i) it holds that $\tau(s_{j+1})-\tau(s_{j+2})>2\delta$ and $\tau(s_{j+2})$ is a local minimum on $\tau[s_{j+1},s_{j+3}]$. But the local minimum on $\tau[f(p^{-}), f(p^{+})]$ is $\tau(z_{\min})$ and it holds that $z_{\min}<s_{j+1}<s_{j+2}$. Therefore it is $\tau(z_{\min})<\tau(s_{j+2})$ and $\tau(s_{j+1})-\tau(z_{\min})>2\delta$. So it holds that $\range{\tau(z_{\min})}{\delta}\cap \range{\tau(s_{j+1})}{\delta}=\emptyset$ and $\range{\tau(s_{j+1})}{\delta} \cap   \range{\tau(s_{j+2})}{\delta} =\emptyset$.
%
%The curve $\pi[p^-,p^+]$ has a minimum $\pi[p_i]$ and by \propref{frechet}(ii) it has to be in $\range{\tau(z_{\min})}{\delta}$. The curve $\pi[p_i,p^+]$ has a maximum in $\pi(p^+)$ and by \propref{frechet}(iii) it has to be in $\range{\tau(s_{j+1})}{\delta}$. Then there cannot exist a point on $\pi[f^{-1}(s_{j+1}),p^+]$ that can be matched to a point in $\range{\tau(s_{j+2})}{\delta}$. Thus $s_{j+2}$ cannot be in $[f(p^{-}), f(p^{+})]$.
%

Furthermore, by \propref{signature}(\ref{item:sig:length:2}) it follows that
\[\min(\tau[s_{j+1},f(p^+)])\geq \tau(s_{j+1})-2\delta > \tau(s_{j+2}),\]
and therefore   $s_{j+2}\notin [f(p^{-}), f(p^{+})]$.

Now, \propref{signature}(\ref{item:sig:sj1:max}) implies the upper bound:
$ \maxSubC{\tau}{t}{s_{j+1}}{f(p^{+})} \leq \tau(s_{j+1}).$

\end{proof}

\begin{claim}
%If $\tau[f(p^{-}),f(p^{+})]$ contains a signature vertex $\tau(s_k)$, then 
%$\tau(s_k)$ must be a local maximum on $\tau$.
%For any $s_{k} \in [f(p^{-}),f(p^{+})]$, where $\tau(s_k)$ is a signature vertex, 
%it holds that $s_k=s_{j+1}$.
$s_j \notin [f(p^{-}),f(p^{+})]$
\claimlab{no:minimum}
\end{claim}
\begin{proof}
%We first prove that $s_j \notin [f(p^{-}),f(p^{+})]$. 
%Assume the opposite for the sake of contradiction. 
For the sake of contradiction, assume the claim is false, i.e. $s_j \in [f(p^{-}),f(p^{+})]$.
We have (by definition)
\[z_{\min} \in [f(p^{-}),f(p^{+})] \cap [s_j,s_{j+1}]\]
Furthermore, by definition 
$\tau(z_{\min}) = \minSubC{\tau}{t}{f(p^{-})}{f(p^{+})}$, and
by \propref{signature}(\ref{item:sig:sj:min}),
we have
$\tau(s_j) = \minSubC{\tau}{t}{s_{j-1}}{s_{j+1}}$.
This would imply that 
\[ \tau(z_{\min}) = \min \brc{\tau(t)~|~ t \in [f(p^{-}),f(p^{+})] \cap
[s_j,s_{j+1}]}  =
\tau(s_j).\]
By the theorem statement
$\pi(p_i) \notin \range{\tau(s_j)}{\delta} = \range{\tau(z_{\min})}{\delta}$. 
However, by \propref{frechet},
\[\pi(p_i) = \minSubC{\pi}{p}{p^{-}}{p^{+}} \in \range{
\minSubC{\tau}{t}{f(p^{-})}{f(p^{+})}}{\delta}=\range{\tau(z_{\min})}{\delta}. \]

%OLD PART OF PROOF:
%Assume $z_{\min} > f(p^{-})$, thus $\tau(z_{\min})$ would be a local minimum on $\tau$.
%Being a local minimum, both $\tau(z_{\min})$ and $\tau(s_j)$ must be vertices.
%Therefore, by our general position assumption, this would imply that
%$z_{\min}=s_j$.
%By the theorem statement
%$\pi(p_i) \notin \range{\tau(s_j)}{\delta} = \range{\tau(z_{\min})}{\delta}$. 
%However, by \propref{frechet},
%\[\pi(p_i) = \minSubC{\pi}{p}{p^{-}}{p^{+}} \in \range{
%\minSubC{\tau}{t}{f(p^{-})}{f(p^{+})}}{\delta}=\range{\tau(z_{\min})}{\delta}. \]
%
%Now, assume $z_{\min} = f(p^{-})$. In this case, \propref{frechet} would imply that 
%we are in the trivial case, which is excluded by \assref{nontrivialcase}.
\end{proof}

%
%\begin{claim}
%$\pi(p_i) \in [s_j+\delta, s_{j+1}-\delta]$ 
%\claimlab{pi:inside:signature:edge}
%\end{claim}
%\begin{proof}
%\claimref{no:minimum} implies that if $\tau[f(p^{-}), f(p^{+})]$ contains a
%signature vertex, then it must be a local maximum. In our case, the only contained
%signature vertex could be $\tau(s_{j+1})$. We have
%\[ \tau(f(p_i)) \in [\tau(z_{\min}), \tau(s_{j+1})] \subseteq [\tau(s_j), \tau(s_{j+1})]. \]
%Therefore by \propref{frechet},
%$\pi(p_i) \in [\tau(s_j)-\delta, \tau(s_{j+1})+\delta]$. By the lemma statement  
%$\pi(p_i) \notin \range{\tau(s_j)}{\delta} \cup \range{\tau(s_{j+1})}{\delta}.$
%\end{proof}

We now introduce some more notation which will be used throughout the proof.
\begin{eqnarray*} 
t_{\min} &=& \argmin_{t \in [f(p^{-}),s_{j+1}] } \tau(t)\\
x &=& \max\{p \in [0,p^{-}] ~|~ \pi(p) = \min(\tau(t_{\min}) + \delta,
\tau(s_{j+1}) - \delta)\}\\
p_{\max}&=& \argmax_{p \in [x, p^{-}] } \pi(p)\\
y &=& \min\{t \in [t_{\min}, 1] ~|~ \tau(t) = \pi(p_{\max}) - \delta\}
\end{eqnarray*}

In the next few claims we argue that these variables are well-defined. In
particular, that $x$ and $y$ always exist in the non-trivial case
(\claimref{x:exists} and \claimref{y:exists}).  Clearly $t_{\min}$ is
well-defined and by our initial assumptions we have $z_{\min} \leq t_{\min}$
(since $z_{\min} \in [s_j,s_{j+1}]$). We also derive some bounds along the way,
which will be used throughout the later parts of the proof.

\begin{claim}[Existence of $x$]
It holds that
\begin{compactenum}[(i)]
\item $\min(\tau(t_{\min}) + \delta, \tau(s_{j+1})-\delta) \in \{\pi(p) ~|~ p
\in [f^{-1}(s_j),p^{-}]\}$
\item $\minSubC{\pi}{p}{x}{p^{-}} \geq \min(\tau(t_{\min}) + \delta, \tau(s_{j+1})-\delta) = \pi(x)$
\item $\tau(s_j) < \tau(t_{\min})$
\end{compactenum}
\claimlab{x:exists}
\claimlab{x:high}
\claimlab{sj:low}
\end{claim}
\begin{proof}
We first prove part (i) of the claim.
We show that there exist two parameters
$f^{-1}(s_j) \leq p_1 \leq  p_2 \leq p^{-}$ such that 
\[\pi(p_1) \leq \min\pth{\tau(t_{\min})+\delta, \tau(s_{j+1})-\delta} \leq \pi(p_2).\]
Since the curve is continuous, this would imply the claim.
Indeed, we can choose $p_1=f^{-1}(s_j)$ and  $p_2=p^{-}$.
If $s_{j+1}\geq f(p^{+})$ we have 
\begin{equation*}
\pi(p_2)=\pi(p^{-}) \geq \tau(z_{\min}) + \delta \geq \tau(t_{\min}) + \delta, 
%\label{eq:x:exists}
\end{equation*}
(since we assume the non-trivial case). 
Otherwise, if  $s_{j+1} < f(p^{+})$ and therefore $\tau(s_{j+1})=\maxSubC{\tau}{t}{f(p^{-})}{f(p^{+})}$, then by \propref{signature} and
\propref{frechet},
\[\pi(p_2)=\pi(p^{-}) = \maxSubC{\pi}{p}{p^{-}}{p^{+}} \geq
\maxSubC{\tau}{t}{f(p^{-})}{f(p^{+})} - \delta = \tau(s_{j+1}) -\delta. \]
Thus, in both cases, it holds that 
$\pi(p_2) \geq \min\pth{\tau(t_{\min})+\delta, \tau(s_{j+1})-\delta}$.

As for $p_1$, by \claimref{no:minimum} we have $0\leq s_j \leq f(p^{-}) \leq t_{\min} \leq s_{j+1}$ and 
by \propref{frechet} 
\[\pi(p_1) =\pi(f^{-1}(s_j)) \leq \tau(s_j) +\delta .\]
It follows by \propref{signature}(\ref{item:sig:sj:min}) that 
$\pi(p_1) < \tau(t_{\min}) + \delta$
and by \propref{signature}(\ref{item:sig:length}) that 
$\pi(p_1)  < \tau(s_{j+1})-\delta.$
Now, part (ii) of the claim follows directly from the above,
since $\pi(x)$ is defined as the last point along the prefix subcurve 
$\pi[0,p^{-}]$ with the specified value and since $\pi(x) \leq \pi(p^{-})$.
Note that part (iii) we indirectly proved above. 
%\Amer{Doesn't the (iii) directly follow from Property 1(ii) and general position assumption?}
\end{proof}

\begin{claim}[Existence of $y$] It holds that
\begin{compactenum}[(i)]
\item $ \pi(p_{\max}) - \delta  \in \{\tau(t) ~|~ t \in [t_{\min},s_{j+1}]\}$
\item $\maxSubC{\tau}{t}{t_{\min}}{y} \leq \pi(p_{\max}) - \delta =\tau(y)$
\item $\pi(p_{\max}) < \tau(s_{j+1}) + \delta$
%\item $y \leq s_{j+1}$
\end{compactenum}
\claimlab{y:exists}
\claimlab{sj1:high}
\end{claim}
\begin{proof}
To prove part (i) of the claim we show that there exist two parameters
$t_{\min} \leq t_1 \leq t_2 \leq s_{j+1}$,
such that 
\[\tau(t_1)\leq  \pi(p_{\max})-\delta \leq \tau(t_2). \]
We choose $t_1=t_{\min}$ and $t_2=s_{j+1}$.
Since we have the non-trivial case, we know 
\[ \pi(p_{\max}) \geq \pi(p^{-}) \geq \tau(z_{\min})+\delta \geq \tau(t_{\min})+\delta = \tau(t_1)+\delta.\]
Now, for $t_2$, we know that $s_j \leq f(x) \leq f(p_{\max}) \leq f(p^{-}) \leq s_{j+1}$.
By \propref{frechet} and by \propref{signature}(\ref{item:sig:sj1:max})
\[\tau(t_2) = \tau(s_{j+1}) \geq \tau(f(p_{\max})) \geq \pi(p_{\max})-\delta. \] 
Since the subcurve is continuous, there must be a parameter $t_1\leq t\leq
t_2$ which satisfies the claim.
Now, part (ii) of the claim also follows directly, since $\tau(y)$ is the first point along the suffix subcurve
$\tau[t_{\min},1]$ with the specified value and since $\tau(y) \geq \tau(t_{\min})$.
Note that part (iii) we just proved above.
\end{proof}

The following claim follows directly from \claimref{x:exists} and \claimref{y:exists}.
\begin{claim}
$s_j \leq f(x) \leq y \leq s_{j+1}$
\claimlab{xy:inside}
\end{claim}

The following claim will be used throughout the proof.
\begin{claim} 
%$\pi(p_{\max}) - 2\delta \leq \min(\tau(t_{\min}) + \delta, \tau(s_{j+1})-\delta) = \pi(x)$
$\pi(p_{\max}) - 2\delta \leq \pi(x)$
\claimlab{not:so:bad}
\claimlab{max:dist}
\end{claim}

\begin{proof}
We need to show that 
\[ \pi(p_{\max}) - 2 \delta \leq \min(\tau(t_{\min}) + \delta, \tau(s_{j+1})-\delta)\]
\claimref{y:exists} immediately implies $\pi(p_{\max}) \leq \tau(s_{j+1}) + \delta$.
On the other hand, by \claimref{xy:inside}, 
\[s_j \leq f(x) \leq f(p_{\max}) \leq f(p^{-}) \leq t_{\min} \leq s_{j+1}.\]
By \propref{frechet} and by \propref{signature}(\ref{item:sig:descent}),
\[\pi(p_{\max}) \leq \tau(f(p_{\max})) + \delta \leq \tau(t_{\min}) + 3\delta \]
\end{proof}

The next two claims (\claimref{frechet:y} and
\claimref{frechet:x}) show that our choice of $x$ and $y$ is suitable for fixing some part
of the broken matching: the subcurve $\pi[x,p^{-}]$ can be matched
entirely to $\tau(y)$ and the subcurve $\tau[f(x),y]$ can be matched entirely to
$\pi(x)$. After that, it remains to match the subcurve $\pi[p^{+}, f^{-1}(y)]$. For
this we have the case analysis that follows.

\begin{claim}
$\setSubC{\pi}{p}{x}{p^{-}} \subseteq  [\pi(p_{\max}) -2\delta, \pi(p_{\max})] = \range{\tau(y)}{\delta}.$
\claimlab{frechet:y}
\claimlab{frechet:ya}
\end{claim}
\begin{proof}
By \claimref{x:high} and \claimref{not:so:bad},
\[\minSubC{\pi}{p}{x}{p^{-}} \geq \min\pth{\tau(t_{\min})+\delta, \tau(s_{j+1})-\delta} \geq \pi(p_{\max})-2\delta.\]
On the other hand, by definition of $p_{\max}$, we have
$\maxSubC{\pi}{p}{x}{p^{-}} = \pi(p_{\max}).$
The last equality of the claim follows directly from the definition of $y$ and from \claimref{y:exists} ($y$ is well-defined).
\end{proof}

\begin{claim}
$\setSubC{\tau}{t}{f(x)}{y} \subseteq 
%[\tau(t_{\min}), \tau(t_{\min})+2\delta] 
\range{\min\pth{\tau(t_{\min}) + \delta, \tau(s_{j+1})-\delta}}{\delta}
= \range{\pi(x)}{\delta}.$
\claimlab{frechet:x}
\end{claim}
\begin{proof}
We first prove the lower bound on the minimum of the subcurve $\tau[f(x),y]$.
By \claimref{x:high}, and by \propref{frechet}, we have
\[\minSubC{\tau}{t}{f(x)}{f(p^{-})} \geq \minSubC{\pi}{p}{x}{p^{-}} - \delta \geq \min\pth{\tau(t_{\min}), \tau(s_{j+1})-2\delta}.\] 
By definition, $\tau(t_{\min})$ is a minimum on $\tau[f(p^{-}), s_{j+1}]$, thus
\[\minSubC{\tau}{t}{f(p^{-})}{y} \geq \tau(t_{\min}) \geq
\min\pth{\tau(t_{\min}), \tau(s_{j+1})-2\delta},\]
for $y \leq s_{j+1}$, which is ensured by \claimref{xy:inside}. 

We now prove the upper bound on the maximum of the subcurve $\tau[f(x),y]$.
Since by \claimref{sj1:high} $s_j \leq f(x) \leq t_{\min} \leq s_{j+1}$ and since by 
\propref{signature}(\ref{item:sig:descent}), $\tau[s_j, s_{j+1}]$ may
not descend by more than $2\delta$, it follows that 
\[\maxSubC{\tau}{t}{f(x)}{t_{\min}} \leq \tau(t_{\min})+2\delta.\]
By definitions of $x$ and $y$ and by \claimref{max:dist}
\[\maxSubC{\tau}{t}{t_{\min}}{y} \leq \pi(p_{\max})-\delta \leq \pi(x) +\delta \leq \tau(t_{\min}) + 2\delta.\]
By \propref{signature}(\ref{item:sig:sj1:max}) we also have, 
\[\maxSubC{\tau}{t}{f(x)}{y} \leq \tau(s_{j+1}),\]
for $y \leq s_{j+1}$, which is ensured by \claimref{xy:inside}. 

Together this implies the claim. The last equality of the claim follows directly from the
definition of $x$ and from \claimref{x:exists} ($x$ is well-defined).
\end{proof}

Now we have established the basic setup for our proof. In the following, we
describe the case analysis based on the structure of the two curves $\tau$ and
$\pi$.  Consider walking along the
subcurve $\pi[p^{+},1]$. At the beginning of the subcurve, we have $\pi(p^{+})
\in [\pi(x), \pi(p_{\max})]$.  One of the following events may happen
during the walk: either we go above $\pi(p_{\max})$, or we go below 
$\pi(x)$, or we stay inside this interval. Let $q$ denote
the time at which for the first time one of these events happens. Formally,
we define the intersection function $g:\Re\rightarrow [p^{+},1] \cup \brc{p_{\infty}}$,
\begin{eqnarray*}
g(h) &=&  \min (\{p \in [p^{+}, 1] ~|~ \pi(p)=h \} \cup \{p_{\infty}\}) \\
%q &=& \min\pth{g(\pi(p_{\max})), g(\tau(t_{\min})+\delta)},
q &=& \min\pth{g(\pi(p_{\max})), g(\pi(x))},
%q^{(p_{\max})} &=& \min (\{p \in [p^{+}, 1] ~|~ \pi(p)=\pi(p_{\max}) \} \cup \{1\})\\
%q^{(t_{\min})} &=& \max(\{p \in [p^{+}, 1] ~|~ \pi(p)=\tau(t_{\min})+\delta \}\cup \{1\})\\
%q &=& \min(q^{(p_{\max})}, q^{(t_{\min})})
\end{eqnarray*}
where $p_{\infty}>1$ is some fixed constant for the case that the suffix curve
$\pi[p^{+},1]$ does not contain the value $h$.
We distinguish the following main cases. In each of the cases, we devise a
matching scheme to fix the broken matching. For each case, our construction
ensures that the extended subcurves cover the subcurve $\tau[f(p^{-}),f(p^{+})]$
and that the subcurves line up with suitable prefix and suffix curves, such that
we can always use $f$ for the parts of $\pi$ and $\tau$ not covered in the matching
scheme. We need to prove in each case that the \Frechet distance between the
specified subcurves is at most $\delta$. If this is the case, we call the
matching scheme \emph{valid}.

We have to make further distinction between the case when $f(p^{+}) \leq y$ and the case $f(p^{+}) > y$. If $f(p^{+}) \leq y$ holds, the three aforementioned events are described by \caseref{level}, \caseref{upwards} and \caseref{downwards}. If it happens that $f(p^{+}) > y$, it becomes more complicated to repair the matching. This is discussed in \caseref{matroska}.

\begin{case}[$\pi$ stays level]
$p^{+} \leq f^{-1}(y) \leq q$.
\caselab{level}
\end{case}
\begin{figure}[h]\centering
\includegraphics[page=2]{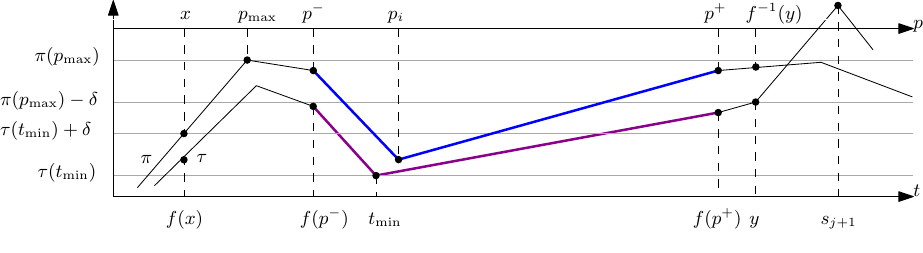}\\
\caption{Example of \caseref{level}. The broken part of the matching $f$ is indicated by fat lines.}
\figlab{shortcutting:case2}
\end{figure}
\caseref{level} is the simplest case. We intend to use the following matching scheme:
\begin{eqnarray*}
\pi(x) &\Leftrightarrow& \tau[f(x),y]\\
\pi[x,p^{-}] &\Leftrightarrow& \tau(y)\\
\pi[p^{+},f^{-1}(y)] &\Leftrightarrow& \tau(y)
\end{eqnarray*}

\begin{proof}[Proof of \caseref{level}]
\claimref{frechet:x} implies that the \Frechet distance between $\tau[f(x),y]$
and $\pi(x)$ is at most $\delta$.
\claimref{frechet:y} implies that the \Frechet distance between $\pi[x,p^{-}]$
and $\tau(y)$ is at most $\delta$.
Finally, by our case distinction and by \claimref{not:so:bad}
\[\setSubC{\pi}{p}{p^{+}}{f^{-1}(y)} 
\subseteq [\pi(x), \pi(p_{\max})] 
\subseteq [\pi(p_{\max})-2\delta, \pi(p_{\max})] 
= \range{\tau(y)}{\delta}.\]
Therefore, also the \Frechet distance between $\pi[p^{+},f^{-1}(y)]$ and $\tau(y)$ is at most
$\delta$.
\end{proof}

\begin{case}[$\pi$ tends upwards]
$q < f^{-1}(y)$ and $q=g(\pi(p_{\max}))$
\caselab{upwards}
\end{case}
\begin{figure}[h]\centering
\includegraphics[page=3]{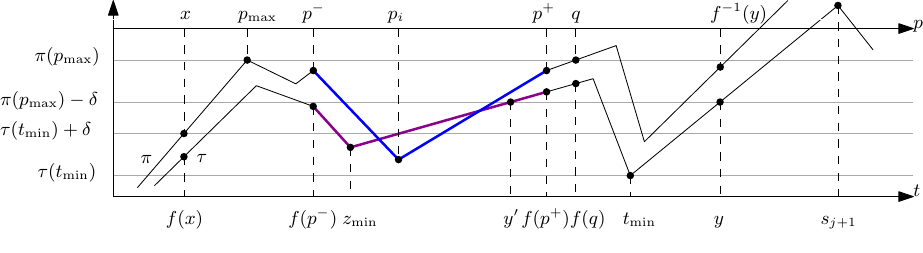}\\
\caption{Example of \caseref{upwards}. The broken part of the matching $f$ is indicated by fat lines.}
\figlab{shortcutting:case3}
\end{figure}
 In \caseref{upwards}, let $y'=\max\{t\in [0,f(q)]~ |~
 \tau(t)=\tau(y) \}$ and $z=\max\{p^{+},f^{-1}(y')\}$. We intend to use
 the following matching scheme:
\begin{eqnarray*}
\pi(x) &\Leftrightarrow& \tau[f(x),y']\\
\pi[x,p^{-}] &\Leftrightarrow& \tau(y')\\
\pi[p^{+},z] &\Leftrightarrow& \tau(y')\\
\pi(z) &\Leftrightarrow& \tau[y',f(p^{+})]
\end{eqnarray*}
(Note that if $y' > f(p^{+})$, then the last line of the above matching scheme is simply dropped.) 

\begin{proof}[Proof of \caseref{upwards}]
We first argue that $y'$ exists. To this end, we show that there exist
two parameters $0\leq t_1 < t_2 \leq  f(q)$, such that
\[ \tau(t_1) \leq \tau(y) = \pi(p_{\max})-\delta \leq \tau(t_2). \]
We choose $t_1=z_{\min}$ and $t_2=f(q)$.
Note that $z_{\min} \leq f(p^{+}) \leq f(q)$.
Now, by \propref{frechet},
\[\tau(t_2) = \tau(f(q)) \geq \pi(q)-\delta = \pi(p_{\max})-\delta.\] 
Since we are assuming the non-trivial case, 
\[ \pi(p_{\max}) \geq \pi(p^{-}) \geq \tau(z_{\min}) + \delta = \tau(t_1) + \delta. \] 
Thus, since $\tau[0,f(q)]$ is continuous, $y'$ must exist and it holds that $f(p^{-}) \leq z_{\min} \leq y'$.
It remains to prove that the matching scheme is valid.
Since $y'\leq f(q) < y$,
\claimref{frechet:x} implies that the \Frechet distance between $\tau[f(x),y']$
and $\pi(x)$ is at most $\delta$.
\claimref{frechet:y} implies that the \Frechet distance between $\pi[x,p^{-}]$
and $\tau(y')$ is at most $\delta$.
For the last two lines of the matching scheme we distinguish two cases.
If $y' > f(p^{+})$, then $z=f^{-1}(y')$ and we need to prove that
\[ \setSubC{\pi}{p}{p^{+}}{f^{-1}(y')} \subseteq \range{\tau(y')}{\delta}. \]
By our case distinction and by \claimref{not:so:bad}
\[ \setSubC{\pi}{p}{p^{+}}{q} \subseteq [\pi(x), \pi(p_{\max})] 
\subseteq [\pi(p_{\max})-2\delta, \pi(p_{\max})] = \range{\tau(y')}{\delta}. \]
Since $f^{-1}(y') \leq q$, this implies the validity of the matching.
Otherwise, if $y' \leq f(p^{+})$, then $z=p^{+}$ and we need to prove 
that 
\[ \setSubC{\tau}{t}{y'}{f(p^{+})} \subseteq \range{\pi(p^{+})}{\delta}. \]
This can be derived as follows. On the one hand, by \propref{frechet}, since $y'\in [f(p^{-}),f(p^{+})]$
\[\maxSubC{\tau}{t}{y'}{f(p^{+})} 
\leq \maxSubC{\pi}{p}{f^{-1}(y')}{p^{+}}+\delta 
= \maxSubC{\pi}{p}{p^{-}}{p^{+}}+\delta 
= \pi(p^{+})+\delta. \]
On the other hand, by the definition of $y'$ and since $y' \leq f(p^{+}) \leq f(q)$
\[\minSubC{\tau}{t}{y'}{f(p^{+})} = \tau(y') = \pi(p_{\max})-\delta \geq \pi(p^{+}) -\delta. \]
\end{proof}

\begin{case}[$\pi$ tends downwards]
$q < f^{-1}(y)$ and $q=g(\pi(x))$.
\caselab{downwards}
\end{case}
\begin{figure}[h]\centering
\includegraphics[page=4]{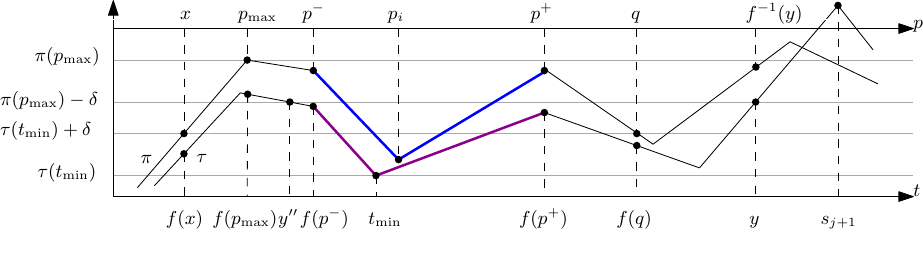}\\
\caption{Example of \caseref{downwards}. The broken part of the matching $f$ is indicated by fat lines.}
\figlab{shortcutting:case4}
\end{figure}
In \caseref{downwards}, let $y''=\min\{ t \in [f(p_{\max}),1] ~|~ \tau(t) =
\tau(y) \}$. 
We intend to use the following matching scheme:
\begin{eqnarray*}
\pi(p_{\max}) &\Leftrightarrow& \tau[f(p_{\max}), y'']\\
\pi[p_{\max},p^{-}] &\Leftrightarrow& \tau(y'')\\
\pi[p^{+},q] &\Leftrightarrow& \tau(y'')\\
\pi(q) &\Leftrightarrow& \tau[y'',f(q)]
\end{eqnarray*}

\begin{proof}[Proof of \caseref{downwards}]
Clearly, $y''$ exists in the non-trivial case, since  
\[ \tau(f(p_{\max})) \geq \tau(y) \geq \tau(z_{\min}),\]
and $f(p_{\max}) \leq f(p^{-}) \leq z_{\min}$.

We prove the validity of the matching scheme line by line.
Note that by definition $\tau(y'') = \tau(y) = \pi(p_{\max})-\delta$.
For the first matching:
by the definition of $y''$ and by \propref{frechet},
\[\setSubC{\tau}{t}{f(p_{\max})}{y''} \subseteq \range{\pi(p_{\max})}{\delta}. \]
The validity of the second matching follows from \claimref{frechet:y} since $p_{\max} \geq x$.
For the third matching:
By our case distinction and by \claimref{not:so:bad}
\[\setSubC{\pi}{t}{p^{+}}{q} 
\subseteq [\pi(x), \pi(p_{\max})]
\subseteq [\pi(p_{\max})-2\delta, \pi(p_{\max})]
=\range{\tau(y'')}{\delta}.\]
As for the last matching, since $f(x) \leq f(p_{\max}) \leq y''$ and since by our case distinction
$f(q) < y$, \claimref{frechet:x} implies
\[ \setSubC{\tau}{t}{y''}{f(q)} \subseteq  \range{\pi(x)}{\delta} = \range{\pi(q)}{\delta}.\]
\end{proof}

%\begin{figure}\centering
%\includegraphics[page=1]{figs/shortcutting_case1234}\\
%\vspace{\baselineskip}
%\includegraphics[page=2]{figs/shortcutting_case1234}\\
%\vspace{\baselineskip}
%\includegraphics[page=3]{figs/shortcutting_case1234}\\
%\vspace{\baselineskip}
%\includegraphics[page=4]{figs/shortcutting_case1234}\\
%\caption{Examples of \caseref{trivial}-\caseref{downwards}. The broken part of the matching
%$f$ is indicated by fat lines.}
%\figlab{shortcutting:case1234}
%\end{figure}

%Now, if indeed $f(p^{+}) \leq y$, then we would be in one of the previous cases.
%Examples of these cases are illustrated in \figref{shortcutting:case1234}.
%However, it may also happen that $f(p^{+}) > y$. We cover this case next.

\begin{case}[The matroska case] $f(p^{+}) > y$.
\caselab{matroska}
\end{case}
\caseref{matroska} seems to be the most difficult case to handle. However, we have
already established a suitable set of tools in the previous cases. We devise an iterative
matching scheme and prove an invariant (\claimref{frechet:xa}) to verify that the \Frechet distance of
the subcurves is at most $\delta$.
We first define $z_{\min}^{(1)}=z_{\min}$, $t_{\min}^{(1)}=t_{\min}$, 
$x^{(1)}=x$, and $y^{(1)}=y$.
Now, for $a = 2, \dots$ let 
\begin{eqnarray*}
z_{\min}^{(a)} &=& \argmin_{t\in [y^{(a-1)},f(p^{+})] } \tau(t)\\
t_{\min}^{(a)} &=& \argmin_{t\in [y^{(a-1)},s_{j+1}]} \tau(t)\\
x^{(a)}&=&  \min\brc{p\in [x^{(a-1)},p_{\max}] ~|~ \pi(p) = \min\pth{\tau(t_{\min}^{(a)})+\delta,  \tau(s_{j+1})-\delta} } \\ 
%x^{(a)}&=&  \max\brc{p\in [0,p^{-}] ~|~ \pi(p) = \min\pth{\tau(t_{\min}^{(a)})+\delta,  \tau(s_{j+1})-\delta} } \\ 
%y^{(a)}&=& \min \{t \in [t_{\min}^{(a)},1] ~|~ \tau(t) = \pi(p_{\max}) - \delta\},
y^{(a)}&=& \min \{t \in [t_{\min}^{(a)},s_{j+1}] ~|~ \tau(t) = \pi(p_{\max}) - \delta\},
\end{eqnarray*}
We describe the intended matching scheme, beginning with the following subcurves:
\begin{eqnarray*}
\pi(x^{(1)}) &\Leftrightarrow& \tau[f(x^{(1)}),y^{(1)}]\\
\pi[x^{(a-1)},x^{(a)}] &\Leftrightarrow& \tau(y^{(a-1)})\\
\pi(x^{(a)}) &\Leftrightarrow& \tau[y^{(a-1)},y^{(a)}],
\end{eqnarray*}
where the last two matchings are repeated while incrementing $a$ (starting with
$a=2$).   After each iteration, we are left with the
unmatched subcurves $\pi[x^{(a)}, p^{-}]$ and $\tau[y^{(a)}, f(p^{+})]$.  
We would like to complete the matching with the following scheme
\begin{eqnarray*}
\pi[x^{(a)},p^{-}] &\Leftrightarrow& \tau(y^{(a)})\\
\pi(p^{+}) &\Leftrightarrow& \tau[y^{(a)},f(p^{+})]
\end{eqnarray*}

This is indeed possible if 
%\Amer{the equation below contained one more $\pi(p^{-})$, which does not make sense}
%\[\maxSubC{\pi}{p}{x^{(a)}}{p^{-}} \pi(p^{-}) \leq \tau(z_{\min}^{(a+1)}) + \delta.\]
\[ \pi(p^{-}) \leq \tau(z_{\min}^{(a+1)}) + \delta.\]
The above is the equivalent to the trivial case (\caseref{trivial}).  
We first prove correctness in this case (\caseref{matroska}(\ref{item:trivial})). To this end, we extend
\claimref{frechet:x} as follows. Note that this claim will also be used in the non-trivial cases that follow.

%\begin{claim}
%For any two integers $a>1$ and $b\geq 1$, it holds that
%$\setSubC{\pi}{p}{x^{(a-1)}}{x^{(a)}} \subseteq   \range{\tau(y^{(b)})}{\delta}$
%and $x\leq x^{(a-1)}$.
%\claimlab{frechet:ya}
%\end{claim}
%\begin{proof}
%This follows directly from \claimref{frechet:y} since 
%$[x^{(a-1)}, x^{(a)}] \subseteq [x,p^{-}]$ and since $\tau(y^{(b)}) = \tau(y)$.
%\end{proof}
%
\begin{claim}
%For $y^{(a)} \leq s_{j+1}$, 
$\setSubC{\tau}{t}{y^{(a-1)}}{y^{(a)}} \subseteq \pbrc{\tau(t_{\min}^{(a)})
~,~  \min(\tau(t_{\min}^{(a)}) + 2\delta, \tau(s_{j+1})  )}
\subseteq \range{\pi(x^{(a)})}{\delta}.$
\claimlab{frechet:xa}
\end{claim}
\begin{proof}
Recall that $s_j \leq y \leq y^{(a-1)} \leq y^{(a)} \leq s_{j+1}$ by
\claimref{xy:inside}.
By definition of $t_{\min}^{(a)}$,
\[ \minSubC{\tau}{t}{y^{(a-1)}}{y^{(a)}} \geq \minSubC{\tau}{t}{y^{(a-1)}}{s_{j+1}} \geq \tau(t_{\min}^{(a)}). \]
By \propref{signature}(\ref{item:sig:descent}) and the definitions of $y^{(a)}$ and $t_{\min}^{(a)}$, we also have that
\[ \maxSubC{\tau}{t}{y^{(a-1)}}{y^{(a)}} \leq  \tau(t_{\min}^{(a)}) + 2\delta, \]
and by \propref{signature}(\ref{item:sig:sj1:max}), we also have that
\[ \maxSubC{\tau}{t}{y^{(a-1)}}{y^{(a)}} \leq \tau(s_{j+1}). \]
This proves the first part of the claim. For the second part we use the definition of $\pi(x^{(a)}) = \min(\tau(t_{\min}^{(a)})+\delta,
\tau(s_{j+1})-\delta)$, which implies
\begin{eqnarray*} 
\tau(t_{\min}^{(a)}) &\geq & \pi(x^{(a)}) - \delta\\
\min(\tau(t_{\min}^{(a)}) + 2\delta, \tau(s_{j+1})) &=& \pi(x^{(a)}) + \delta.
\end{eqnarray*}
\end{proof}

\begin{claim}[Correctness of \caseref{matroska}(\ref{item:trivial})]
If for some value of $a$, it holds that 
%$\maxSubC{\pi}{p}{x^{(a)}}{p^{-}} \leq \tau(z_{\min}^{(a+1)}) + \delta$ then
$\pi(p^{-}) \leq \tau(z_{\min}^{(a+1)}) + \delta$ then
the above matching scheme is valid.
\claimlab{trivial:subcase}
\end{claim}

\begin{proof}
By \claimref{frechet:x}, \claimref{frechet:ya} and \claimref{frechet:xa} the
iterative part of the matching scheme is valid.
It remains to prove the validity of the last two matchings.
By \claimref{frechet:y},
\[\setSubC{\pi}{p}{x}{p^{-}} \subseteq  [\pi(p_{\max}) -2\delta, \pi(p_{\max})] =
\range{\tau(y)}{\delta} = \range{\tau(y^{(a)})}{\delta}.\]
Since $x \leq x^{(a)} \leq p^{-}$, this implies that the \Frechet distance
between $\pi[x^{(a)},p^{-}]$ and $\tau(y^{(a)})$ is at most $\delta$.
As for the other matching, we have by our case distinction 
\[\pi(p^{-}) \leq %\maxSubC{\pi}{p}{x^{(a)}}{p^{-}} \leq 
\tau(z_{\min}^{(a+1)}) + \delta = \minSubC{\tau}{t}{y^{(a)}}{f(p^{+})} + \delta,\]
while (by \propref{frechet}) the matching $f$ testifies that
\[\maxSubC{\tau}{t}{f(p^{-})}{f(p^{+})} \leq
\maxSubC{\pi}{p}{p^{-}}{p^{+}} + \delta = \pi(p^{-}) + \delta.\]
Since $f(p^{+}) \geq y^{(a)} \geq y \geq f(p^{-})$, the above implies
\[\setSubC{\tau}{t}{y^{(a)}}{f(p^{+})} \subseteq \range{\pi(p^{-})}{\delta} = \range{\pi(p^{+})}{\delta}.\]
Note that the proof holds both if $s_{j+1} < f(p^{+})$ or $s_{j+1} \geq f(p^{+})$.
\end{proof}

From now on, we will assume the non-trivial (sub)case.  Our matching scheme is
based on a stopping parameter $\Astop$, which (intuitively) depends on whether
$f$ matched some point on the missing subcurve $\pi[p^{-},p^{+}]$ to a signature
vertex $\tau(s_{j+1})$ of $\tau$.

\begin{definition}(Stopping parameter $\Astop$)
If $s_{j+1} \geq f(p^{+})$ then let $\Astop$ be the minimal value of $a$
satisfying  $f(p^{+}) \leq y^{(a)}.$
%Otherwise, let $\Astop$ be the maximal value of $a$ satisfying $y^{(a)} \leq s_{j+1}.$
Otherwise, let $\Astop$ be the minimal value of $a$ such that $y^{(a)} = y^{(a+1)} \leq s_{j+1}.$
\deflab{a:stop}
\end{definition}

\begin{claim}
The stopping parameter $\Astop$ (\defref{a:stop}) is well-defined 
and the matching scheme is valid for $a\leq \Astop.$ 
\end{claim}
\begin{proof}
\newcommand{\Astall}{\widehat{a}}
We first argue that there must be a value of $a$ such that 
$t_{\min}^{(a+1)}=y^{(a)}=y^{(b)}$ for any $b>a$.
Recall that by our initial assumptions, $z_{\min} \in [s_j,s_{j+1}]$ and 
thus $z_{\min} \leq t_{\min} \leq s_{j+1}$. As a consequence, \claimref{y:exists} 
testifies that in the non-trivial case, the point $\tau(y)$ exists and is well-defined.
We defined $y^{(1)}=y$ and for $a>1$ we defined
\[y^{(a)}= \min \{t \in [t_{\min}^{(a)},s_{j+1}] ~|~ \tau(t) = \pi(p_{\max}) - \delta\}.\]
Since by \claimref{sj1:high}, $\tau(s_{j+1}) \geq \pi(p_{\max})-\delta = \tau(y^{(a)})$,
there must be a value of $a$ such that 
\[\minSubC{\tau}{t}{y^{(a)}}{s_{j+1}} \geq \tau(y^{(a)}).\]
Let this value of $a$ be denoted $\Astall$.
In this case, it follows by definition that $t_{\min}^{(\Astall+1)} = y^{(\Astall)}$, 
which implies that $y^{(\Astall+1)}=y^{(\Astall)}$ and $t_{\min}^{(\Astall+2)}=y^{(\Astall+1)}, \dots$.
This has the effect that $y^{(\Astall)}=y^{(b)}$ for any $b > \Astall$.

Now, if $s_{j+1} < f(p^{+})$, then the above analysis implies that $\Astop$ is well-defined.
However, if $s_{j+1} \geq f(p^{+})$, we defined $\Astop$ to be the minimal value of $a$
such that $f(p^{+}) \leq y^{(a)}$. Now it might happen that $y^{(\Astall)} \leq f(p^{+}) \leq s_{j+1}$.
In this case, there exists no value of $a$ such that $f(p^{+}) \leq y^{(a)}$, thus $y^{(\Astop)}$ does not exist.
We can reduce this case to the trivial case (\caseref{matroska}(\ref{item:trivial})) as follows. 
By \claimref{sj1:high}, $\tau(s_{j+1}) \geq \pi(p_{\max})-\delta$ and by
\propref{signature}(\ref{item:sig:sj1:max}), $\tau(s_{j+1})$ must be a maximum on $\tau[s_{j},
s_{j+2}]$. Thus, by definition of $t_{\min}^{(\Astop)}$, we would have 
$\tau\left(z_{\min}^{(\Astop)}\right) = \tau\left(t_{\min}^{(\Astop)}\right) = \tau(y^{(\Astop-1)}) = \pi(p_{\max})-\delta \geq \pi(p^-)-\delta$, which is
\caseref{matroska}(i) (the trivial case).
Thus, also in the non-trivial case, $\Astop$ is well-defined.

The validity of the matching scheme  for $a\leq \Astop$ follows from \claimref{frechet:x}, \claimref{frechet:ya} and \claimref{frechet:xa}.
\end{proof}

It follows that the iterative part of the matching scheme is valid for $a\leq
\Astop.$ 
Now we are left with the unmatched subcurves $\pi[x^{(\Astop)}, p^{-}]$ and
$\tau[y^{(\Astop)}, f(p^{+})]$ and we have to complete the matching scheme.
In order to set up a case analysis with a similar structure as before, we define 
\begin{eqnarray*}
%q^{(\Astop-1)} &=& \min\pth{\brc{ p \in [p^{+}, 1] ~|~
%\pi(p)=\tau(t_{\min}^{(\Astop-1)})+\delta} \cup \{1\} }\\
%r_{\high} &=& \min\pth{\brc{ p \in [p^{+}, 1]~|~ \pi(p) = \pi(p_{\max})} \cup
%\{1\}}\\
%r_{\low} &=& \min\brc{ p \in [p^{+}, 1]~|~ \pi(p) = \tau(s_{j+1}) -\delta}
q'&=& \min\pth{g(\pi(p_{\max})), g(\tau(t_{\min}^{(\Astop-1)})+\delta)}\\
q''&=& \min\pth{g(\pi(p_{\max})), g(\tau(s_{j+1})-\delta)}
\end{eqnarray*}

\begin{table}[h]
\center
\begin{tabular}{|c|m{7.5cm}|m{5cm}|}
\hline
case & definition & intended matching\\
\hline
\ref{case:matroska}(\ref{item:trivial}) & 
%$\exists a: \maxSubC{\pi}{p}{x^{(a)}}{p^{-}} \leq \tau(z_{\min}^{(a+1)}) + \delta$
$\exists a: \pi(p^{-}) \leq \tau(z_{\min}^{(a+1)}) + \delta$
&
$\pi[x^{(a)},p^{-}] \Leftrightarrow \tau(y^{(a)})$\newline
$\pi(p^{+}) \Leftrightarrow \tau[y^{(a)},f(p^{+})]$\\
\hline
\ref{case:matroska}(\ref{item:level}) & 
$p^{+} \leq f^{-1}(y^{(\Astop)}) \leq q'$
&
$\pi[x^{(\Astop)},p^{-}]  \Leftrightarrow \tau(y^{\Astop})$\newline
$\pi[ p^{+}, f^{-1}(y^{\Astop})] \Leftrightarrow \tau(y^{\Astop})$ \\
\hline
\ref{case:matroska}(\ref{item:upwards}) & 
$p^+\leq q' < f^{-1}(y^{(\Astop)})$ and \newline
$q'= g(\pi(p_{\max}))$ 
&
This case can be reduced\newline
to \caseref{matroska}(i) 
\\
\hline
\ref{case:matroska}(\ref{item:downwards}) & 
$p^+\leq  q' < f^{-1}(y^{(\Astop)})$ and\newline
$ q' = g(\tau(t_{\min}^{\Astop-1})+\delta)$ 
&
$\pi[x^{(\Astop-1)},p^{-}] \Leftrightarrow \tau(y^{(\Astop-1)})$\newline
$\pi[p^{+}, q'] \Leftrightarrow \tau(y^{(\Astop-1)})$\newline
$\pi(q') \Leftrightarrow \tau[y^{(\Astop-1)}, f(q')]$
\\
\hline
\ref{case:matroska}(\ref{item:upwards:sig}) &
$p^{+} > f^{-1}(y^{(\Astop)})$ and \newline
$q'' = g(\pi(p_{\max}))$\newline
For the matching scheme, let\newline 
$x'=\min\{t\in [x^{(\Astop)},p_{\max}]~|~ \tau(t)=\tau(s_{j+1})-\delta \}$ \newline
$y'=\max\{t\in [0,f(q)]~|~ \tau(t)=\tau(y) \}$ \newline
$z=\max\{p^{+},f^{-1}(y')\}$ 
& 
$\pi[x^{(\Astop)},x'] \Leftrightarrow \tau(y^{(\Astop)})$\newline
$\pi(x') \Leftrightarrow \tau[y^{(\Astop)},y']$\newline
$\pi[x',p^{-}] \Leftrightarrow \tau(y')$\newline
$\pi[p^{+}, z] \Leftrightarrow \tau(y')$\newline
$\pi[z] \Leftrightarrow \tau[y',f(p^{+})]$
\\
\hline
\ref{case:matroska}(\ref{item:downwards:sig}) & 
$p^{+} > f^{-1}(y^{(\Astop)})$ and \newline
$q'' = g(\tau(s_{j+1})-\delta)$
&
$\pi[x^{(\Astop)},p^{-}] \Leftrightarrow \tau(y^{(\Astop)})$\newline
$\pi[p^{+},q''] \Leftrightarrow \tau(y^{(\Astop)})$\newline
$\pi(q'') \Leftrightarrow \tau[y^{(\Astop)},f(q'')]$\\
\hline
\end{tabular}
\caption{Subcases for \caseref{matroska}:
(\ref{item:trivial}) trivial case (\claimref{trivial:subcase}),
(\ref{item:level}) $\pi$ stays level, 
(\ref{item:upwards}) $\pi$ tends upwards,
(\ref{item:downwards}) $\pi$ tends downwards,
(\ref{item:upwards:sig}) unmatched signature vertex and $\pi$ tends upwards,
(\ref{item:downwards:sig}) unmatched signature vertex and $\pi$ tends downwards.
Examples of these cases are shown in \figref{shortcutting:case5i} and
\figref{shortcutting:case5v}.
}
\tablab{matroska:subcases}
\end{table}

The exact case distinction is specified in \tabref{matroska:subcases}:
\begin{inparaenum}[(i)]
\item trivial case (see \claimref{trivial:subcase}),\label{item:trivial} 
\item $\pi$ stays level, \label{item:level}
\item $\pi$ tends upwards,\label{item:upwards}
\item $\pi$ tends downwards,\label{item:downwards}
\item unmatched signature vertex and $\pi$ tends upwards, \label{item:upwards:sig}
\item unmatched signature vertex and $\pi$ tends downwards. \label{item:downwards:sig}
\end{inparaenum}
It remains to prove that the case analysis is complete and to prove correctness
in each of these subcases.

\begin{claim}
The case distinction of subcases of \caseref{matroska} (\tabref{matroska:subcases}) is complete. 
\claimlab{matroska:complete}
%In the cases (\ref{item:upwards:sig}) and (\ref{item:downwards:sig}), it holds that  $q''<p_{\infty}$.
\end{claim}
\begin{proof}
We assume that we are not in the trivial case
\caseref{matroska}(\ref{item:trivial}). 
%By \defref{a:stop},
If $f(p^{+}) \leq y^{(\Astop)}$ (also $f(p^{+}) \leq y^{(\Astop)} \leq s_{j+1}$) 
%if $f(p^{+}) \leq s_{j+1}$, then $f(p^{+}) \leq y^{(a)} \leq s_{j+1}$ and 
we get one of
\caseref{matroska}(\ref{item:level})-(\ref{item:downwards}).
Otherwise we have $f(p^{+}) > y^{(\Astop)}$ (also $f(p^{+}) > s_{j+1} \geq y^{(\Astop)}$). In this case, we get one of  
\caseref{matroska}(\ref{item:upwards:sig})-(\ref{item:downwards:sig}).
In the following, we argue that, indeed, if the subcurve of $\tau$ specified by the parameter interval 
$[f(p^{-}),f(p^{+})]$ contains the signature vertex at $s_{j+1}$, it must be that 
\[ \tau(s_{j+1}) - \delta \in \setSubC{\pi}{p}{p^{+}}{1}\]
and thus, $q''\neq p_\infty$ and $\pi(q'')$ is be one of 
$\brc{\pi(p_{\max}), \tau(s_{j+1})-\delta}$.
%$\brc{g(\pi(p_{\max})), g(\tau(s_{j+1})-\delta)}$.
Assume that $s_{j+2} \neq 1$, i.e., the next signature vertex after
$\tau(s_{j+1})$ is not the last signature vertex. In this case, by
\propref{signature} and \propref{frechet}, we have
\[\tau(s_{j+1}) \geq \tau(s_{j+2})+2\delta \geq \pi(f^{-1}(s_{j+2})) +
\delta.\]
Since $\pi$ is continuous, this implies that there must exist a point $\pi(t)$
with $t \geq p^{+}$ and $\pi(t) \leq \tau(s_{j+1})-\delta$.
Now, assume that $s_{j+2} = 1$. In this case, we have by the theorem statement
that $\pi(p_i) \notin \range{\tau(s_{j+2})}{4\delta}$.
It must be that either $\tau(s_{j+1}) \geq \pi(p_i) > \tau(s_{j+2}) + 4\delta $
(in which case we can apply the above argument), or 
$\pi(p_i) < \tau(s_{j+2}) - 4\delta \leq \tau(s_{j+1}) - 5\delta$. The second
case is not possible since by \claimref{max:dist} and by \propref{frechet}
we have 
\begin{eqnarray*}
\pi(p_i) &\geq& \tau(f(p_i)) -\delta \geq \tau(t_{\min}) -\delta \geq
\pi(p_{\max}) - 4\delta \\
&& \geq \pi(p^{+})-4\delta \geq \pi(f^{-1}(s_{j+1}))-4\delta
\geq \tau(s_{j+1}) - 5\delta.
\end{eqnarray*}
Here, $\pi(p^{+}) \geq \pi(f^{-1}(s_{j+1}))$ follows from $s_{j+1} \in [f(p^{-}), f(p^{+})]$ 
in \caseref{matroska}(\ref{item:upwards:sig})-(\ref{item:downwards:sig})
and the fact that 
\[ \maxSubC{\pi}{p}{p^{-}}{p^{+}} = \pi(p^{+}), \]
by our initial assumptions.
\end{proof}

\begin{proof}[Proof of \caseref{matroska}(\ref{item:level})]
%(Note that by the definition of $q'$, it must be that $f(p^{+}) < y^{(\Astop)}$.)

By \claimref{frechet:y} the \Frechet distance between $\pi[x^{(\Astop)},p^{-}]$ and
$\tau(y^{(\Astop)})$ is at most $\delta$.
By our case distinction,  
\[\setSubC{\pi}{p}{p^{+}}{f^{-1}(y^{(\Astop)})} 
\subseteq [\tau(t_{\min}^{(\Astop-1)})+\delta, \pi(p_{\max})] 
\subseteq [\pi(p_{\max}) -2\delta, \pi(p_{\max})] 
= \range{\tau(y^{(\Astop)})}{\delta},\] 
since by \claimref{max:dist},
\[\tau(t_{\min}^{(\Astop-1)})+\delta \geq \tau(t_{\min})+\delta \geq
\pi(p_{\max}) -2\delta. \]
(Note that $t_{\min}^{(\Astop-1)}$ always exists since $y^{(1)} \leq s_{j+1}$ by
\claimref{xy:inside}.) 
This implies that also the second matching is valid.
\end{proof}

\begin{proof}[Proof of \caseref{matroska}(\ref{item:upwards})]
We can reduce this case to \caseref{matroska}(\ref{item:trivial}) (the trivial case) as
follows.
By our case distinction,  $f(p^{+}) \leq f(q') < y^{(\Astop)}$. Let $b$ be the
maximal value of $a$ such that $f(q') \in [y^{(a)}, y^{(\Astop)}]$.
By \propref{frechet} 
it must be that $\tau(f(q')) \geq \pi(p_{\max}) - \delta =
\tau(y^{(b)})$. Thus $\min(\tau[y^{(b)},f(q')])\geq \pi(p_{\max})-\delta$. This holds since for any $a'$, $\tau$ goes upwards in $\tau(y^{(a')})$, then intersects $\pi(p_{\max})-\delta$ downwards and goes upwards again in $\tau(y^{(a'+1)})$. By our case distinction, $f(p^{+}) \in [y^{(b)},
f(q')]$. Thus, 
\[ \tau(z_{\min}^{(b+1)}) = \minSubC{\tau}{t}{y^{(b)}}{f(p^{+})} \geq \minSubC{\tau}{t}{y^{(b)}}{f(q')} \geq \pi(p_{\max}) - \delta\geq \pi(p^-) - \delta.\]
\end{proof}

\begin{proof}[Proof of \caseref{matroska}(\ref{item:downwards})]
In this case, we rollback the last two matchings of the iterative matching
scheme and instead end with $a=\Astop-1$. Thus, we are left with the 
unmatched subcurves $\pi[x^{(\Astop-1)}, p^{-}]$ and $\tau[y^{(\Astop-1)}, f(p^{+})]$.  
We complete the matching scheme as defined in \tabref{matroska:subcases}.
The validity of the first matching follows directly from
\claimref{frechet:y}, since $x^{(\Astop-1)} > x$.
By the definition of $q'$ and our case distinction, 
\[\setSubC{\pi}{p}{p^{+}}{q'} \subseteq [\tau(t_{\min}^{(\Astop-1)}
+\delta, \pi(p_{\max}))] \subseteq \range{\tau(y^{(\Astop-1)})}{\delta}. \]
This proves validity of the second matching.
\claimref{frechet:xa} implies the validity of the last matching.
\end{proof}

\begin{figure}\centering
\includegraphics[page=1]{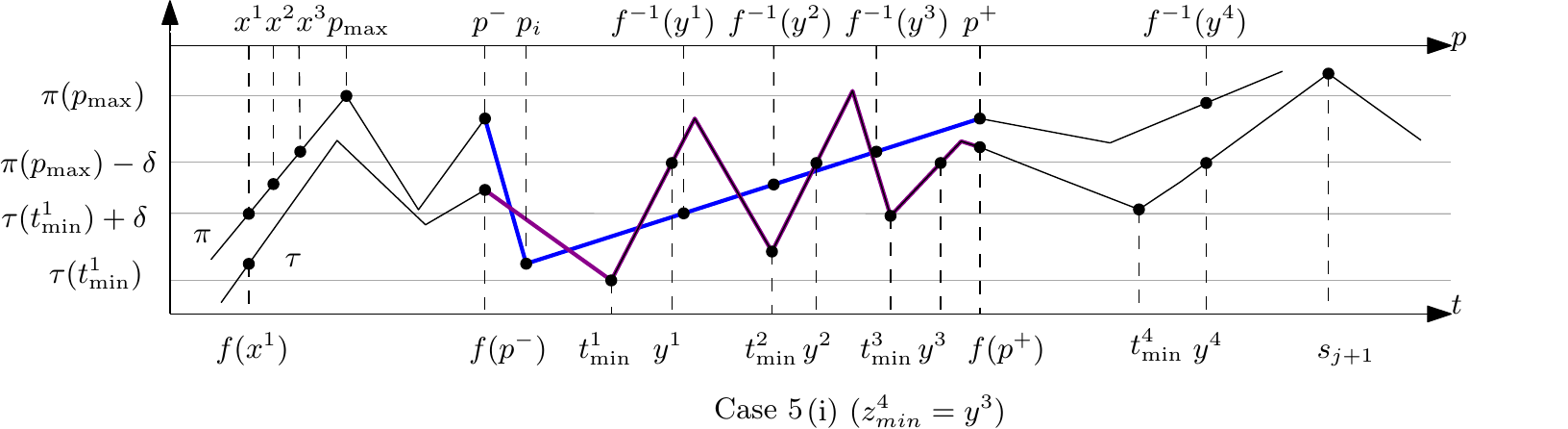}\\
\vspace{\baselineskip}
\includegraphics[page=2]{shortcutting_case5.pdf}\\
\vspace{\baselineskip}
\includegraphics[page=3]{shortcutting_case5.pdf}\\
\vspace{\baselineskip}
\includegraphics[page=4]{shortcutting_case5.pdf}\\
\caption{Examples of
\caseref{matroska}(\ref{item:trivial})-(\ref{item:downwards}).  The broken part of the matching
$f$ is indicated by fat lines.}
\figlab{shortcutting:case5i}
\end{figure}

\begin{figure}\centering
\includegraphics[page=5]{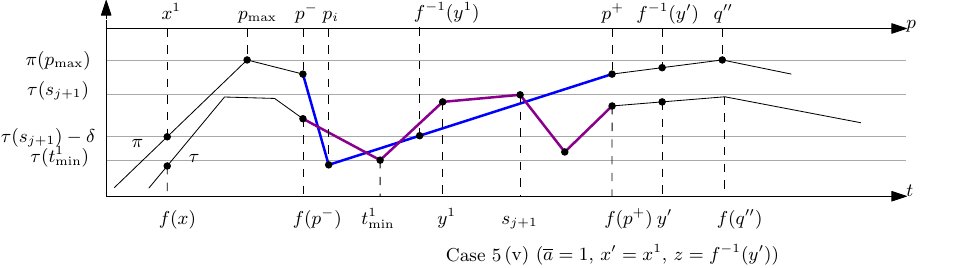}\\
\vspace{\baselineskip}
\includegraphics[page=6]{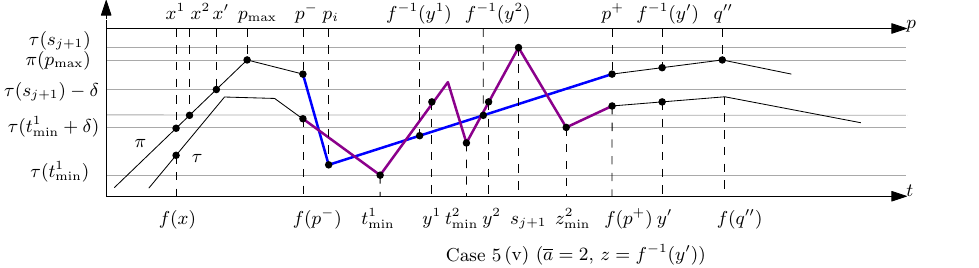}\\
\vspace{\baselineskip}
\includegraphics[page=7]{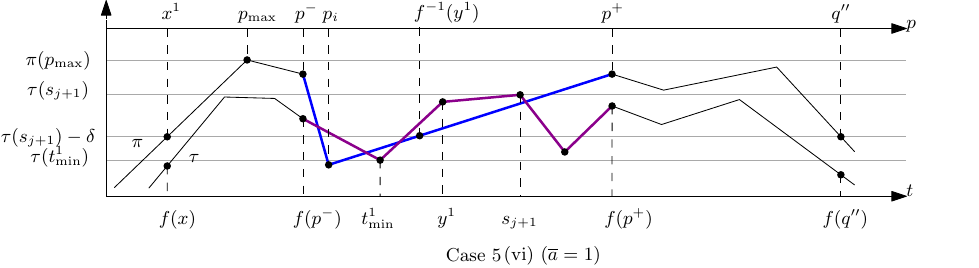}\\
\vspace{\baselineskip}
\includegraphics[page=8]{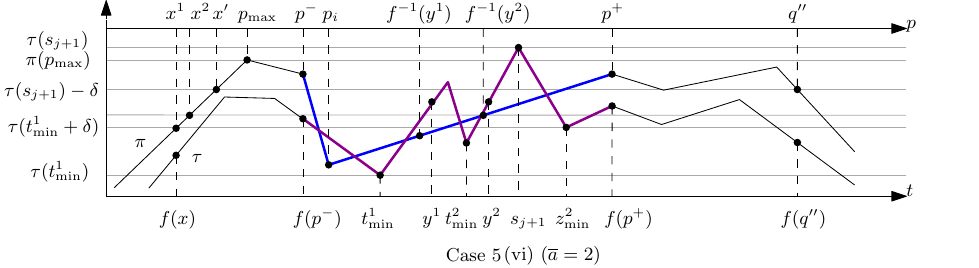}\\
\caption{Examples of
\caseref{matroska}(\ref{item:upwards:sig})-(\ref{item:downwards:sig}).  The broken part of the matching
$f$ is indicated by fat lines.}
\figlab{shortcutting:case5v}
\end{figure}

We have now handled
\caseref{matroska}(\ref{item:trivial})-(\ref{item:downwards}). Examples of these
cases are shown in \figref{shortcutting:case5i}.
We now move on to prove correctness of the remaining cases 
\caseref{matroska}(\ref{item:upwards:sig}) and \caseref{matroska}(\ref{item:downwards:sig}).

\begin{proof}[Proof of \caseref{matroska}(\ref{item:upwards:sig})]
Observe that in this case $q''=q=g(\pi(p_{\max}))$, as in \caseref{upwards}.
Therefore, $y'$ and $z$ are the same as in \caseref{upwards} and must exist.
We argue that $x'$ must also exist. Indeed, we can derive 
$\pi(x^{(\Astop)}) \leq \tau(s_{j+1})-\delta \leq \pi(p_{\max})$,
as follows. Recall that by our case distinction $f(p^{-}) \leq s_{j+1} \leq f(p^{+})$.
By \propref{frechet}, it follows that
\[\tau(s_{j+1}) = \maxSubC{\tau}{t}{f(p^{-})}{f(p^{+})} 
\leq  \maxSubC{\pi}{p}{p^{-}}{p^{+}} +\delta \leq \pi(p_{\max}) +\delta.\]

Now we need to prove the validity of the matching scheme.
The first line follows from \claimref{frechet:ya}.
For the second line we need to prove that
\[\setSubC{\tau}{t}{y^{(\Astop)}}{y'} \subseteq [\tau(s_{j+1})-2\delta,
\tau(s_{j+1})] = \range{\pi(x')}{\delta}.\]
The upper bound follows from \propref{signature}(\ref{item:sig:sj1:max}).
As for the lower bound, by the definition of the stopping parameter,
\[ \minSubC{\tau}{t}{y^{(\Astop)}}{s_{j+1}} = \tau(y^{(\Astop)})=
\pi(p_{\max})-\delta \geq \tau(s_{j+1})-2\delta, \]
(as we just proved above).
By \claimref{descent:after:sj1}, 
\[ \minSubC{\tau}{t}{s_{j+1}}{f(p^{+})} \geq \tau(s_{j+1})-2\delta.\]
By our case distinction and by \propref{frechet},
\[\minSubC{\tau}{t}{f(p^{+})}{f(q'')} \geq \minSubC{\pi}{p}{p^{+}}{q''}-\delta \geq
\tau(s_{j+1})-2\delta. \]
The validity of the third matching is implied by \claimref{frechet:y}.
For the last two matchings we can apply the respective part of the proof of
\caseref{upwards} verbatim.
\end{proof}

\begin{proof}[Proof of \caseref{matroska}(\ref{item:downwards:sig})]
The validity of the first matching follows from \claimref{frechet:y} and
since $x^{(\Astop)} \geq x$.
By our case distinction,
\[ \setSubC{\pi}{p}{p^{+}}{q''} \subseteq 
[\tau(s_{j+1})-\delta, \pi(p_{\max})] \subseteq \range{\tau(y^{(\Astop)}}{\delta}).\]
Thus, also the second matching is valid.
For the last matching we need to prove that
\[\setSubC{\tau}{t}{y^{(\Astop)}}{f(q'')} \subseteq [\tau(s_{j+1})-2\delta,
\tau(s_{j+1})] = \range{\pi(q'')}{\delta}.\]
Again, as in \caseref{matroska}(\ref{item:upwards:sig}), it holds that 
\[\setSubC{\tau}{t}{y^{(\Astop)}}{f(p^{+})} \subseteq [\tau(s_{j+1})-2\delta,
\tau(s_{j+1})].\]
By our case distinction and by \propref{frechet} 
\[ \minSubC{\tau}{t}{f(p^{+})}{f(q'')} \geq 
\minSubC{\pi}{p}{p^{+}}{q''} -\delta \geq 
\tau(s_{j+1})-2\delta\]
This also implies that $f(q'') < s_{j+2}$, by
\propref{signature}(\ref{item:sig:length}). 
Thus, by \propref{signature}(\ref{item:sig:sj1:max}), we conclude
\[\maxSubC{\tau}{t}{f(p^{+})}{f(q'')} \leq \tau(s_{j+1}).\]
Together this implies the validity of the last matching.
\end{proof}

We now proved correctness of the last two cases
\caseref{matroska}(\ref{item:upwards:sig}) and
\caseref{matroska}(\ref{item:downwards:sig}).  Examples of these cases are shown
in \figref{shortcutting:case5v}.

It remains to prove the boundary cases, which we have ruled out so far by
\assref{inner:sig:edge}. There are three boundary cases:
\begin{compactenum}[(B1)]
\item $s_j=0$ and $s_{j+1}=1$,
\item $s_j=0$ and $s_{j+1}<1$,
\item $s_j>0$ and $s_{j+1}=1$.
\end{compactenum}

To prove the claim in each of these cases, we can use the above proof verbatim
with minor modifications.  Note that in the proof, we used $s_j$ in its function
as the minimum on the signature edge $\overline{s_js_{j+1}}$, resp., we used
$s_{j+1}$ in its function as the maximum
on this edge.  Thus, let
 \[s_{\min}=\argmin_{s \in [s_j,s_{j+1}]} \tau(s), ~~~ s_{\max}=\argmax_{s \in [s_{j},s_{j+1}]} \tau(s). \]

\begin{claim}
In each of the cases (B1), (B2) and (B3), it holds that
$f(p_i) \in [s_{\min},s_{\max}]$ and $\tau(s_{\max})-\tau(s_{\min}) \geq
4\delta$.
\claimlab{pi:inside}
\end{claim}
\begin{proof}
By the theorem statement and by \defref{signature}, it holds that
\[ \pi(p_i) \notin \range{v_1}{4\delta} \cup \range{v_{\ell}}{4\delta} 
= \range{\tau(0)}{4\delta} \cup \range{\tau(1)}{4\delta}\]
i.e., the removed vertex $\pi(p_i)$ lies very far from the endpoints of
the curve $\tau$. At the same time, by \defref{signature},
in case $s_j=0$,
\[\tau(0) \geq \tau(s_{\min}) \geq \tau(0) - \delta\] 
and, in case $s_{j+1}=1$,
\[\tau(1) \leq \tau(s_{\max}) \leq \tau(1) + \delta.\]
By the direction-preserving property of \defref{signature} and by
\propref{frechet}, this implies that $f(p_i) \in [s_{\min},s_{\max}]$.
In the cases where $s_j=0$, this implies
 \[\tau(f(p_i)) \geq \tau(s_{\min}) \geq \tau(0) - \delta, \]
therefore, by the above, $\tau(f(p_i)) \geq \tau(0) + 4\delta \geq
\tau(s_{\min})+4\delta$.
Similarly in the cases, where $s_{j+1}=1$, 
we can derive that $\tau(f(p_i)) \leq \tau(1) - 4\delta \leq \tau(s_{\max})-4\delta$.
In each of the cases (B1), (B2) and (B3), this implies the second part of the claim.
\end{proof}

We replace \propref{signature} with the following property. 
\begin{property}[Signature (boundary case)]
\mbox{}
\begin{compactenum}[(i)]
\item $\tau(s_{\max}) - \tau(s_{\min}) > 2\delta$,
\item $\tau(s_{\min}) = \minSubC{\tau}{t}{s_{j-1}}{s_{j+1}}$ 
      (if $s_j=0$, then $\tau(s_{\min}) = \minSubC{\tau}{t}{0}{s_{j+1}}$),
\item $\tau(s_{\max}) = \maxSubC{\tau}{t}{s_j}{s_{j+2}}$ 
      (if $s_{j+1}=1$, then $\tau(s_{\max}) = \maxSubC{\tau}{t}{s_j}{1}$),
\item $\tau(t) \geq \tau(t')-2\delta$ for $s_{\min}\leq t' < t \leq
s_{\max}$,
\item if $s_{j}=0$, then $\tau(s_{\max}) - \tau(s_{j+2}) > 2\delta$.
\end{compactenum}
\proplab{signature:boundary}
\end{property}

\propref{signature:boundary}(\ref{item:sig:sj:min}), (\ref{item:sig:sj1:max}), 
(\ref{item:sig:descent}), and (\ref{item:sig:length:2}) hold by \defref{signature}.
\propref{signature:boundary}(\ref{item:sig:length}) follows from
\claimref{pi:inside}.

Instead of \claimref{descent:after:sj1} we use the claim
\begin{claim}
If $s_{\max} \in [f(p^{-}),f(p^{+})]$ then
$\setSubC{\tau}{t}{s_{\max}}{f(p^{+})} \subseteq [\tau(s_{\max})-2\delta, \tau(s_{\max})]$.
\end{claim}

Instead of \claimref{no:minimum} we use the claim
\begin{claim}
$s_{\min} \notin [f(p^{-}),f(p^{+})]$.
\end{claim}

Now, the theorem follows in the boundary cases (B1), (B2) and (B3), by replacing
$s_j$ with $s_{\min}$ and replacing $s_{j+1}$ with $s_{\max}$.

This closes the proof of \lemref{remove:one}
\end{proof}

\section{$(\nrClusters,\lenClusters)$-center}

\begin{algorithm}[h]\alglab{candidate:generator:center}
 \KwData{curves $P=\{\inputTraj{1},\dots,\inputTraj{n}\} \subset \Delta_m$, parameters $\alpha,\beta > 0$, $\nrClusters,\lenClusters \in \Na$}
 \KwResult{candidate set $\Gamma^{\nrClusters,\lenClusters}_{\alpha,\beta}(S) \subseteq \Delta_\ell$}
 \caption{Generate candidates for $(\nrClusters,\lenClusters)$-center from signature vertices}
For each $\inputTraj{i}$, let $\VtxSet_i$ be the vertex set of its $\alpha$-signature computed by \algref{low:pass}\; 
Compute the union $U$ of the intervals $r=\left[w-4\alpha,w+4\alpha\right]$ for $w
\in \VtxSet=\bigcup_{i=1}^{n} \VtxSet_i$\; 
\eIf{ $\mu(U) > 24\alpha\nrClusters\lenClusters$}
{Return the empty set\;}
{Discretize $U$ with resolution $\beta$, thereby generating a set of vertices~$\widehat{\VtxSet}$\;
Return all possible curves consisting of $\lenClusters$ vertices from $\widehat{\VtxSet}$\;}
\end{algorithm}

\begin{lemma}\lemlab{candidate:generator:center}
Given a set of curves $P=\{\inputTraj{1},\dots,\inputTraj{n}\} $ and parameters
$\alpha,\beta > 0$, and $\nrClusters,\lenClusters \in \Na$, then \algref{candidate:generator:center}
generates a set of candidate solutions $\Gamma^{\nrClusters,\lenClusters}_{\alpha,\beta}(S) \subseteq \Delta_\ell$
of size at most $\left(\floor{\frac{24\alpha\nrClusters\lenClusters}{\beta}}+6\nrClusters\lenClusters\right)^{\lenClusters}$.
Furthermore, if $\alpha \geq \opt^{(\infty)}_{\nrClusters,\lenClusters}(P)$, then the generated set contains 
$k$ candidates $\tilde{C}=\tilde{c_1},\dots,\tilde{c_k}$ with 
\[ \cost_{\infty}(P, \tilde{C}) \leq \alpha + \beta. \]
%\[ \cost_{\infty}(P, \tilde{C}) \leq \opt^{(\infty)}_{\nrClusters,\lenClusters}(P) + \beta. \]
\end{lemma}

\begin{proof} 
Let  $C=c_1,\dots,c_{\nrClusters}$ denote an optimal solution for $P$ and  
let $c_i=z_{i1},\dots,z_{i{\lenClusters}}$ denote the vertices for each cluster
center. 
Consider the union of intervals
%\begin{align*}
\[ R= \bigcup_{i=1}^{\nrClusters} \bigcup_{j=1}^{\lenClusters} [z_{ij}-4\alpha,z_{ij}+4\alpha].\]
%\end{align*}
\lemref{nec:suff} implies that $R$ contains all elements of $\VtxSet$, the
signature vertices computed by \algref{candidate:generator:center}.
Now consider the dual statement, namely, whether the vertices $z_{ij}$ are contained in the set $U$ computed by the algorithm.
If there exists a $z_{ij}$ which is not contained in $U$, then
\thmref{remove:one} implies that we can omit $z_{ij}$ from the solution 
while not increasing the cost beyond $\alpha$.
%without changing its cost. 
Therefore, let $\widehat{C}$ denote the solution 
where all vertices that lie outside $U$ have been omitted.
Clearly, $U$ contains all remaining vertices of cluster centers in $\widehat{C}$. 
Therefore, $\Gamma^{\nrClusters,\lenClusters}_{\alpha,\beta}$ must contain $k$
candidates $\tilde{C}=\tilde{c_1},\dots,\tilde{c_k}$ with 
%\begin{align*}
\[ \cost_{\infty}(P,\tilde{C}) \leq \alpha + \beta. \]
%\[ \cost_{\infty}(P,\tilde{C}) \leq \opt^{(\infty)}_{\nrClusters,\lenClusters}(P) + \beta. \]
%\end{align*}
Note that $R$ consists of at most $\nrClusters\lenClusters$ intervals and has measure at most
$\mu(R)=8\alpha\nrClusters\lenClusters$.
Therefore, the measure of $U$ can be at most 
$\mu(U) \leq \mu(R) +(2\nrClusters\lenClusters)8\alpha = 24\alpha\nrClusters\lenClusters.$
In the worst case a signature vertex lies at each boundary point of $R$.
Furthermore, $U$ consists of at most $\ceil{\frac{\mu(U)}{8\alpha}}\leq 3\nrClusters\lenClusters$ 
intervals, since each interval has measure at least $8\alpha$.
\end{proof}

%\begin{theorem}\thmlab{k:l:center:main}
%Let $\eps >0$ and $\nrClusters,\lenClusters \in \Na$ be given constants. 
%There exists a constant $t_{\nrClusters,\lenClusters,\eps}$ such that 
%given a set of curves $P=\{\inputTraj{1},\dots,\inputTraj{n}\}$
%we can compute a $(1+\eps)$-approximation to $\opt^{(\infty)}_{\nrClusters,\lenClusters}(P)$
%in time $O\pth{{nm\log\pth{\frac{1}{\eps}}+t_{\nrClusters,\lenClusters,\eps}n m\lenClusters \log\pth{m\lenClusters}}\log\pth{\frac{1}{\eps}}}$.
%\end{theorem}

%short version. The long version is above
\begin{theorem}\thmlab{k:l:center:main}
Let $\eps >0$ and $\nrClusters,\lenClusters \in \Na$ be given constants. 
Given a set of curves $P=\{\inputTraj{1},\dots,\inputTraj{n}\}$
we can compute a $(1+\eps)$-approximation to $\opt^{(\infty)}_{\nrClusters,\lenClusters}(P)$
in time $O\pth{nm\log m}$. 
%\Amer{Short version of the theorem, without specified constants}
\end{theorem}

\begin{proof}
We use 
\algref{phase1}
%the algorithm 
described in \secref{cf:approx:uni} to compute a
constant-factor approximation.  We obtain an interval
$[\delta_{\min},\delta_{\max}]$ which contains
$\opt^{(\infty)}_{\nrClusters,\lenClusters}(P)$ and such that
$\delta_{\max}/\delta_{\min}\leq 8$ by \thmref{cf:approx:center:main}.
We can now do a binary search in this interval.  
In each step of the binary
search, we apply \algref{candidate:generator:center} to $P$ a constant number
of times and evaluate every candidate solution. 
More specifically, if we apply \algref{candidate:generator:center} to $P$ with parameters
$\alpha$ and $\beta=\eps\alpha$, 
by \lemref{candidate:generator:center}
we gain the following knowledge:
\begin{compactenum}[(i)]
\item Either $\alpha < (1+\eps)\opt^{(\infty)}_{\nrClusters,\lenClusters}(P)$, or
\item $\alpha \geq \opt^{(\infty)}_{\nrClusters,\lenClusters}(P)$ and 
we have computed a solution with cost at most $(1+\eps)\alpha$.
\end{compactenum}
In both cases, the outcome is correct. Since we want to take an exact decision 
during the binary search, we simply call the procedure twice with 
parameters $\alpha_1=\frac{\alpha}{1+\eps/2},\beta_1=\frac{\eps}{2}\alpha_1$ and
$\alpha_2=\alpha,\beta_2=\eps\alpha_2$. 
Now there are three possible outcomes:
\begin{compactenum}[(i)]
\item Either $\opt^{(\infty)}_{\nrClusters,\lenClusters}(P) < \alpha$, or
\item $\opt^{(\infty)}_{\nrClusters,\lenClusters}(P) > (1+\eps)\alpha$, or
\item $\opt^{(\infty)}_{\nrClusters,\lenClusters}(P) \in [\alpha,(1+\eps)\alpha]$.
\end{compactenum} 
So either we can take an exact decision and proceed with the
binary search, or we obtain a $(1+\eps)$-approximation to the solution and we
stop the search.

%%old longer version, with more constants
%Let $t'_{\nrClusters,\lenClusters,\eps}$ denote the size of the candidate
%set $\Gamma^{\nrClusters,\lenClusters}_{\alpha,\beta}$ generated in one binary search step. 
%By 
%\lemref{candidate:generator:center}
%we have in each step 
%\begin{align*}
%t'_{\nrClusters,\lenClusters,\eps} \leq 
%\left(\floor{\frac{24\nrClusters\lenClusters}{\eps}}+6\nrClusters\lenClusters\right)^{\lenClusters}.
%\end{align*}
%One execution of \algref{candidate:generator:center} takes $O(nm)$ time for
%computing the $n$ signatures (using \algref{low:pass}) and $O\pth{t'_{\nrClusters,\lenClusters,\eps}}$
%time for generating the candidate set.
%Let $t''_{\nrClusters,\lenClusters,\eps}=\pth{ t'_{\nrClusters,\lenClusters,\eps} \atop k }$ denote 
%the number of candidate solutions generated from one such step.
%Evaluating one candidate solution takes $kn$ \Frechet distance computations, where one
%\Frechet distance computation takes time in $O\pth{m\lenClusters
%\log\pth{m\lenClusters}}$ using the algorithm by Alt and Godau
%\cite{ag-cfdbt-95}.  The number of binary search steps is bounded by
%$O\pth{\log\pth{\frac{1}{\eps}}}$ if we use exponential binary search
%intervals. 
%Now the overall running time follows if we set $t'_{\nrClusters,\lenClusters,\eps}=k t''_{\nrClusters,\lenClusters,\eps}$.

%new shorter version
By \lemref{candidate:generator:center} we know that the size of the
candidate set $\Gamma^{\nrClusters,\lenClusters}_{\alpha,\beta}$ is $O(1)$
(where the constant depends on $\eps, \nrClusters$ and $\lenClusters$.

One execution of \algref{candidate:generator:center} takes $O(nm)$ time for
computing the $n$ signatures (using \algref{low:pass}) and 
$O(1)$ time for generating the candidate set.
Evaluating one candidate solution (consisting of $k$ centers from the candidate set) 
takes $kn$ \Frechet distance computations, where one
\Frechet distance computation takes time $O(m \log m)$
using the algorithm by Alt and Godau
\cite{ag-cfdbt-95}.  The number of binary search steps depends only on the constant
$\eps$ and so is $O(1)$, which implies the running time.
\end{proof}

\section{$(\nrClusters,\lenClusters)$-median}

In this section we will make use of a result by Ackermann \etal \cite{abs-cm-10} for 
computing an approximation to the $k$-median problem under an arbitrary dissimilarity
measure $D:\Delta \times \Delta \rightarrow \mathbb{R}_0^+$ on a ground set of items $\Delta$, i.e. a 
function that satisfies $D(x,y) = 0$, iff $x \not= y$. The result roughly says that
we can obtain an efficient $(1+\varepsilon)$-approximation algorithm for the $k$-median
problem on input $P \subseteq \Delta$, if there is an algorithm that given a random sample 
of constant size returns
a set of candidates for the $1$-median that contains with constant probability (over the
choice of the sample) a $(1+\varepsilon)$-approximation to the $1$-median.

We restate the sampling property defined by Ackermann \etal~(\cite{abs-cm-10},Property 4.1).

\begin{definition}[sampling property]
\deflab{sampling:property}
We say a dissimilarity measure D satisfies the (weak)
$[\eps,\lambda]$-sampling property iff there exist integer constants
$m_{\eps,\lambda}$ and $t_{\eps,\lambda}$ such that for each $P\subseteq \Delta$
of size $n$ and for each uniform sample multiset $S\subseteq P$ of size $m_{\eps,\lambda}$ a set 
$\Gamma(S) \subseteq \Delta$ of size at most $t_{\eps,\lambda}$ can be computed satisfying
\[\Pr\left[ \exists \tilde{c} \in \Gamma(S) : D(P,\tilde{c}) \leq (1+\eps) \opt(P) \right] \geq 1-\lambda.\]
Furthermore, $\Gamma(S)$ can be computed in time depending on $\gamma,\delta$ and $m_{\eps,\lambda}$ only.
\end{definition} 

It is likely that the sampling property
(\defref{sampling:property}) does not hold for the \Frechet distance for
arbitrary value of $\lenClusters$. We will therefore prove a modified sampling
property, which allows the size of the sample to depend on $\lenClusters$.

The following lemma intuitively says that curves that lie far away from a candidate median 
have little influence on the shape of the candidate median.

\begin{lemma}\lemlab{omit:far:away}
Given a set of $n$ curves $P=\{\tau_1,\ldots, \tau_n\}$ and a polygonal curve $\pi$, it holds that 
\[\cost_{1}(P,\widehat{\pi}) \leq (1+\eps)\cost_{1}(P,\pi) \] 
for a curve $\widehat{\pi}$ obtained from $\pi$ by omitting any subset of vertices lying outside 
the following union $R_S$:%was $R_S$:
\[
R_S = \bigcup_{ 1\leq i \leq n \atop x_i \leq \frac{2x_1}{\eps} }\; \bigcup_{v \in \VtxSet(\sigma_i) } [v-4x_i,v+4x_i],
\]
where the $\tau_i$ are sorted in increasing order of $x_i=\distFr{\inputTraj{i}}{\pi}$ and where
$\sigma_i$ is the $x_i$-signature of $\inputTraj{i} \in P$. %was written in S
\end{lemma}

\begin{proof}
By \lemref{nec:suff}, the curve $\pi$ has a vertex in each %signature interval \footnote{Has this been defined?} 
range centered at the vertices 
%\Amer{here was Christian's footnote, on not defined term}
of 
$\sigma_i$. These will not be omitted, therefore it is ensured that
$\widehat{\pi}$ has at least 2 vertices, i.e. defines a curve.
By  \thmref{remove:one},  it holds that
$\distFr{\widehat{\pi}}{\inputTraj{i}}\leq x_i$ for the curves $\inputTraj{i}$
with $x_i \leq \frac{2x_1}{\eps}$, that is, for the curves that lie close
to~$\pi$.
We now argue using the triangle inequality that the distances to the curves
that lie further away are only altered by a factor of at most $(1+\eps)$.
Consider any index $i$, such that $x_i > \frac{2x_1}{\eps} = \widehat{x}$.
By the triangle inequality, it holds that 
%\Amer{the third inequality below does not have to be equality, so it was changed}
\begin{eqnarray*}
\distFr{\widehat{\pi}}{\inputTraj{i}} 
&\leq& \distFr{\widehat{\pi}}{\inputTraj{1}} + \distFr{\inputTraj{1}}{\inputTraj{i}} \\
&\leq& \distFr{\widehat{\pi}}{\inputTraj{1}} +\distFr{\inputTraj{1}}{\pi} + \distFr{\pi}{\inputTraj{i}} \\
&\leq& x_1 + x_1 + x_i \\
&<&  (1+\eps) x_i
\end{eqnarray*}
%\begin{eqnarray*}
%\distFr{\widehat{\centerTraj{}}}{\inputTraj{i}} 
%&\leq& \distFr{\inputTraj{1}}{\inputTraj{i}} + \distFr{\widehat{\centerTraj{}}}{\inputTraj{1}}\\
%&\leq& \distFr{\inputTraj{1}}{\pi} + \distFr{\pi}{\inputTraj{i}} +
%\distFr{\widehat{\centerTraj{}}}{\inputTraj{1}}\\
%&=& x_1 + x_i + x_1\\
%&<&  (1+\eps) x_i
%\end{eqnarray*}

Therefore,
\begin{align*}
\cost_{1}(P,\widehat{\pi}) 
\leq \sum_{x_i \leq \widehat{x}}  \distFr{\pi}{\inputTraj{i}} 
+ \sum_{x_i > \widehat{x}} (1+\eps) \distFr{\pi}{\inputTraj{i}} 
\leq (1+\eps) \cost_{1}(P,\pi).
\end{align*}

\end{proof}

The following lemma is in similar spirit as \lemref{omit:far:away}. We prove that the
basic shape of a candidate median can be approximated based on a constant-size
sample.

\begin{lemma}\lemlab{omit:low:prob}
There exist integer constant $m_{\eps,\lambda,\lenClusters}$ such that given a
set of curves $P$  and a curve $\pi=u_1,\dots,u_{\lenClusters}$ for each uniform
sample multiset $S\subseteq P$ of size $m_{\eps,\lambda,\lenClusters} \ge \ceil{\frac{8\lenClusters}{\eps}
\left(\log\left(\frac{1}{\lambda}\right)+\log(\lenClusters)\right)}$
it holds that 
\[\Pr\left[ \cost_{1}(P,\widehat{\pi}) \leq (1+\eps)\cost_{1}(P,\pi) \right] \geq 1-\lambda,\] 
for a curve $\widehat{\pi}$ obtained from $\pi$ by omitting any subset of vertices lying outside 
the following union $R_S$:
\[
R_S = \bigcup_{ \inputTraj{i} \in S }\; \bigcup_{v \in \VtxSet(\sigma_i) } [v-4x_i,v+4x_i],
\]
where the $\tau_i$ are sorted in increasing order of $x_i=\distFr{\inputTraj{i}}{\pi}$ and where
%where $x_i=\distFr{\inputTraj{i}}{\pi}$, sorted in increasing order, and where
$\sigma_i$ is the $x_i$-signature of $\inputTraj{i} \in S$.
\end{lemma}
\begin{proof}
If all vertices of $\pi$ are contained in $R_S$, then $\pi=\widehat{\pi}$ and the claim is implied.
%Assume for a fixed vertex $u_j$ that it is not contained in $R_S$. 
However, this is not necessarily the case. In the following, we consider a fixed vertex $u_j$ 
and we prove that it is either contained in $R_S$ with sufficiently high probability or
ignoring it will not increase the cost of a solution significantly. 

For this purpose, let $T_j \subseteq P$ be the subset of curves $\inputTraj{i}$ with 
\[u_j \in \bigcup_{v\in \VtxSet(\sigma_i) } [v-4x_i,v+ 4x_i].\] 
If any curve of $T_j$ is contained in our sample $S$, then $u_j$ is
contained in $R_S$.

We distinguish two cases.  If $T_j$ is large enough then $u_j$ is
contained in $R_S$ with high probability, or we argue that the total change in cost resulting 
from omitting $u_j$ from $\pi$ will be small.  
We will first argue that for all $j$ with $1\le j \le \ell$ and with $T_j > \frac{\eps n}{4\ell}$ we obtain for our choice of $m_{\eps,\lambda,\lenClusters}$ that at least one element of $T_j$ is contained in $S$. Indeed, we have
\begin{align}
\Pr[ T_j \cap S = \emptyset] &\leq (1-|T_j|/n)^{m_{\eps,\lambda,\lenClusters}} \leq \left( 1 - \frac{\eps}{4 \lenClusters} \right) ^{m_{\eps,\lambda,\lenClusters}}.
\end{align}
We use the union bound, to estimate the probability that this event fails for at least one of the 
sets $T_j$ in question. We choose the parameter $m_{\eps,\lambda,\lenClusters}$ large
enough, such that it holds for the failure probability that
\[\lambda \leq \lenClusters \left( 1 - \frac{\eps}{4 \lenClusters} \right)
^{m_{\eps,\lambda,\lenClusters}}.\]
For this, it suffices to choose 
\begin{align*}
m_{\eps,\lambda,\lenClusters} &\ge
%\frac{\log(\lambda/\lenClusters)}{\log(1-\eps/4\lenClusters)} 
%\leq 
%\frac{\log\left(\frac{1}{\lambda}\right)+\log(\lenClusters)}{\log\left(\frac{1}{1-\eps/c_2\lenClusters}\right)}
%\leq 
%\frac{\log\left(\frac{1}{\lambda}\right)+\log(\lenClusters)}{\log\left(1+\eps/c_2\lenClusters \right)}
%\leq 
\ceil{\frac{8\lenClusters}{\eps}
\left(\log\left(\frac{1}{\lambda}\right)+\log(\lenClusters)\right)}
\end{align*}
to obtain that with probability at least $1-\lambda$ for all $1\le j \le \ell$, simultaneously,
we have that at least one element of $T_j$ is in $S$, if $|T_j| \ge \frac{\eps n}{4 \ell}$.

Now consider the set of curves $T = \{ T_j : 1\le j \le \ell, T_j \cap S = \emptyset\}$.
By our previous considerations we have that with probability at least $1-\lambda$,
$T \subseteq \left\{T_j \left| 1\leq j \leq\lenClusters, |T_j|
\leq \frac{\eps n}{4 \lenClusters} \right. \right\}$. We will assume in the following
that this event happens.

%Concretely, consider the following set of
%curves  \[ T = \left\{T_j \left| 1\leq j \leq\lenClusters, |T_j|
%\leq \frac{\eps n}{4 \lenClusters} \right. \right\}. \]
%Note that the set of curves $T$ is independent of the sample $S$.
% where $c_2$ is a constant c_2=4 which will be fixed later. 
Let $\tilde{\pi}$ denote the curve obtained from $\pi$ by removing all vertices from
$R_S$, which is equivalent to removing all vertices $u_j$ with $T_j \in T$. 

%Let $\tilde{\pi}$ denote the curve obtained from $\pi$ where all vertices $u_j$
%have been removed if $|T_j|\leq \frac{\eps n}{4 \lenClusters}$.
In the following, let $P_T = \bigcup_{T'\in T} T'$ be the set of input curves that are
contained in one of the sets in $T$.
By \thmref{remove:one}, it holds for any curve $\inputTraj{i} \in P \setminus P_T$, that 
$\distFr{\tilde{\pi}}{\inputTraj{i}} \leq \distFr{\pi}{\inputTraj{i}}=x_i$, i.e., 
their distances do not increase 
beyond $x_i$
by the removal. 
%\Amer{added beyond $x_i$}
Let $\inputTraj{q}$ be the curve of this set with minimal distance to $\pi$
(i.e. with smallest index $q$).  
%We choose $c_2$ large enough such that $q \leq |T| + 1 \leq \frac{\eps n}{c_2} + 1 \leq \ceil{n/2}$. 
Since at least half of the input curves have to lie within a radius of $\frac{2}{n}\cost_{1}(P,\pi)$ from $\pi$ (two times the average distance of the input curves to $\pi$)
and since the union of the sets from $T$ has size less than $n/2$ (with probability at least
$1-\lambda$), this implies that $x_q \leq \frac{2}{n}\cost_{1}(P,\pi)$.
%At least half of the input curves have to lie within a radius of $\frac{2}{n}\cost_{1}(P,\pi)$ from $\pi$, otherwise the sum of the costs of the remaining input curves would be negative, what is absurd. Therefore it holds that $x_q \leq \frac{2}{n}\cost_{1}(P,\pi)$.
Therefore,
\begin{align*}
\cost_{1}(P,\tilde{\pi}) &= \cost_{1}(P\setminus P_T, \tilde{\pi}) + \cost_{1}(P_T, \tilde{\pi})\\
 &\leq \cost_{1}(P\setminus P_T, \pi) + \sum_{\inputTraj{} \in P_T} \left( \distFr{\inputTraj{}}{\pi} + \distFr{\pi}{\inputTraj{q}} + \distFr{\inputTraj{q}}{\tilde{\pi}} \right) \\
 & \leq \cost_{1}(P\setminus P_T, \pi) +  \cost_{1}(P_T, \pi) + 2|P_T|x_q\\
 &\leq \cost_{1}(P,\pi) + \frac{\eps n}{4} \cdot \frac{4\cost_{1}(P,\pi)}{n}\\
 &= \left(1+{\eps}\right)\cost_{1}(P,\pi)  .
\end{align*}
%if we choose $c_2$ large enough, say $c_2=4$. 
\end{proof}

\subsection{Generating Candidate Solutions}

Our next step is to define an algorithm that generates a set of candidate curves from the 
sample set. 

\begin{algorithm}[h]\alglab{candidate:generator}
 %\KwData{vertices $\VtxSet=\{v_1,\dots,v_t\} \subset \Re$, parameters $0 < \costParam_{\min} < \costParam_{\max},\eps \in (0,1]$, $\lenClusters,n \in \Na$}
 \KwData{curves $S=\{\inputTraj{1},\dots,\inputTraj{s}\} \subset \Delta_m$, parameters $\alpha,\beta > 0$, $\lenClusters \in \Na$}
 \KwResult{candidate set $\Gamma^{\lenClusters}_{\alpha,\beta}(S) \subseteq \Delta_\ell$}
 \caption{Generate candidates for $(1,\lenClusters)$-median from signature vertices}
For each $\inputTraj{i}$, let $\VtxSet_i$ be the vertex set of the 
signature of size at most $\ell+3$~(\thmref{compute:signatures})\; 
Compute the union $U$ of the intervals $r=\left[v-8\alpha,v+8\alpha\right]$ for $v
\in \VtxSet=\bigcup_{i=1}^{s} \VtxSet_i$\; 
Discretize $U$ with resolution $\beta$, thereby generating a set of vertices~$\widehat{\VtxSet}$\;
Return all possible curves consisting of $\lenClusters$ vertices from
$\widehat{\VtxSet}$\;
\end{algorithm}

We prove some properties of \algref{candidate:generator} and of the candidate
set generated by it. This proof serves as a basis for the proof of the sampling
property in \thmref{modified:sampling:prop}. 
%\Amer{the Algorithm 2, line 2, had the union up to $t$, which is not defined}

\begin{lemma}\lemlab{candidate:generator}
Given a set of curves $S=\{\inputTraj{1},\dots,\inputTraj{s}\} $ and parameters
$\alpha,\beta,\eps > 0$, and $\lenClusters \in \Na$, with 
%$\alpha\geq\frac{\opt^{2}_{1,\lenClusters}(S)}{\eps s}$,
$\alpha \geq \min_{1\leq i\leq s} \frac{ \distFr{\inputTraj{i}}{c_s} }{\eps}$, where $c_s$ denotes an optimal $(1,\lenClusters)$-median of $S$. 
%\Amer{The definition of $c_s$ was before in the proof}
There exist $\widehat{c} \in \Delta_{\lenClusters}$ with
\[ \cost_{1}(S,\widehat{c}) \leq (1+\eps)\opt^{(1)}_{1,\lenClusters}(S),\] 
and 
\algref{candidate:generator} computes a set of candidates $\Gamma^{\lenClusters}_{\alpha,\beta}(S) \subseteq \Delta_\ell$
of size  $\left({\frac{16\alpha s(\lenClusters+3)}{\beta}}\right)^{\lenClusters}$
which contains an element $\tilde{c}$, such that 
\[\distFr{\widehat{c}}{\tilde{c}}\leq \beta.\]
\end{lemma}

\begin{proof}
%Let $c_s$ be an optimal $(1,\lenClusters)$-median of $S$.
Let $\inputTraj{1},\dots,\inputTraj{s}$ denote the input curves in the increasing order of their
distance denoted by $x_i=\distFr{\centerTraj{s}}{\inputTraj{i}}$. For every $\inputTraj{i}$,
consider its $x_i$-signature denoted by $\sigma_i$.   
By \lemref{nec:suff}, each vertex of
$\centerTraj{s}$ lies within distance $4x_i$ to a vertex of some signature
$\sigma_i$ otherwise we can omit it by \thmref{remove:one}.  Hence, there must be a solution where
$\centerTraj{s}$ has its vertices in the union of the intervals. 
%\Amer{the union below was up to $n$. Here the input contains $s$ curves (as samples).}
\[ \bigcup_{i=1}^s \bigcup_{v \in \VtxSet(\sigma_i) }  [v-4x_i,v+4x_i].\] 

Since $x_i$ could be very large, we cannot cover this entire region with candidates. 
Instead, we consider the following union of intervals: 
\[ R_S = \bigcup_{ x_i
\leq \widehat{x}}\; \bigcup_{v \in \VtxSet(\sigma_i) } [v-4x_i,v+4x_i],\]
with $\widehat{x}=\frac{2x_1}{\eps}$.
Now, let $\widehat{\centerTraj{}}$ be the median obtained from $\centerTraj{s}$ by
omitting all vertices which do not lie in $R_S$. 
\lemref{omit:far:away} implies $\cost_{1}(S,\widehat{\centerTraj{}}) \leq (1+\eps)  \cost_{1}(S,\centerTraj{s}).$

Clearly, the generated set contains a curve $\tilde{\centerTraj{}}$
which lies within \Frechet distance $\beta$ of $\widehat{\centerTraj{}}$. 
Indeed, by \lemref{canonical:signature}, the vertices of $\sigma_i$ are contained in
the set of signature vertices computed by \algref{candidate:generator}, since 
\corref{nec:suff} implies that 
%$\ell \geq |\VtxSet\pth{\centerTraj{s}}| \geq |\VtxSet\pth{\widehat{\sigma}_i}|-2$.  
$\ell \geq |\VtxSet\pth{\centerTraj{s}}| \geq |\VtxSet\pth{\sigma_i}|-2$.  
%\Amer{(Here was $|\VtxSet\pth{\widehat{\sigma}_i}|$. The hat version is not defined and does not make sense.)}
If a signature of size $\ell+3$ does not
exists, then by the general position assumption, there must be a signature of
size $\ell+2$.
The algorithm sets $\widehat{x}= 2\alpha \geq 2 x_1/\eps$. 
Therefore, the generated candidate set covers the region $R_S$ with resolution~$\beta$.
\end{proof}

Before we prove the modified sampling property we prove the following two easy
lemmas.

\begin{lemma}\lemlab{markov}
Let $\lambda > 0$. Given a set of curves $P$, for each uniform sample multiset
$S\subseteq P$ it holds that 
\begin{align*}
\Pr\left[\cost_1(S,c) \geq \frac{|S|}{\lambda n} \opt^{(1)}_{1,\lenClusters}(P)\right] \leq \lambda,
\end{align*}
where $c$ is an optimal median of $P$.
\end{lemma}

\begin{proof}
It holds that 
\[\Exp\pbrc{\cost_1(S,c)} =  \frac{|S|}{n} \opt^{(1)}_{1,\lenClusters}(P). \]
Since $\cost_1(S,c)$ is a nonnegative random variable we can apply Markov's inequality and obtain 
\[\Pr\left[\cost_1(S,c) \geq \frac{\Exp\pbrc{\cost_1(S,c)}}{\lambda} \right] \leq \lambda,\]
which implies the claim.
\end{proof}

\begin{lemma}\lemlab{diam:S}
Let $\lambda > 0$. Given a set of curves $P$, for each uniform sample multiset
$S\subseteq P$ of size at least $\ceil{5\log(\frac{1}{\lambda})}+1$ it holds that 
\[\Pr\pbrc{ 12 \opt_{1,\lenClusters}^{(1)}(S) \geq \min_{\tau \in P}
\distFr{\tau}{c} } \geq 1-\lambda, \]
where $c$ denotes an optimal $(1,\lenClusters)$-median of $P$.
\end{lemma}

\begin{proof}
We analyze two cases. For the first case, assume that there exists a curve $q\in
\Delta_{\lenClusters}$, such that 
\[ \cardin{\brc{\tau \in P: \distFr{q}{\tau} \leq r}} \geq \frac{5}{7} |P|, \] 
where $r=\frac{\distFr{q}{c}}{5}$.
That is, a large fraction of $P$  lies within a small ball far away from the
optimal center. We let
\[ Q= \brc{\tau \in P:  \distFr{q}{\tau} < 2r},\]
and we claim that $Q$ has size at most $\frac{6}{7}|P|$.
Assume the opposite for the sake of contradiction. In this case, it follows by the triangle inequality that
\begin{eqnarray*}
\cost_{1}(P,q) &\leq \cost_{1}(P,c) - |Q|r + |P\setminus Q| 5r \leq
\cost_{1}(P,c) - \frac{1}{7} |P| r < \cost_1(P,c).
\end{eqnarray*}
This would imply that $c$ is not optimal. 
We analyze the event that at least one curve of $P$ lies within
\Frechet distance $r$ of $q$ and at least one curve lies further than $2r$ from
$q$. If this event happens, then again by the triangle inequality 
\[ \opt^{(1)}_{1,\lenClusters}(S)\geq \max_{\pi,\pi' \in S} \distFr{\pi}{\pi'} \geq r \geq \frac{\min_{\tau \in
P}\distFr{c}{\tau}}{6}.\]
Clearly, it holds for the $i$th sample point $s_i$, that
\[ \Pr\pbrc{ \distFr{s_i}{q} \leq r } \geq \frac{5}{7}~~\text{and}~~ \Pr\pbrc{ \distFr{s_i}{q} \geq 2r } \geq \frac{1}{7}.\]
Using this, we can show that $\ceil{5 \log(\frac{1}{\lambda})}$ samples suffice to ensure that this
event happens with probability of at least $(1-\lambda)$.

Now, assume the second case that no such $q$ exists. Let $s_1$ be the first
sample point and let $\widehat{s_1}$ be its minimum-error $\ell$-simplification
(\defref{min:error:simp}). We need to prove the claim in the case that
\begin{equation}\label{assump1}
12 \distFr{s_1}{\widehat{s_1}} < \distFr{s_1}{c},  
\end{equation}
since $\opt^{(1)}_{1,\lenClusters}(S)$ is lower-bounded by
$\distFr{s_1}{\widehat{s_1}}$. 
By the case analysis, it holds that 
\[ \cardin{\brc{\tau \in P: \distFr{\widehat{s_1}}{\tau} \leq r}} < \frac{5}{7} |P|, \] 
for $r=\frac{\distFr{\widehat{s_1}}{c}}{5}$. Therefore, it holds for each of the remaining sample points $s_i$, for $1 < i
\leq |S|$, that
\begin{equation*}
\Pr\pbrc{ \distFr{\widehat{s_1}}{s_i} > r} \geq \frac{2}{7}.
\end{equation*}
In case this event happens, it holds by the triangle inequality
\begin{equation*}
\distFr{s_1}{s_i} 
\geq \distFr{s_i}{\widehat{s_1}}-\distFr{\widehat{s_1}}{s_1}
\geq r-\distFr{\widehat{s_1}}{s_1} 
\geq \frac{\distFr{c}{s_1}-\distFr{s_1}{\widehat{s_1}}}{5}-\distFr{\widehat{s_1}}{s_1}
\geq \frac{\distFr{c}{s_1}}{10}.
\end{equation*}
The analysis of this event is almost the same as in the first case. In this
case, a total number of $\ceil{5 \log(\frac{1}{\lambda})} + 1$ samples suffices to ensure that this
event happens with probability of at least $(1-\lambda)$.
\end{proof}

We are now ready to prove the modified sampling property.
\begin{theorem}\thmlab{modified:sampling:prop}
There exist  integer constants $m_{\eps,\lambda,\lenClusters}$ and
$t_{\eps,\lambda,\lenClusters}$ such that given a set of curves
$P=\brc{\inputTraj{1},\dots,\inputTraj{n}}$  for a uniform
sample multiset $S\subseteq P$ of size $m_{\eps,\lambda,\lenClusters}$
we can generate a candidate set $\Gamma(S) \subset \Delta_{\lenClusters}$ of size $t_{\eps,\lambda,\lenClusters}$ satisfying
\[\Pr\left[ \exists q \in \Gamma(S): \cost_{1}(P,q) \leq (1+\eps)\opt^{(1)}_{1,\lenClusters}(P) \right] \geq 1-\lambda.\] 
Furthermore, we can compute $\Gamma(S)$ in time depending on $\lenClusters,\lambda$ and $\eps$ only.
\end{theorem}

\begin{proof}
Let $\lambda'=\frac{\lambda}{4}$ and $\eps'=\frac{\eps}{4}$.
Let $c$ denote an optimal $(1,\ell)$-median of $P$ and let $c_s$ denote an optimal $(1,\ell)$-median of $S$.  
We use the algorithm described in \secref{cf:approx:uni} to compute a
constant-factor approximation to $\opt^{(1)}_{1,\lenClusters}(S) $ and obtain an interval 
$[\delta_S^{\min}, \delta_S^{\max}]$ which contains
$\opt^{(1)}_{1,\lenClusters}(S)$ and by \thmref{cf:approx:median:main} it holds
that $\delta_S^{\max}/\delta_S^{\min} \leq 65$.
We apply \algref{candidate:generator} to $S$ with parameters
\[ \alpha = \frac{6\delta_S^{\max}}{\eps'} ~~ \text{ and } ~~
\beta = \frac{\eps'\lambda' \delta_S^{min}}{|S|},
\]
and obtain a set $\Gamma^{\lenClusters}_{\alpha,\beta}$.

With the help of \lemref{omit:low:prob}, we can now adapt the proof of 
\lemref{candidate:generator} to our probabilistic setting.
Let $\centerTraj{}=z_1,\dots,z_{\lenClusters}$ be an optimal 
$(1,\ell)$-median of $P$.
%\Amer{it was just "optimal median" in all three occurances above}
Let $\inputTraj{1},\dots,\inputTraj{n}$ denote the input curves in the increasing order of their
distance denoted by $x_i=\distFr{\centerTraj{}}{\inputTraj{i}}$. For every $\inputTraj{i}$,
consider its $x_i$-signature denoted by $\sigma_i$.   
By \lemref{nec:suff}, each vertex of
$\centerTraj{}$ lies within distance $4x_i$ to a vertex of some signature
$\sigma_i$, otherwise we can omit it by \thmref{remove:one}.  Hence, there must
be a $(1,\lenClusters)$-median whose vertex set is contained in the union of
the intervals 
\[ \bigcup_{\inputTraj{i} \in P} \bigcup_{v \in \VtxSet(\sigma_i) }  [v-4x_i,v+4x_i].\] 
Let this solution be denoted $\centerTraj{}$. 
%\Amer{we say above "let $c$ be an optimal median", and now let this median be denoted as $c$. I know what is meant, but is it clear?}

So, consider the following union of intervals: 
\[ R_1 = \bigcup_{\inputTraj{i} \in S} \bigcup_{v \in \VtxSet(\sigma_i) } [v-4x_i,v+4x_i],\]

Let $\widehat{\centerTraj{1}}$ be the median obtained from $\centerTraj{}$ by
omitting all vertices which do not lie in $R_1$. 
\lemref{omit:low:prob} implies 
\begin{align*}
\Pr\left[ \cost_{1}(P,\widehat{\centerTraj{1}}) \leq (1+\eps')  \cost_{1}(P,\centerTraj{})\right]
\geq 1-\lambda',
\end{align*}
if we choose $|S| \geq \ceil{\frac{8\lenClusters}{\eps}
\left(\log\left(\frac{1}{\lambda'}\right)+\log(\lenClusters)\right)}$.

So, now consider the following union of intervals: 
\[ R_2 = \bigcup_{\inputTraj{i} \in P \atop x_i \leq \widehat{x}} \bigcup_{v \in \VtxSet(\sigma_i) } [v-4x_i,v+4x_i],\]
where $\widehat{x}=\frac{2 x_1}{\eps'}$.
Let $\widehat{\centerTraj{2}}$ be the median obtained from $\widehat{\centerTraj{1}}$ by
omitting all vertices which do not lie in $R_2$.
\lemref{diam:S} implies that if $|S| \geq \ceil{5\log(\frac{1}{\lambda'})}+1$, then
it holds with a probability of at least $(1-\lambda')$ that
$x_1 \leq  12 \opt^{{(1)}}_{1,\lenClusters}(S).$
\algref{candidate:generator} sets $\widehat{x} = 4\alpha$, %\Amer{Christian wrote that this is a bit unclear formulation. I don't know how to reformulate it nicely. It is possible that the same problem is in \lemref{candidate:generator}}, 
therefore  
\[\widehat{x} = \frac{24 \delta_S^{\max}}{\eps'} \geq \frac{24
\opt^{(1)}_{1,\lenClusters}(S)}{\eps'}\geq \frac{2x_1}{\eps'}. \]
Thus, we can apply \lemref{omit:far:away} and obtain 
\begin{align*}
\cost_{1}(P,\widehat{\centerTraj{2}}) \leq (1+\eps')  \cost_{1}(P,\widehat{\centerTraj{1}}).
\end{align*}
Therefore, with probability $1-2\lambda'$, the generated set
$\Gamma^{\lenClusters}_{\alpha,\beta}$ contains a curve $\tilde{\centerTraj{}}$
which lies within \Frechet distance $\beta$ of $\widehat{\centerTraj{2}}$. 

\lemref{markov} implies that with probability at least $1-\lambda'$ it holds that 
\[\opt^{(1)}_{1,\lenClusters}(S) \leq \cost_{1}(S,c) \leq \frac{|S|}{\lambda' n}\opt^{(1)}_{1,\lenClusters}(P).\]
Thus, with the same probability it holds that
\begin{align*}
\beta = \frac{\eps'\lambda' \delta_S^{min}}{|S|} 
\leq \frac{\eps'\lambda' \opt^{(1)}_{1,\lenClusters}(S)}{|S|} 
\leq \frac{\eps' \opt^{(1)}_{1,\lenClusters}(P) }{n}.
\end{align*}

We conclude that with probability $1-3\lambda' > 1-\lambda$ (union bound)
there exists a candidate in $\tilde{c} \in \Gamma^{\lenClusters}_{\alpha,\beta}$ such that
\begin{align*}
\cost_{1}(P,\tilde{c}) &\leq \cost_{1}(P,\widehat{c_2}) + \beta n 
\leq (1+\eps')\cost_{1}(P,\widehat{c_1}) + \beta n\\
&\leq (1+\eps')^2 \cost_{1}(P,c) + \beta n
\leq ((1+\eps')^2 + \eps') \opt^{(1)}_{1,\lenClusters}(P) 
\leq (1+\eps) \opt^{(1)}_{1,\lenClusters}(P). 
\end{align*}

Furthermore, by \lemref{candidate:generator} the size of $\Gamma^{\lenClusters}_{\alpha,\beta}$ is bounded as follows
\begin{align*}
t_{\eps',\lambda',\lenClusters} 
\leq \left({\frac{16\alpha |S|(\lenClusters+3)}{\beta}}\right)^{\lenClusters} 
= \pth{c_1 \eps^{-4}\lambda^{-1} \lenClusters^3 \pth{\log^2\pth{\frac{1}{\lambda}}+\log^2 \lenClusters}}^{\lenClusters},
\end{align*}
where $c_1$ is a sufficiently large constant.
\end{proof}

The following theorem follows from Ackermann \etal~\cite{abs-cm-10} (\thmref{modified:sampling:prop}).
For this purpose, recall that the analysis of Ackermann \etal does not require the distance function 
to satisfy the triangle inequality. Therefore we can adopt the $(k,\ell)$-median formulation from 
\secref{problem} which uses the dissimilarity measure $D$ on the set $\Delta_{\ell} \cup P$. 
%\Amer{it was $k$-median, but in the referred section it is $(k,\ell)$}
% after some minor modifications that are required
%to make sure that only centers of complexity at most $\ell$ are used. For this purpose, recall that the distance function
%from Ackermann \etal~\cite{abs-cm-10} is not necessarily symmetric, so we can simply set the distance to an input time series to
%$\infty$ while the distance from an input time series to a time series of complexity $\ell$ is the \Frechet{} distance. These modification are only
%required for the proof; their algorithm remains unchanged. 
To achieve the running time we use Alt and Godau's algorithm~\cite{ag-cfdbt-95} for distance computations.
%Furthermore, we can deploy Indyk's technique~(see~\thmref{1:l:median:main}) during the execution of the algorithm 
%to reduce the number of candidates per sample.

%\begin{theorem}\thmlab{k:l:median:main}
%Let $\varepsilon, \lambda, \ell >0$ be constants.
%Given a set of curves $P=\brc{\inputTraj{1},\dots,\inputTraj{n}} \subset
%\Delta_{m}$, there exists an algorithm that with
%constant probability returns a $(1+\eps)$-approximation to the
%$(\nrClusters,\lenClusters)$-median problem for input instance $P$, and that
%has running time $O(n 2^{O(k)} m \log m)$, where the $O$-notation in the
%exponent hides dependencies on $\eps, \lambda$ and $\ell$.
%%
%%
%%in \[O(n2^{O(\nrClusters
%%m_{\eps/3,\lambda,\ell}\log(\frac{k}{\eps}m_{\eps/3,\lambda,\ell}))} m\lenClusters
%%\log{m\lenClusters}),\]where
%%$m_{\eps,\lambda,\ell}$ is a constant depending only on $\eps,\lambda$ and
%%$\ell$.
%\end{theorem}

\begin{theorem}\thmlab{k:l:median:main}
%\Amer{This is the conference version of the theorem, that does not specify the exponential dependence on $k$. Also, the former version contained $\lambda$, that is nowhere used, as a remnant from the earlier versions.}
Let $\eps, k,\ell >0$ be constants.
Given a set of curves $P=\brc{\inputTraj{1},\dots,\inputTraj{n}} \subset
\Delta_{m}$, there exists an algorithm that with
constant probability returns a $(1+\eps)$-approximation to the
$(\nrClusters,\lenClusters)$-median problem for input instance $P$, and that
has running time $O(n m \log m)$.
%, where the $O$-notation in the
%exponent hides dependencies on $\eps, \lambda$ and $\ell$.
\end{theorem}

\section{Constant-factor approximation in various cases}
\seclab{cf:approx:uni}

It is not difficult to compute a constant-factor approximation for our problem.
We include the details for the sake of completeness. Our algorithm first
simplifies the input curves before applying a known approximate clustering algorithm
designed for general metric spaces.  Note that an approximation scheme which
first applies clustering and then simplification does not yield the same
running time, since the distance computations are expensive.

\begin{algorithm}[h]\alglab{phase1}
 \KwData{curves $P = \brc{\inputTraj{1},\dots,\inputTraj{n}}$, parameters $\nrClusters>0,\lenClusters>0$}
 \KwResult{cluster centers $C= \brc{c_1,\dots,c_{\nrClusters}}$ and cost $\initialCost$}
 \caption{Constant-factor approximation for $(\nrClusters,\lenClusters)$-clustering}
For each $\inputTraj{i}$ we compute an approximate minimum-error
$\ell$-simplification $\inputTrajSimp{i}$ (\lemref{apx:min:error:simp})\;
We apply a known approximation algorithm for clustering in general metric spaces
on $\widehat{P}=\inputTrajSimp{1}, \dots ,\inputTrajSimp{n}$ (i.e.,
Gonzales' algorithm~\cite{Gonzalez1985} for $\nrClusters$-center and
Chen's~algorithm~\cite{c-kmc-09} for $\nrClusters$-median)\; 
We return the resulting cluster centers
$C=\centerTraj{1},\dots,\centerTraj{\nrClusters}$ with approximate cost 
\begin{eqnarray*}
D_{\infty}&=&\cost_{\infty}(C,\widehat{P})+ \max_{1\leq i\leq n}
\distFr{\inputTraj{i}}{\inputTrajSimp{i}} \\
D_1&=&\cost_{1}(C,\widehat{P})+ \sum_{1\leq i\leq n}
\distFr{\inputTraj{i}}{\inputTrajSimp{i}} 
\end{eqnarray*}
for $(\nrClusters,\lenClusters)$-center and $(\nrClusters,\lenClusters)$-median, respectively\;
\end{algorithm}

\begin{lemma}\lemlab{c:f:approx}
The cost $D_{\infty}$ (resp., $D_{1}$) and solution $C$ computed by \algref{phase1} constitute a
$(\alpha+\beta+\alpha\beta)$-approximation to the
$(\nrClusters,\lenClusters)$-center problem (resp., the
$(\nrClusters,\lenClusters)$-median problem), where $\constA$ is the approximation
factor of the simplification step and $\constB$ is the approximation factor of
the clustering step. 
\end{lemma}

\begin{proof}
We first discuss the case of $(\nrClusters,\lenClusters)$-center. The
$(\nrClusters,\lenClusters)$-median will follow with a simple modification.
First, we have that 
\begin{eqnarray*}
\initialCost_{\infty} 
&=& \cost_{\infty}(C,\widehat{P})+ \max_{1\leq i\leq n} \distFr{\inputTraj{i}}{\inputTrajSimp{i}} \\
&=& \max_{1\leq i\leq n}\; \min_{\centerTraj{} \in C}
\distFr{\inputTrajSimp{i}}{\centerTraj{}} + \max_{1\leq i\leq n} \distFr{\inputTraj{i}}{\inputTrajSimp{i}} \\
&\geq& \max_{1\leq i\leq n}\; \min_{\centerTraj{} \in C}
\pth{\distFr{\inputTrajSimp{i}}{\centerTraj{}} +
\distFr{\inputTraj{i}}{\inputTrajSimp{i}}} \\
&\geq& \max_{1\leq i\leq n}\; \min_{\centerTraj{} \in C}
\distFr{\inputTraj{i}}{\centerTraj{}} \\
&\geq& \cost_{\infty}(C,P).
\end{eqnarray*}

Now, let $\costOpt$ be the optimal cost for a solution to the $(\nrClusters,\lenClusters)$-center problem
 for $P=\{\inputTraj{1},\dots,\inputTraj{n}\}$.
It holds that 
\begin{eqnarray*}
\initialCost_{\infty} &\leq& \cost_{\infty}(C,\widehat{P})  + \constA \costOpt, 
\end{eqnarray*}
since $\costOpt$ is lower bounded by the distance of any input time series to its optimal
$\lenClusters$-simplification and this is the minimal \Frechet distance the time series can have to any curve
with at most $\lenClusters$ vertices. 
Now, consider an optimal solution $C^{*}$ with cost $\costOpt$. We can relate it to
$\initialCost_{\infty}$ as follows. In the following, let $p_i \in C^{*}$ be the center of this
optimal solution which is closest to $\inputTraj{i}$.
\begin{eqnarray*}
\costOpt &=& \max_{i=1,\ldots, n}\; \distFr{\inputTraj{i}}{p_i}\\
 &\geq& \max_{i=1,\ldots, n}\; \pth{\distFr{\inputTrajSimp{i}}{p_i} - \distFr{\inputTraj{i}}{\inputTrajSimp{i}}}\\
&\geq& \max_{i=1,\ldots, n}\; \distFr{\inputTrajSimp{i}}{p_i}
- \max_{i=1,\ldots, n} \distFr{\inputTraj{i}}{\inputTrajSimp{i}}\\
&\geq& \cost_{\infty}(C^{*},\widehat{P}) 
- \max_{i=1,\ldots, n} \distFr{\inputTraj{i}}{\inputTrajSimp{i}}\\
&\geq&  \frac{1}{\constB}\pth{ \cost_{\infty}(C,\widehat{P})} - \constA \costOpt\\
&\geq&  \frac{1}{\constB}\pth{ \initialCost_{\infty} - \constA \costOpt} - \constA \costOpt
\end{eqnarray*}
It follows that $\initialCost_{\infty} \leq (\constA+\constB+\constA\constB)\costOpt$. 

\end{proof}

\begin{theorem}\thmlab{cf:approx:center:main}
Given a set of $n$ curves $P=\lbrace\inputTraj{1},\ldots , \inputTraj{n} \rbrace \subseteq
\Delta_m$ and parameters $\nrClusters,\lenClusters \in \Na$,
we can compute an $8$-approximation to $\opt^{(\infty)}_{\nrClusters,\lenClusters}(P)$
and a witness solution in time $O(\nrClusters n m \lenClusters \log(m
\lenClusters))$.
\end{theorem}
\begin{proof}
The theorem follows by \lemref{c:f:approx} and by setting $\alpha=\beta=2$.  We
use \lemref{apx:min:error:simp} to obtain a $2$-approximate simplification for
each curve.  Then, we use Gonzales' algorithm which yields a $2$-approximation
for the simplifications. Each distance computation takes $O(m\lenClusters
\log(m\lenClusters))$ time using Alt and Godau's algorithm~\cite{ag-cfdbt-95}.
\end{proof}

\begin{theorem}\thmlab{cf:approx:median:main}
Given a set of $n$ curves $P=\lbrace \inputTraj{1},\ldots , \inputTraj{n} \rbrace \subseteq
\Delta_m$ and parameters $\nrClusters,\lenClusters \in \Na$,
we can compute a $65$-approximation to $\opt^{(1)}_{\nrClusters,\lenClusters}(P)$
and a witness solution in time $O((nk+k^7\log^5 n)m\ell \log(m\ell))$.
\end{theorem}
\begin{proof}
The theorem follows from \lemref{c:f:approx} and by setting $\alpha=2$ and
$\beta=21$.  We use \lemref{apx:min:error:simp} to obtain a $2$-approximate
simplification for each curve.  Then, we use the algorithm of
Chen~\cite{c-kmc-09} to solve the discrete version of the $k$-median problem on
the simplifications.  Each distance computation takes $O(m\lenClusters
\log(m\lenClusters))$ time using Alt and Godau's algorithm~\cite{ag-cfdbt-95}.
Chen's algorithm yields an $10.5$-approximation for the discrete problem, where
the centers are constrained to lie in $P$.  Since the \Frechet distance
satisfies the triangle inequality, this implies a $21$-approximation for our
problem. Therefore, setting $\beta=21$ yields a correct bound.
\end{proof}

These results can be easily extended to a $(\nrClusters,\lenClusters)$-means variant
of the problem, as well as to multivariate time series, using known
simplification algorithms, such as the algorithm by Abam \etal~\cite{abh-sals-10}. 
%\Amer{This paragraph is omitted in the CV, but here could stay as a comment}

\section{On computing signatures}
\seclab{computing:signatures}

In this section we discuss how to compute signatures efficiently. Our signatures have a unique hierarchical structure as testified by \lemref{canonical:signature}.  Together with the concept of vertex permutations~(\defref{vtx:permutation}) this allows us to construct a data structure, which supports efficient queries for the signature of a given size~(\thmref{compute:signatures}).  If the parameter $\delta$ is given, we can compute a signature in linear time using \algref{low:pass}. Furthermore, we show that our signatures are approximate simplifications in \lemref{apx:min:error:simp}. 
%\Amer{Shortened introductory text from the CV, }

\begin{lemma}\lemlab{canonical:signature}
Given a polygonal curve $\tau:[0,1]\rightarrow\Re$ with vertices in general position, there exists a
series of signatures $\sigma_1,\sigma_2,\dots,\sigma_k$ and corresponding parameters
$0 = \delta_1 < \delta_2 < \dots < \delta_{k+1}$, such that 
\begin{compactenum}[(i)]
\item $\sigma_i$ is a %minimal 
$\delta$-signature of $\tau$ for any $\delta \in [\delta_i,\delta_{i+1})$, 
\item the vertex set of $\sigma_{i+1}$ is a subset of the vertex set of $\sigma_{i}$,
\item $\sigma_k$ is the linear interpolation of $\tau(0)$ and $\tau(1)$.
\end{compactenum}
\end{lemma}

\begin{proof}
We set $\sigma_1=\tau$ and obtain the desired series by a series of edge contractions.  Clearly, $\sigma_1$ is a minimal $\delta$-signature for $\delta=0$. We now conceptually increase the signature parameter $\delta$ until a smaller signature is possible. In general, let $0=t_0 < t_1 < \dots < t_{\ell}=1$ be the series of parameters that defines $\sigma_i$. Let 
\begin{equation}\label{min:edge}
\delta_{i+1}=\min\brc{|\tau(t_1)-\tau(t_2)|,
|\tau(t_{\ell-1})-\tau(t_{\ell})|, \min_{2 \leq j \leq \ell-2}
\frac{|\tau(t_j)-\tau(t_{j+1})|}{2}}. 
\end{equation}
We contract the edge where the minimum is attained to obtain $\sigma_{i+1}$. By the general position assumption, this edge is unique. If the edge is connected to an endpoint, we only remove the interior vertex, otherwise we remove both endpoints of the edge. We now argue that the obtained curve $\sigma_{i+1}$ is a $\delta_{i+1}$-signature. 

Let $\overline{\tau(t_j)\tau(t_{j+1})}$ be the contracted edge and assume for now that $2\leq j \leq \ell-2 $.  We prove the conditions in \defref{signature} in reverse order.  Observe that 
\begin{equation}\label{range}
\tau(t_j),\tau(t_{j+1}) \in \cbrc{\tau(t_{j-1}),\tau(t_{j+2})},
\end{equation}
since otherwise the contracted edge would not minimize the expression in (\ref{min:edge}).  By induction the \emph{range condition} was satisfied for $\sigma_i$ and by the statement in (\ref{range}) it cannot be violated by the edge contraction.

The contracted edge was the shortest interior edge of $\sigma_i$ and by construction we have that
\begin{equation}\label{edge:length}
|\tau(t_j)-\tau(t_{j+1})| = 2\delta_{i+1}
\end{equation}
Therefore, the \emph{minimum-edge-length} condition is
also satisfied for $\sigma_{i+1}$. 

Since $\delta_i<\delta_{i+1}$, we have to prove the \emph{direction-preserving} condition only for the newly established edge $\overline{\tau(t_{j-1})\tau(t_{j+2})}$ of $\tau_{i+1}$. For any $s,s'\in [t_j,t_{j+1}]$ it holds that $|\tau(s)-\tau(s')|\leq 2\delta_{i+1}$. Indeed, by induction, the range condition held true for the contracted edge and by Equation (\ref{edge:length}) its length was $2\delta_{i+1}$. For any $s,s' \in [t_{j-1},t_j]$ the direction-preserving condition holds by induction, and the same holds for $s,s' \in [t_{j+1},t_{j+2}]$. The remaining case is $s,s' \in [t_{j-1},t_{j+2}]$ where the interval $[s,s']$  crosses the boundary of at least one of the edges. In this case, the direction-preserving condition holds by the range property of $\sigma_i$ and by Equation~(\ref{edge:length}).

It remains to prove the \emph{non-degeneracy} condition. Assume for the sake of contradiction that it would not hold, i.e., either that $\tau(t_{j-1})\in \cbrc{\tau(t_{j-2}), \tau(t_{j+2})}$, or that  $\tau(t_{j+2})\in \cbrc{\tau(t_{j-1}), \tau(t_{j+3})}$. Since the two cases are symmetric, we only discuss the first one and the other case will follow by analogy. Then, (\ref{range}) would imply that $\tau(t_{j-1}) \in \cbrc{\tau(t_{j-2}),\tau(j)}$, which contradicts the range property of $\sigma_i$.

So far we proved the conditions of \defref{signature} in the case that an interior edge is being contracted.  Now, assume that $j=1$ and again let the contracted edge be $\overline{\tau(t_j)\tau(t_{j+1})}$ (the case $j=\lenClusters-1$ is analogous). Again, we prove the conditions in reverse order.  By induction, the range condition is satisfied for the first two edges of $\sigma_i$, as well as the non-degeneracy condition.  Since it holds for the length of the second edge that $|\tau(t_2)-\tau(t_3)| > 2\delta_{i+1}$, it must be that $\cbrc{\tau(t_1)-\delta_{i+1}, \tau(t_1)+\delta_{i+1}} \cup \cbrc{\tau(t_1),\tau(t_3)}$ spans the range of values  on $\tau[t_1,t_3]$. Thus, the range condition is implied for $\sigma_{i+1}$. Similarly, $|\tau(t_2)-\tau(t_3)| > 2\delta_{i+1}$ and $|\tau(t_1)-\tau(t_2)|=\delta_{i+1}$ implies the minimum-edge-length condition, i.e. that  $|\tau(t_1)-\tau(t_3)|>\delta$. The arguments for the direction-preserving condition are the same as above for $j>1$. The non-degeneracy condition on the vertex at $t_2$ is not affected by the edge-contraction, since $\tau(t_2)$ stays a minimum (resp. maximum) in $\sigma_{i+1}$ if it was a minimum (resp. maximum) in $\sigma_{i}$. Otherwise, the contracted edge would not minimize the expression in (\ref{min:edge}).

By construction it is clear that the vertex set of the  $\sigma_{i+1}$ is a subset of the vertex set of $\sigma_{i}$ for each $i$, as well as that $\sigma_k$ is the linear interpolation of $\tau(0)$ and $\tau(1)$. This completes the proof of the lemma.  
\end{proof}

\begin{definition}[Canonical vertex permutation]
\deflab{vtx:permutation}
Given a curve $\inputTraj{}:[0,1]\rightarrow\Re$ with $m$ vertices %$=w_1,\dots,w_m$ 
in general position, consider its canonical signatures $\sigma_1,\dots,\sigma_k$ of \lemref{canonical:signature}. We call a permutation of the vertices of $\inputTraj{}$ canonical if for any two vertices $x,y$ of $\tau$ it holds that if $x \notin \VtxSet\pth{\sigma_{i}}$ (the vertex set of $\sigma_{i}$) and $y\in \VtxSet\pth{\sigma_{i}}$, for some $i$, then $x$ appears before $y$ in the permutation. Furthermore, we require that the permutation contains a token separator for every $\sigma_i$, for $1\leq i \leq k$, such that $\sigma_i$ consists of all vertices appearing after the separator.
\end{definition}

\begin{lemma}\lemlab{compute:vtx:permutation}
Given a curve  $\inputTraj{}:[0,1]\rightarrow\Re$ with $m$ vertices
%$=w_1,\dots,w_m$  with $w_i \in \Re$ 
in general position, we can compute a canonical vertex permutation (\defref{vtx:permutation}) in $O(m \log m)$ time and $O(m)$ space.
\end{lemma}
\begin{proof}
Let $w_1,\ldots,w_m$ be the vertices of the curve $\tau$.
%\Amer{I have added this sentence to make it unique.}
The idea is to simulate the series of edge contractions done in the proof of \lemref{canonical:signature}. We build a min-heap from the vertices $w_2,\dots,w_{m-1}$ using certain weights, which will be defined shortly.\footnote{A heap of the edges can be alternatively used.} We then iteratively extract the (one or more) vertices with the current minimum weight from the heap and update the weights of their neighbors along the current signature curves. The extracted vertices are recorded in a list $L$ in the order of their extraction and will form the canonical vertex permutation in the end. Before every iteration we append a token separator to $L$. In this way, all vertices extracted during one iteration are placed between two token separators in $L$. After the last iteration we again append a token separator and at last the two vertices $w_1$ and $w_m$.

More precisely, let $v_1,\dots,v_k$ denote the vertices contained in the heap in the beginning of one particular iteration, sorted in the order in which they appear along the curve $\sigma$. We call the curve
\[ \sigma = w_1,v_1,\dots,v_k,w_m,\] 
the \emph{current signature}. For every vertex we keep a pointer to the heap element which represents its current predecessor and successor along the current signature. We also keep these pointers to the virtual elements $w_1$ and $w_m$, which are not included in the heap.  We define the weight $W(\cdot)$ for every vertex $v_i$ in the heap as follows:
\begin{compactenum}[(i)]
\item if $i=1$, then $W(v_i)=\min\pth{|w_1-v_1|, \frac{|v_1-v_{2}|}{2}}$,
\item if $i=k$, then $W(v_i)=\min\pth{|w_{m}-v_{k}|, \frac{|v_{k}-v_{k-1}|}{2}}$, otherwise
\item $W(v_i)=\min\pth{\frac{|v_i-v_{i-1}|}{2}, \frac{|v_i-v_{i+1}|}{2}}$.
\end{compactenum}
Initially, the current signature equals $\tau$ and initializing these weights takes $O(n)$ time in total. Following the argument in \lemref{canonical:signature}, we need to contract the edge(s) with minimum length (where exceptions hold for the first and last edge). This is captured by the choice of the weights above. Assume for simplicity that the minimum is attained for exactly one edge\footnote{The two other possible cases are as follows.  It may be that multiple edges of the same length are contracted at once. In this case more than two vertices need to be extracted. Furthermore, it may be that only one vertex $v_1$ or $v_k$ is extracted at once. This corresponds to the case that an edge adjacent to $w_1$ or $w_m$ is being contracted. } with endpoints $v_i$ and $v_{i+1}$ for some $i$. In this case, $v_{i}$ and $v_{i+1}$ are the next two elements to be extracted from the heap and their weight must be equal to $\frac{|v_i-v_{i+1}|}{2}$. Using the pointers to $v_{i-1}$ (unless $i=1$) and $v_{i+2}$ (unless $i=k$), we now update the weights of these neighbors and update the pointers such that $v_{i-1}$ (resp., $w_1$) becomes predecessor of $v_{i+2}$ (resp., $w_m$). Computing the new weight of one of these neighboring vertices can be done in $O(1)$ time, updating the weights in the heap takes $O(\log n)$ time per vertex. We can charge every update to the extraction of a neighboring vertex. Since every vertex is extracted at most once, we have $O(n)$ weight updates in total.
\end{proof}
\begin{lemma}\lemlab{lb:vtx:permutation}
Given a curve $\tau:[0,1]\rightarrow\Re$ with $m$ vertices
%$\inputTraj{}=w_1,\dots,w_m$ with $w_i \in \Re$ 
in general position, the problem of computing a canonical vertex permutation (\defref{vtx:permutation}) has time-complexity $\Theta(m\log m)$.
\end{lemma}
\begin{proof}
By \lemref{compute:vtx:permutation}, we can compute this canonical vertex permutation in time $O(m \log m)$.  We show the lower bound via a reduction from the problem of sorting a list of $M=\frac{m-2}{2}$ natural numbers. Let $a_1,\dots,a_{M}$ be the elements of the list in the order in which they appear in the list. We can determine the maximal element $a_{\max}$ in $O(M)$ time. We now construct a curve $\inputTraj{}$ as follows: \[ \inputTraj{} =~ 0,~  x_1,~ x_1-a_1,~ \dots,~ x_i,~ x_i-a_i,~ \dots, x_M, x_M-a_M, x_{M+1},  \] where $x_i=2ia_{\max}$. The constructed curve contains an edge of length $a_i$ for every $a_i$ of our sorting instance. We call these edges \emph{variable edges}. The remaining edges of the $\inputTraj{}$ are called \emph{connector edges}. All connector edges are longer than $a_{\max}$. A canonical vertex permutation of $\inputTraj{}$ would provide us the with variable edges in ascending order of their length. 
\end{proof}

The following lemma testifies that we can query the canonical vertex permutation for a signature of a given size $\ell$. (Note that a canonical signature of size
exactly $\ell$ may not exist.)

\begin{lemma}\lemlab{extract:signature}
Given a canonical vertex permutation (\defref{vtx:permutation}) of a curve $\inputTraj{}$, we can in $O(\ell \log \ell)$ time extract the canonical signature of $\inputTraj{}$ of maximal size $\ell'$ with $\ell'\leq \ell$.
\end{lemma}
\begin{proof}
Let $L'$ denote the suffix of the canonical vertex permutation which contains the last $\ell$ vertices. If there is no token separator at the starting position of $L'$, then we remove the maximal prefix of $L'$ which contains not token separator. In this way, we obtain the vertices of the canonical signature $\sigma$ of maximal size $\ell'$ with $\ell'\leq \ell$. We now sort the vertices in the order of their appearance along $\sigma$ and return the resulting curve.
\end{proof}

The following theorem follows from \lemref{compute:vtx:permutation} and \lemref{extract:signature}. Furthermore, \lemref{lb:vtx:permutation} testifies that the preprocessing time is asymptotically tight.

\begin{theorem}\thmlab{compute:signatures}
Given a curve $\inputTraj{}:[0,1]\rightarrow \Re$ with $m$ vertices in general position, we can construct a data structure in time $O(m \log m)$ and space $O(m)$, such that given a parameter $\ell$ we can extract in time $O(\ell \log \ell)$ a canonical signature of maximal size $\ell'$ with $\ell'\leq \ell$.
\end{theorem}

\begin{algorithm}[]\alglab{low:pass} %was[tb] in brackets
 \KwData{curve $\inputTraj{}=\tau(t_1),\dots,\tau(t_m)$ with $0=t_1<\ldots<t_m=1$, parameter $\delta>0$}
 \KwResult{values $s_1<\ldots<s_{\lenClusters}$, such that $\sigma=\tau(s_1),\tau(s_2),\dots,\tau(s_{\lenClusters})$ is the $\delta$-signature of $\tau$}
 \caption{Computing a $\delta$-signature}
       $j=1$; $a=1$; $s_1=0$ \tcc*[r]{assign first vertex of $\sigma$}
       \lRepeat(\tcc*[f]{scan beginning of time series}){ $\tau(t_j) \notin
       \range{\tau(0)}{\delta}$ or $j \geq m$ }{$j=j+1$\;}
       $b=j$\tcc*[r]{$\tau(t_b)$ is first point outside $\range{\tau(0)}{\delta}$}
       \For(\tcc*[f]{scan remaining time series}){$i=j+1$ to $m$}{
           \eIf{$\tau(t_b) \in \cbrc{\tau(s_a), \tau(t_i)}$}{
                   $b=i$ \tcc*[r]{update furthest point from $\tau(s_a)$ seen so far}
           }{
           \If(\tcc*[f]{gone backwards too far}){$|\tau(t_i)-\tau(t_b)| > 2 \delta $}{
                   $a=a+1$;
                   $s_a=t_b $  \tcc*[r]{append furthest point to $\sigma$}  
                   $b=i$ \tcc*[r]{update furthest point (and change direction)}  
           }}   
       }
       \If{$\tau(t_b) \notin \range{\tau(1)}{\delta}$}{
       $a=a+1$; $s_a=t_b $;  
       }
       $a=a+1$; $s_a=1$;  \tcc*[r]{assign last vertex of $\sigma$}
\end{algorithm}
%\AnneX{I changed the notation in the algorithm, so the proofs have to be updated (TODO), also changed $j>m$ to $j\geq m$  in line 3.}

%\begin{algorithm}[]\alglab{low:pass} %was[tb] in brackets
% \KwData{curve $\inputTraj{}=w_1,\dots,w_m$, parameter $\delta>0$}
% \KwResult{curve $\sigma=v_1,\dots,v_l$ (signature)}
% \caption{Computing a $\delta$-signature}
%       $j=1$; $l=1$; $v_1=w_1$\tcc*[r]{assign first vertex of $\sigma$}
%       \lRepeat(\tcc*[f]{scan beginning of time series}){ $w_j \notin [w_1-\delta,w_1+\delta]$ or $j>m$ }{$j=j+1$\;}
%       $u=w_j$\tcc*[r]{furthest point from $p_l$ seen so far}
%       \For(\tcc*[f]{scan remaining time series}){$i=j+1$ to $m$}{
%           \eIf{$v_l\leq u \leq w_i$ or $v_l\geq u \geq w_i$ }{
%                   $u = w_i$ \tcc*[r]{update furthest point}
%           }{
%           \If(\tcc*[f]{gone backwards too far}){$|w_i-u| > 2 \delta $}{
%                   $l=l+1$; 
%                   $v_l=u$ \tcc*[r]{append furthest point to $\sigma$}  
%                   $u=w_i$ \tcc*[r]{update furthest point (and change direction)}  
%           }}   
%       }
%       \If{$u \notin [w_m-\delta,w_m+\delta]$}{
%       $l=l+1$; $v_l=u$ 
%       }
%       $l=l+1$; $v_l=w_m$ \tcc*[r]{assign last vertex of $\sigma$}
%\end{algorithm}

%\begin{lemma}\lemlab{low:pass:correct}
%Given a curve $\inputTraj{}=w_1,\dots,w_m$ with $w_i \in \Re$ in general position, 
%and given a parameter $\delta>0$, \algref{low:pass} computes a $\delta$-signature 
%$\sigma=v_1,\dots,v_l$ of $\inputTraj{}$ in $O(m)$ time and space.
%\end{lemma} 

\begin{lemma}\lemlab{low:pass:correct}
Given a curve $\inputTraj{}:[0,1]\rightarrow\Re$ with $m$ vertices 
%$\tau(t_1),\dots,\tau(t_m)$ where $0=t_1<\ldots<t_m=1$, with $\tau(t_i) \in \Re$ 
in general position, and given a parameter $\delta>0$, \algref{low:pass} computes a $\delta$-signature $\sigma:[0,1]\rightarrow\Re$ of $\inputTraj{}$ with $\lenClusters$ vertices %$\sigma(s_1),\dots,\sigma(s_l)$ where $0=s_1<\ldots<s_l=1$, 
in $O(m)$ time and space.
\end{lemma} 

\begin{proof}
We prove that \algref{low:pass} produces the values $s_1<\ldots <
s_\lenClusters$ that define a proper $\delta$-signature
$\sigma=\tau(s_1),\tau(s_2),\dots,\tau(s_{\lenClusters})$ of $\tau$ according to
\defref{signature}.  The algorithm operates in three phases: (1) lines 2-4, (2)
lines 5-11, and (3) lines 12-14. In the first phase the algorithm finds the
first vertex $\tau(t_j)$ of $\tau$ which lies outside the interval
$\range{\tau(0)}{\delta}$ and assigns its index to the variable $b$. 

In the trivial case, $\tau$ is entirely contained in the interval
$\range{\tau(0)}{\delta}$. In this case, the second phase is not executed and
the condition in line 12 evaluates to false. The algorithm returns the correct
signature, which has two vertices, $s_1=0$ and $s_2=1$.
Otherwise, $\tau$ must leave the interval $\range{\tau(0)}{\delta}$.
We claim that the following invariants hold at the end of each iteration of the for-loop in phase 2 (lines 5-11):
\begin{compactitem}
\item[(I1)]  $\tau(s_1),\dots,\tau(s_a)$ is a correct prefix of the $\delta$-signature,
\item[(I2)]  for any $x\in [s_a,t_i]$ it holds: 
\begin{compactenum}[(a)] 
\item if $a>1$ then $\tau(x) \in \cbrc{\tau(s_a),\tau(t_b)}$
\item if $a=1$ then $\tau(x) \in [\tau(0)-\delta,\tau(t_b)]$ when
$\tau(t_b)>\tau(0)$\\ (resp., $\tau(x) \in [\tau(t_b),\tau(0)+\delta]$ when
$\tau(t_b)<\tau(0)$).
\end{compactenum}
\item[(I3)]  \begin{compactenum}[(a)]
\item if $a>1$, then $|\tau(s_a)-\tau(t_b)|>2\delta$,
\item if $a=1$ then $|\tau(0)-\tau(t_b)|>\delta$,
\end{compactenum}
\item[(I4)]  if $t_i>t_b$, then for any $x \in [t_b, t_i]:$ $|\tau(t_b)-\tau(x)|\leq 2 \delta$,
\item[(I5)]  the direction-preserving property holds for the subcurve $\tau[s_a,t_b]$,
\end{compactitem}
We prove the invariants by induction on $i$. The base case happens after execution of
line 4, before the first iteration of the for-loop. For ease of notation, we
define $i=b$ for this case. Invariants (I1), (I3) and (I4) hold by construction.
The other two invariants follow immediately from the observation that
$\tau(t_b)$ is the first point outside the interval $\range{\tau(0)}{\delta}$.
%
%$\tau(t_b)$ is the first vertex outside $\left[\tau(0)-\delta, \tau(0)+\delta\right]$ and the signature prefix contains only $\tau(s_1)=\tau(0)$. This implies the validity of the invariant (I2). The invariants (I1) and (I3) follow immediately by construction. Before the iteration it is $i=j$. Therefore it is also $b=i$, implying the invariant (I4). The invariant (I5) holds, otherwise $\tau(t_b)$ would not be the first vertex outside $\left[\tau(0)-\delta, \tau(0)+\delta\right]$.
%

Now we prove the induction step.  One may have the following intuition.
During the execution of the for-loop in lines 5-11, we implicitely maintain a
general direction in which the curve $\tau$ is moving. 
This direction is \emph{upwards} if $\tau(s_{a}) < \tau(s_b)$ and
\emph{downwards} otherwise.
Furthermore, we maintain that $\tau(t_b)$ is the furthest point from $\tau(s_a)$
on the current signature edge (starting at $\tau(s_a)$) in the current general
direction.
Note that a new vertex is appended to the signature prefix only when $\tau$ has
already moved in the opposite direction by a distance %of at least 
greater than $2\delta$.
Only then, we say that the current general direction of the curve has changed.

Consider an arbitrary iteration $i$ of the for-loop. There are three cases,
\begin{compactenum}[(i)]
\item line 7 is executed and $i$ becomes the new $b$\\
(this happens if $\tau$ is moving in the current general direction beyond
$\tau(s_b)$),
\item lines 10 and 11 are executed and a new signature vertex is appended to the
signature prefix\\
(this happens if $\tau$ has changed its general direction), 
\item no assignments are being made\\
(this happens if $\tau$ locally changes direction, but the current general
direction does not change).
\end{compactenum}

For each invariant we consider each of the three cases above.
\begin{compactitem}
\item (I1): If the signature prefix was not changed in the previous iteration
(cases (i) and (iii)), then (I1) simply holds by induction. Otherwise, we argue
that the new signature prefix is correct. By induction,
$\tau(s_1),\dots,\tau(s_{a-1})$ is a correct signature prefix.  
The conditions of \defref{signature} for $\tau(s_1),\dots,\tau(s_{a})$ follow by the induction
hypotheses (I2),(I3), and (I5) in the iteration step $i'<i$, in which the last 
value of $b$ was assigned. In particular, (I2) implies the range condition
and the non-degeneracy, (I3) implies the minimum-edge-length condition, and (I5)
implies the direction-preserving condition.

\item (I2): Assume $a=1$ and $\tau(t_b)>\tau(0)$. Since $a=1$, we cannot be in
case (ii). Furthermore, once we enter the for-loop, the current general
direction is fixed until $a$ is incremented for the first time. 
Therefore, by (I2) we have for $x\in\left[ s_1,
t_{i-1}\right]$ that $\tau(x)\in\left[ \tau(0)-\delta, \tau(t_{b'})\right]$,
where $b'$ holds the value of $b$ before we entered the for-loop in the current
iteration.  Now, in case (i) the claim follows immediately.
In case (iii) it follows from the (false) condition in line 9, that
$\tau(t_i) > \tau(t_b)-2\delta \geq \tau(0)-\delta$, and by the (false)
condition in line 6, that $\tau(t_i)<\tau(t_b)$. 
The case $a=1$ and $\tau(t_b)<\tau(0)$ is analogous.

It remains to prove the invariant for $a>1$.
Assume case (ii) and assume that the general direction changed from
upwards to downwards (the opposite case is analogous). 
Let $a'$ and $b'$ be the values of $a$ and $b$ before the
new assignment in lines 10 and 11. By (I2),
we have $\tau(t_{b'}) \geq \tau(x)$ for any $x \in [t_{b'},t_{i-1}]$. By (I4), we
have $\tau(t_{b'})-2\delta \leq \tau(x)$ for any $x \in [t_{b'},t_{i-1}]$.
By the (true) condition in line 9, we have $\tau(t_i)<\tau(t_{b'})-2\delta$.
Therefore, for any $x \in [t_{b'}, t_i]$, we have
$\tau(x) \in [\tau(t_i),\tau(t_{b'})]$, which implies (I2) after the assignment
in lines 10 and 11.

Now, in case (i) and case (ii), we have by (I2)  for $x\in[s_a,t_{i-1}]$ 
that $\tau(x)\in \cbrc{\tau(s_a),\tau(t_b')}$. The correctness in case (i)
follows immediately. In case (iii), assume $\tau(t_b) > \tau(s_a)$ (the opposite
case is analogous). 
It follows from the (false) condition in line 9 and by (I3), that
$\tau(t_i) > \tau(t_b)-2\delta \geq \tau(s_a)$, and by the (false)
condition in line 6, that $\tau(t_i)<\tau(t_b)$. 

\item (I3) holds in case $a=1$, since the for-loop was 
started after the curve $\tau$ left the interval $\range{\tau(0)}{\delta}$
for the first time and by (I2) $\tau(t_b)$ is furthest
point from $\tau(0)$. In case $a>1$, (I3) also holds, since we append the
parameter $t_b$ to the signature prefix and re-initialize $b$ with $i$ only
after the curve has moved by a distance of at least $2\delta$ (line 9) from
$\tau(t_b)$. The distance is further maintained by (I2).

\item (I4) is clearly satisfied in case (i) and (ii), since $b=i$ is assigned.
In case (iii) it holds by the (false) condition in line 9 that
$|\tau(t_b)-\tau(t_i)|\leq 2\delta$ and for the remaining curve $\tau\left[
t_b, t_{i-1}\right]$ it follows by induction.

\item (I5) holds since we assign a new signature vertex with parameter $s_a=t_b$
as soon as the curve moves by more than $2\delta$ in the opposite direction (case (ii)).  
\end{compactitem}

In phase 3, there are two cases. If the range condition is satisfied for the
last signature edge from $\tau(s_a)$ to $\tau(1)$, the algorithm only appends
the last vertex $\tau(1)$ of the curve $\tau$ to the signature $\sigma$.
Otherwise, the algorithm appends $\tau(t_b)$ and $\tau(1)$ to the signature.  In
both cases, the conditions in \defref{signature} are satisfied also for the last
part. 

If we use a linked-list for the signature, the running time and space
requirements of the algorithm are linear in $m$, since the execution of one
iteration of the for-loop takes constant time and there are at most $m$ such
iterations.
\end{proof}

From \lemref{low:pass:correct} we obtain the following theorem.

\begin{theorem}\thmlab{computing:delta:signature}
Given a curve $\tau:[0,1]\rightarrow\Re$ with $m$ vertices 
%$\inputTraj{}=w_1,\dots,w_m$  with $w_i \in \Re$ 
in general position, and given a parameter $\delta>0$, we can compute a $\delta$-signature of $\inputTraj{}$ in $O(m)$ time and space.
\end{theorem}

%\begin{theorem}\thmlab{computing:delta:signature}
%Given a curve $\inputTraj{}=w_1,\dots,w_m$  with $w_i \in \Re$ in general position, and given a parameter $\delta>0$, we can compute a $\delta$-signature of $\inputTraj{}$ in $O(m)$ time and space.
%\end{theorem}

\begin{lemma}\lemlab{apx:min:error:simp}
Given a curve $\tau:[0,1]\rightarrow\Re$ with $m$ vertices 
%$\inputTraj{}=w_1,\dots,w_m$  with $w_i \in \Re$ 
in general position, and given a parameter $\ell \in \Na$, we can compute in $O(m \log m)$ time a curve $\pi:[0,1]\rightarrow\Re$ of at most $\ell$ vertices, such that $\distFr{\inputTraj{}}{\pi} \leq 2 \distFr{\inputTraj{}}{\pi^*}$, for $\pi^*$ being a minimum-error $\ell$-simplification of $\inputTraj{}$ (\defref{min:error:simp}). 
\end{lemma}

\begin{proof}
Let $\sigma_1,\dots,\sigma_k$ be the signatures of $\tau$ with corresponding parameters $\delta_1,\dots,\delta_k$ as defined in \lemref{canonical:signature}. \lemref{fd:signature} implies that $\distFr{\sigma_i}{\tau} \leq \delta_i$. Consider the signature $\sigma_i$ with the maximal number of $\ell'\leq \ell$ vertices. We claim that 
\[ \frac{\delta_i}{2}\leq \distFr{\pi^{*}}{\tau} \leq \delta_i .\]
The second inequality follows from the definition of $\pi^{*}$ and the fact that $\sigma_i$ consists of at most $\ell$ vertices. To see the first inequality, consider the signature $\sigma_{i-1}=w_1,\dots,w_{h}$ with $h>\ell$. By \lemref{canonical:signature}, for any $\eps>0$, the signature $\sigma_{i-1}$ is a $\delta$-signature of $\tau$ for $\delta=\delta_i-\eps$. By \defref{signature}, it holds for
\[r_j = \left[w_j-\frac{\delta}{2}, w_j+\frac{\delta}{2}\right],\]
that for any $1\leq j \leq h-1: r_j \cap r_{j+1} = \emptyset.$ By the arguments in the proof of \lemref{nec:suff}, it follows that any curve $\pi$ with $\distFr{\pi}{\tau}\leq \frac{\delta}{2}$ needs to consist of at least $h>\ell$ vertices. Therefore, the first inequality follows. We can compute the signature $\sigma_i$ in $O(m \log m)$ time using \thmref{compute:signatures}.
\end{proof}

\section{Hardness of clustering under the \Frechet distance}
\seclab{nphard}

A metric embedding is a function between two metric spaces which preserves the
distances between the elements of the metric space. The embedding is called
isometric if distances are preserved exactly. 
It is known that one can isometrically embed any bounded subset of a
$d$-dimensional vector space equipped with the $\ell_{\infty}$-norm into
$\Delta_{3d}$~\cite{IndMat04}.
This immediately implies \NP-hardness for $\ell \geq 6 $ knowing that 
the clustering problems we consider are \NP-hard under the
$\ell_{\infty}$ distance for $d\geq 2$. In this section we prove that the \NP-hardness
holds from $\ell=2$. This is achieved by preserving $\ell=d$ in the
metric embedding. 

We begin by establishing the basic facts about clustering under the 
$\ell_{\infty}$-norm.
The $k$-center problem under $\ell_{\infty}$ is \NP-hard for $d \geq 2$ as shown
by Feder and Greene \cite{feder1988optimal}.\footnote{Note that this result can
also be obtained from an earlier result by Megiddo
and~Supowit~\cite{megiddo1984geo}. They show that approximating the $k$-center
problem under $\ell_1$ within a factor smaller than $2$ is \NP-hard.}
Even approximating the optimal cost within a factor smaller than $2$ was
shown to be \NP-hard by the same authors.
The $k$-median problem under $\ell_1$ was proven to be \NP-hard  for $d \geq 2$ by Megiddo
and~Supowit~\cite{megiddo1984geo}. The following well-known observation implies
that the $k$-median problem is also \NP-hard under $\ell_{\infty}$ for $d=2$ (and
therefore also for $d \geq 2$). 

\begin{observation}
For any two points $w$ and $v$ in $\mathbb{R}^2$ it holds that 
$\| w-v \|_{\infty} = {\| T(w)-T(v) \|_{1}},$
where $T$ is a rotation by $\frac{\pi}{4}$ followed by a uniform scaling with
$\frac{1}{\sqrt{2}}$.
\end{observation}

We now describe the metric embedding used in the \textbf{NP}-hardness reduction.

\begin{lemma}\lemlab{embedding}
Any metric space $(S,\ell_{\infty})$, where $S \subset \Re^d$ is a bounded set, can be embedded
isometrically into $\Delta_{d}$. Furthermore, if $S$ is discrete, the embedding
and its inverse can be computed in time linear in $|S|$ and~$d$.  \end{lemma}
\begin{proof}
In the following, we view a list of reals $w=\cbrc{v_1,\dots,v_d}$, $v_i \in \Re$, 
from two different perspectives. We either (i) interpret $w$ as an element of $\Re^d$, 
or (ii) interpret $w$ as a curve of $\Delta_d$.
The interpretation we take should be clear from the context. 
 
Let $\delta= \max_{w \in S}  \|w\|_{\infty}$, that is, the maximal norm of an
element of $S$.  Since $S$ is bounded, $\delta$ is well-defined. Note that
$\delta$ also bounds the maximal coordinate value of an element of $S$ and
likewise $-\delta$ bounds the minimal coordinate of an element of $S$.  

We define the
translation vector \[ T= \pth{6\delta,-6\delta,6\delta,-6\delta, \dots }. \] 
Thus, $T$ translates every even coordinate by $6\delta$ and every odd
coordinate by $-6\delta$. 

Let $w,x \in S$ be two elements of the metric space. 
We argue that \[\distFr{T(w)}{ T(x)} = \|w-x\|_{\infty}. \]
Note that by the triangle inequality 
\[\| w-x\|_{\infty} \leq \|w\|_{\infty} + \|x\|_{\infty} \leq 2\delta.\]

By \obsref{segments}, the \Frechet distance between $T(w)$
and $T(x)$ is at most $\| w-x\|_{\infty}$, since we can associate the
$i$th coordinate of $w$ with the $i$th coordinate of $x$ to construct an
eligible matching~$f$.  We claim that the matching $f$ is in fact optimal.
Assume for the sake of contradiction that there exists a matching $g$ which is
``better'' than $f$, that is, $g$ matches each point of $T(w)$ to a
point on $T(x)$ within a distance strictly smaller than $\| w-x\|_{\infty}$.
It must be that $f$ and $g$ are structurally different from each other, in the
sense that their corresponding paths in the free space diagram do not visit the
same cells. This follows from our construction of $T$ which ensures that the
path corresponding to $f$ is optimal among all paths which visit the same cells.

So consider an arbitrary prefix curve $\widehat{w}$ where $f$ and $g$
structurally differ, that is, the image of $\widehat{w}$ under $g$ contains
either fewer or more vertices than the image under $f$.
Assume fewer (otherwise let $\widehat{w}$ be the corresponding suffix curve of
$T(w)$).  By our construction of $f$, the image of $\widehat{w}$ under
$f$ has the same number of vertices as $\widehat{w}$.  Let $\widehat{x}$ denote
the image of $\widehat{w}$ under $g$.   
By our construction of $T$, it holds that the difference between 
any two consecutive coordinate values of $T(w)$ is at least $8\delta$.  
Therefore, the $(2\delta)$-signature of $\widehat{w}$ is equal to $\widehat{w}$.
It follows from \lemref{nec:suff} that $\widehat{x}$ needs to have at least as
many vertices as $\widehat{w}$. However, by our choice of $\widehat{w}$, the subcurve
$\widehat{x}$ has fewer vertices than $\widehat{w}$. A contradiction.
Thus $f$ is optimal and $\distFr{T(w)}{ T(x)} = \| w-x\|_{\infty}$.
\end{proof}

 The \textbf{NP}-hardness reduction takes an instance of the $k$-center (resp.,
$k$-median) problem under $\ell_{\infty}$ in $\Re^d$ and embeds it into $\Delta_d$
using \lemref{embedding}.  If we could solve the $(k,d)$-center (resp.,
$(k,d)$-median) problem described in \secref{problem}, then by
\lemref{embedding}, we could apply the inverse embedding function to the
solution to obtain a solution for the original problem instance. The same holds
for approximate solutions. Note that the embedding given in \lemref{embedding}
works for any point in the convex hull of $S$, therefore also for the centers
(resp., medians) that form the solution. We obtain the following theorems.

\begin{theorem}\thmlab{center:nphard}
The $(\nrClusters,\lenClusters)$-center problem (where $\nrClusters$ is part of the input) is \NP-hard for
$\lenClusters \geq 2$.  Furthermore, the problem is \NP-hard to approximate within a
factor of $2$.
\end{theorem}

\begin{theorem}\thmlab{median:nphard}
The $(\nrClusters,\lenClusters)$-median problem (where $\nrClusters$ is part of the input) 
is \NP-hard for $\lenClusters \geq 2$. 
\end{theorem}

\section{Doubling dimension of the metric space}
\label{section:doublingdimension}

%We first give the standard definitions of balls and the doubling dimension of a metric space. We then 
In this section we show that the unconstrained metric space of univariate time series 
under the \Frechet distance has unbounded doubling dimension~(\lemref{doub:unlimit}).
%We then show that 
Even if we restrict the complexity of the curves 
%in the metric space 
to $\ell\geq 4$, the 
doubling dimension is unbounded~(\lemref{doub:limit}). 
%This claim holds starting from a curve complexity of $4$. 
For lower complexities, one can easily
show that the doubling dimension is bounded. Note that for $\ell=2$ the 
metric space $(\Delta_\ell, d_F)$ equals the metric space $(\Re^2,\ell_{\infty})$. 
Note that the infinity of the doubling dimension is simply caused by the fact that the
metric space is incomplete. However, it remains that standard techniques for
doubling spaces cannot be applied.

\begin{definition}
In any metric space $(X,D)$ a ball of center $p\in X$ and radius $r\in \Re$  is defined as the set
$\{q\in X : D(p,q) \leq r\}$.
\end{definition}

\begin{definition}
The doubling dimension of a metric space is the smallest positive integer $d$
such that every ball of the metric space can be covered by $2^d$ balls of half
the radius. 
\end{definition}

\begin{lemma}
\label{lemma:doub:unlimit}
The doubling dimension of the metric space $(\Delta, d_F)$ is unbounded.
\end{lemma}

\begin{proof}
Assume for the sake of contradiction that there exists a positive integer $d$
which equals the doubling dimension of the metric space $(\Delta, d_F)$.  We
show that such $d$ cannot exist by constructing $2^d+1$ elements of
$\Delta$ which lie in a ball of radius $\frac{1}{8}$ while no two elements can be
covered by a ball of radius $\frac{1}{16}$.  However, by the pidgeon hole principle, at
least one of the smaller balls would have to cover two different curves in the set. 
We construct a set of curves
$P=\tau_{1},\dots,\tau_{2^d+1}$ with $\tau_i=0,i,i-\frac{1}{2},2^d+2$. The set
$P$ is contained in the ball of radius $\frac{1}{8}$ centered at the curve
\[c=0,\frac{7}{8}, \frac{5}{8}, \dots,i-\frac{1}{8},i-\frac{3}{8}, \dots,
2^d+\frac{7}{8}, 2^d+\frac{5}{8}, 2^d+2.\] See \figref{ddex}~(left) for an example.
Any two curves $\tau_i,\tau_j \in P$
have \Frechet distance $\frac{1}{4}$ to each other.  Now, assume that a ball of
radius $\frac{1}{16}$ exists which contains both $\tau_i$ and $\tau_j$. Let 
its center be denoted $c_{ij}$. We can derive a contradiction using the triangle inequality:
\[\frac{1}{4}=\distFr{\tau_i}{\tau_j}\leq \distFr{\tau_i}{c_{ij}} +\distFr{c_{ij}}{\tau_j}=\frac{1}{8}. \]
\end{proof}

\begin{figure}
\centering
\includegraphics{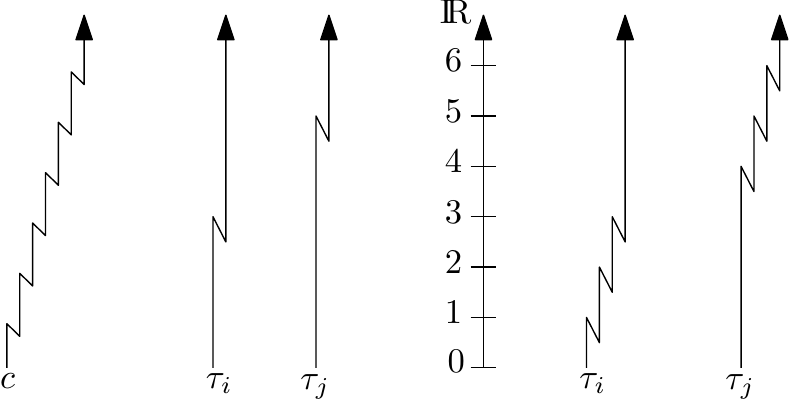}
\caption{Examples of the constructed curves in \lemref{doub:unlimit}~(left) and \lemref{doub:limit}~(right). }
\figlab{ddex}
\end{figure}

\begin{lemma}
\label{lemma:doub:limit}
For any integer $\ell>3$, the doubling dimension of the metric space $(\Delta_{\ell}, d_F)$ is unbounded.
\end{lemma}
\begin{proof}
The proof is similar to the proof of \lemref{doub:unlimit}. However, this time we argue that no two 
curves in the set can be covered by a ball of half the radius because there exists no suitable 
center in $\Delta_{\ell}$, that is, the center would need to have complexity higher than $\ell$.
As in the other proof, we define a set $P=\lbrace \tau_1,\dots,\tau_{2^d+1}\rbrace \subset \Delta_{\ell}$.
For $s=\floor{\frac{\ell-2}{2}}$, let \[ \tau_i = 0, s(i-1) + 1, s(i-1) +
\frac{1}{2}, \dots, s(i-1) + j, s(i-1) + j - \frac{1}{2}, \dots, s i, s i -
\frac{1}{2}, s(2^d+2).\] 
See \figref{ddex}~(right) for an example with $\ell=8$.
Clearly, each $\tau_i \in P$ is an element of $\Delta_{\ell}$, since its complexity is at most $\ell$.
Furthermore, the set $P$ is contained in the ball with radius $\frac{1}{4}$ centered at $c=0,s(2^{d}+2)$.
Note that the $(\frac{1}{8})$-signature of any $\tau_i\in P$ is equal to $\tau_i$ itself.
Thus, by \lemref{nec:suff}, any center of a ball of radius at most $\frac{1}{8}$, which contains $\tau_i$
 needs to have a vertex in each interval $\pbrc{v-\frac{1}{8},v+\frac{1}{8}}$ for each vertex $v$ of $\tau_i$.
By construction, these intervals are pairwise disjoint for each curve and
across all curves in $P$ (except for the intervals around the two endpoints).
Therefore, such a ball with radius at most $\frac{1}{8}$ which would cover two
different curves $\tau_i,\tau_j \in P$, would need to have more than $\ell$
vertices and is therefore not contained in $\Delta_{\ell}$.
Indeed,  the number of pairwise disjoint signature intervals induced by any $\tau_i,\tau_j \in P$ with $i\neq j$,
is $2+4s = 2 \ell - 2 > \ell$.
\end{proof}

\bibliographystyle{abbrv}
      \bibliography{Trajectories}

\end{document}